\documentclass[lettersize,journal]{IEEEtran}
\usepackage{amsmath,amsfonts}
\usepackage{algorithmic}
\usepackage{algorithm}
\usepackage{array}
\usepackage[caption=true,font=normalsize,labelfont=sf,textfont=sf]{subfig}
\usepackage{textcomp}
\usepackage{stfloats}
\usepackage{url}
\usepackage{verbatim}
\usepackage{graphicx}
\usepackage{cite}
\usepackage{amssymb}
\usepackage{mathtools}
\usepackage{amsthm}
\usepackage{multirow} 
\usepackage{dsfont}
\usepackage{booktabs}
\usepackage{color}
\usepackage{hyperref}
\newtheorem{theorem}{Theorem}

\newtheorem{lemma}{Lemma}
\newtheorem{definition}{Definition}
\theoremstyle{definition}
\newtheorem{remark}{Remark}

\newcommand{\E} {\mathbb{E}}

\renewcommand{\P} {\mathbb{P}}

\DeclareMathOperator*{\argmax}{arg\,max}
\DeclareMathOperator*{\argmin}{arg\,min}

\newcommand{\beq}{ \begin{equation} }
	\newcommand{\eeq}{ \end{equation} }

\newcommand{\bsf}{{\boldsymbol f}}
\newcommand{\bsg}{{\boldsymbol g}}

\newcommand{\bso}{{\boldsymbol o}}

\newcommand{\bsu}{{\boldsymbol u}}
\newcommand{\bsv}{{\boldsymbol v}}
\newcommand{\bsw}{{\boldsymbol w}}
\newcommand{\bsx}{{\boldsymbol x}}
\newcommand{\bsy}{{\boldsymbol y}}
\newcommand{\bsz}{{\boldsymbol z}}

\newcommand{\bsI}{{\boldsymbol I}}

\newcommand{\bsL}{{\boldsymbol L}}

\newcommand{\bsmu}{{\boldsymbol \mu}}
\newcommand{\bszeta}{{\boldsymbol \zeta}}
\newcommand{\bssigma}{{\boldsymbol \sigma}}

\onecolumn

\begin{document}

\title{Exact Matching in Correlated Networks with Node Attributes for Improved Community Recovery } 

 \author{%
   \IEEEauthorblockN{Joonhyuk Yang and Hye Won Chung}
   \IEEEauthorblockA{\\School of Electrical Engineering,
                     KAIST\\
                      \{joonhyuk.yang, hwchung\}@kaist.ac.kr}
                   }



\maketitle

\begin{abstract}

We study community detection in multiple networks with jointly correlated node attributes and edges. This setting arises naturally in applications such as social platforms, where a shared set of users may exhibit both correlated friendship patterns and correlated attributes across different platforms. Extending the classical Stochastic Block Model (SBM) and its contextual counterpart (Contextual SBM or CSBM), we introduce the correlated CSBM, which incorporates structural and attribute correlations across graphs. To build intuition, we first analyze correlated Gaussian Mixture Models, wherein only correlated node attributes are available without edges, and identify the conditions under which an estimator minimizing the distance between attributes achieves exact matching of nodes across the two databases. For the correlated CSBMs, we develop a two-step procedure that first applies $k$-core matching to most nodes using edge information, then refines the matching for the remaining unmatched nodes by leveraging their attributes with a distance-based estimator. We identify the conditions under which the algorithm recovers the exact node correspondence, enabling us to merge the correlated edges and average the correlated attributes for enhanced community detection. Crucially, by aligning and combining graphs, we identify regimes in which community detection is impossible in a single graph but becomes feasible when side information from correlated graphs is incorporated. Our results illustrate how the interplay between graph matching and community recovery can boost performance, broadening the scope of multi-graph, attribute-based community detection.

\end{abstract}

\section{Introduction}\label{sec:intro}

Identifying the community labels of nodes in a graph or database--often referred to as community recovery or community detection--is a fundamental problem in network analysis, with broad applications in machine learning, social network analysis, and biology. A central premise of many community detection methods is that nodes within the same community are more likely to be densely connected or to share similar attributes than those in different communities.

To enable rigorous analysis of community detection, various probabilistic models have been developed.  Among these, the Stochastic Block Model (SBM), introduced by Holland, Laskey, and Leinhardt \cite{HLL83}, is one of the most extensively studied. In its classical formulation, the SBM partitions $n$ nodes into $r$ communities, with edges formed independently with probability $p \in [0,1]$ between nodes in the same community, and $q \in [0,p)$ between nodes in different communities. This framework effectively captures community structure in a variety of domains--for instance, in social networks where edges represent interactions or friendships. In the regime where $p = \frac{a \log n}{n}$ and $q = \frac{b \log n}{n}$ for constants $a, b > 0$ and $r = \Theta(1)$, it has been shown in \cite{AS15, ABH15, Abbe17} that exact recovery of community labels is information-theoretically possible if and only if $\sqrt{a} - \sqrt{b} > \sqrt{r}$. 

Since the standard SBM focuses on graph connectivity alone, it neglects potentially informative node attributes. To remedy this, the Contextual Stochastic Block Model (CSBM) includes node attributes alongside structural connectivity. For instance, in \cite{abbe2022}, the authors consider a two-community CSBM in which each node has a Gaussian-distributed attribute vector of dimension $d$, with mean either $\bsmu$ or $-\bsmu$ (depending on the node's community) and covariance $\bsI_d$. This augmented framework leverages both graph structure and node attributes, leading to improved community recovery. Indeed, it is shown in \cite{abbe2022} that the Signal-to-Noise Ratio (SNR), derived from both edges and attributes, governs the feasibility of exact community recovery, thereby demonstrating that combining these two sources of information can outperform methods relying on only one.

While CSBMs address the integration of attributes and structure within a single network, many real-world scenarios naturally involve multiple correlated networks. For example, users often participate in several social platforms (e.g., Facebook and LinkedIn), resulting in correlated patterns of connectivity and behavior. However, privacy and platform-specific anonymization often obscure user identities, making it nontrivial to match corresponding users across networks. This task, known as graph matching, is critical for leveraging multi-graph information in downstream tasks. In fact, exact node correspondence enables the merging of edges from both networks, facilitating more accurate community detection. Prior work \cite{RS21, YC23} has shown that once exact matching is attainable in correlated SBMs, community recovery becomes easier than if only one graph were available. Moreover, Gaudio et al. \cite{GRS22} established precise information-theoretic limits for exact community recovery under correlated SBMs.

Building on these insights, this work addresses the more general setting in which both edges and node attributes are correlated across multiple networks. For instance, a user on both Facebook and LinkedIn may exhibit similar patterns of social interaction and share comparable attribute profiles. We posit that the joint correlation in structure and attributes can significantly enhance community recovery. To formalize this, we extend the CSBM to the correlated multi-network setting, resulting in a new model we refer to as correlated CSBM.

As a preliminary step, we analyze correlated Gaussian Mixture Models (GMMs), which capture correlation in node attributes alone, without considering edges.
Unlike prior studies  \cite{DCK19, KN22} that assume random databases with no underlying community structure, our model incorporates both attribute correlation and latent communities.
By extending the alignment techniques developed in \cite{DCK19, KN22} to settings with hidden communities, we identify conditions under which a minimum-distance estimator (defined in \eqref{eq:estimator}) can reliably recover the underlying permutation $\pi_* : [n] \to [n]$, even when the community memberships are unknown.

For the full correlated CSBM, we propose a two-step algorithm for exact matching. In the first step, we apply the $k$-core matching algorithm \cite{CKNP20}, which uses edge information to match all but $o(n)$ nodes. In the second step, we apply the minimum-distance estimator to align the remaining unmatched nodes using their attribute vectors. We derive precise conditions under which this two-step procedure, combining both structural and contextual information, achieves exact matching.

Once node alignment is established, merging the correlated edges results in a denser graph, while averaging the correlated attributes enhances the effective signal-to-noise ratio. Crucially, we identify regimes in which exact community recovery is impossible using a single graph alone but becomes feasible when information from correlated graphs is incorporated. This demonstrates the pivotal role of graph matching in facilitating community detection, and illustrates how their interplay can significantly improve recovery guarantees.

To the best of our knowledge, this work is the first to study community recovery in correlated graphs that incorporate both edge correlations and correlated node attributes, thereby substantially expanding the landscape of multi-graph and attribute-informed community detection.

\subsection{Models}\label{sec:model}
We introduce two new models to capture correlations across graphs with latent community structure: correlated Gaussian Mixture Models, which model correlations in node attributes alone, and correlated Contextual Stochastic Block Models, which jointly incorporate both node attributes and graph structure.

\subsubsection{\textbf{Correlated Gaussian Mixture Models}}\label{sec:model1} We first assign $d$-dimensional features (or attributes) to $n$ nodes. Let $V_1 := [n]$ denote the set of nodes in the first database, and for each node $i \in V_1$, the node attribute is given by
\begin{equation}\label{eq:node feature g1}
    \bsx_i = \boldsymbol{\mu} \sigma_i +\bsz_i,
\end{equation}
where $\boldsymbol{\mu} \in \mathbb{R}^d$ is a fixed mean vector, $\sigma_i \in \{-1,+1\}$ is the latent community label, and $\bsz_i \sim \mathcal{N}(\mathbf{0},\bsI_d)$ for each $i\in [n]$. 
Let $\bssigma := \{\sigma_i\}_{i=1}^n$ denote the vector of community labels associated with $V_1$.

Next, for the node set $V_2 := [n]$, we assign each node $i$ the attribute
\begin{equation}\label{eq:node feature g2}
    \bsy_i' =\boldsymbol{\mu} \sigma_i +\rho \bsz_i +\sqrt{1-\rho^2}\bsw_i,
\end{equation}
where $\rho \in [0,1]$ is a correlation parameter and $\bsw_i \sim \mathcal{N}(\mathbf{0},\bsI_d)$ is independent Gaussian noise.
Equivalently, the pair $(\bsx_i,\bsy_i')$ can be jointly represented as
\begin{equation}\label{eq:correlated gaussian feature}
    (\bsx_i,\bsy_i')\sim \mathcal{N}\left( \left(\boldsymbol{\mu}\sigma_i,  \boldsymbol{\mu}\sigma_i  \right),\Sigma_{d} \right),
\end{equation}
where
\begin{equation}
    \Sigma_{d}:= \left[
\begin{matrix}
    \bsI_d & \text{diag}(\rho) \\
\text{diag}(\rho) & \bsI_d \\
\end{matrix}
\right].
\end{equation}
We interpret the attributes $\{\bsx_i\}_{i=1}^n$ and $\{\bsy_i'\}_{i=1}^n$ as forming two databases, represented by the matrices  $X:=(\bsx_1, \bsx_2, \ldots ,\bsx_n)^\top \in \mathbb{R}^{n\times d}$ and  $Y':=(\bsy_1', \bsy_2', \ldots ,\bsy_n')^\top \in \mathbb{R}^{n\times d}$. We then define a new database $Y$ by permuting the rows of $Y'$ by $\pi_* : [n] \to [n]$.
The \textit{latent correspondence} between the two databases $X$ and $Y$ is governed by the ground-truth permutation $\pi_*$, assumed to be uniformly distributed over the set $S_n$ of all permutations of $n$ elements. The permuted second database is denoted by $Y=(\bsy_1, \bsy_2, \ldots ,\bsy_n)^\top=(\bsy'_{\pi_*^{-1}(1)},\bsy'_{\pi_*^{-1}(2)},\ldots,\bsy'_{\pi_*^{-1}(n)})^\top \in \mathbb{R}^{n\times d}$. 
Accordingly, the community label vectors for $X$ and $Y$ are given by $\bssigma^1 := \bssigma$ and $\bssigma^2 := \bssigma \circ \pi_*^{-1}$, respectively. We denote the resulting pair of correlated databases as 
\[
(X,\,Y) \;\sim\; \text{CGMMs}\bigl(n,\boldsymbol{\mu},d,\rho\bigr).
\]

\subsubsection{\textbf{Correlated Contextual Stochastic Block Models}}\label{sec:model2}  
Let $V = [n]$ denote the vertex set, and let $\bssigma := \{\sigma_i\}_{i=1}^n$ be the vector of community labels, where each $\sigma_i \in \{-1, +1\}$ is drawn independently and uniformly at random. We first generate a \textit{parent graph} $G \sim \text{SBM}(n, p, q)$, where $p, q \in [0,1]$, with $p > q$ and $q = \Theta(p)$. Specifically, the vertex set is partitioned into
$
V^+ := \{i \in [n] : \sigma_i = +1\}, V^- := \{i \in [n] : \sigma_i = -1\},
$
and edges are formed independently as follows: If $\sigma_u \sigma_v = +1$, an edge $(u,v)$ is placed with probability $p$; if $\sigma_u \sigma_v = -1$, an edge is placed with probability $q$.

Two subgraphs, $G_1$ and $G_2'$, are then independently generated by sampling each edge of $G$ with probability $s \in [0,1]$. This guarantees that both $G_1$ and $G_2'$ are edge-subsampled versions of the same parent graph, and the conditional probability satisfies
\[
\P\left\{(u,v) \in \mathcal{E}(G_2') \mid (u,v) \in \mathcal{E}(G_1) \right\} = s,
\]
for all distinct $u, v \in [n]$, where $\mathcal{E}(G)$ denotes the edge set of graph $G$.

Each node in $G_1$ and $G_2'$ is assigned a correlated Gaussian attribute, as defined in \eqref{eq:correlated gaussian feature}. Specifically, for each $i \in [n]$, the pair $(\bsx_i, \bsy_i')$ is jointly Gaussian, conditioned on $\sigma_i$, with mean depending on $\boldsymbol{\mu}$ and correlation coefficient $\rho \in [0,1]$. The parameter $\boldsymbol{\mu}$ is drawn uniformly from the sphere
\[
\left\{\boldsymbol{\mu} \in \mathbb{R}^d : \|\boldsymbol{\mu}\|^2 = R \right\}
\]
for some fixed $R > 0$. Finally, a hidden permutation $\pi_* : [n] \to [n]$ is applied to the nodes of $G_2'$ to produce the observed graph $G_2$. The community label vectors for $G_1$ and $G_2$ are given by $\bssigma^1 := \bssigma$ and $\bssigma^2 := \bssigma \circ \pi_*^{-1}$, respectively.

Let $X$ and $Y'$ denote the attribute matrices corresponding to $G_1$ and $G_2'$, respectively, and define $Y$ by permuting the rows of $Y'$ according to $\pi_*$. 
Similarly, let $A$, $B'$, and $B$ denote the adjacency matrices of $G_1$, $G_2'$, and $G_2$, respectively. The resulting correlated graph pair is denoted
\[
(G_1, G_2) \sim \text{CCSBMs}(n, p, q, s;\, R, d, \rho).
\]

\subsection{Prior Works}\label{sec:prior work}

      \begin{table*}[t]
    \centering
    \caption{Characteristics of various models: This table summarizes the differences between the previous models and our two new models, depending on the existence of community structure, edges, node attributes, and/or correlated graphs.  }
    \label{tbl:model}
   \footnotesize{    \begin{tabular}{c|c|c|c|c}
    \toprule
    
 Models & Communities & Edges & Node attributes & Correlated graphs    \\ \midrule
 Correlated Erd\H{o}s-R\'enyi graphs    & -& O&-& O \\

 Correlated Gaussian databases      & -&- &O &O \\

  Correlated Gaussian-Attributed Erd\H{o}s-R\'enyi model   & -& O&O&O \\

      Stochastic Block Model       & O& O&-&- \\

 Gaussian Mixture Model       & O& -&O&- \\

 Contextual Stochastic Block Model     & O& O&O&- \\

 Correlated Stochastic Block Models   & O& O&-&O \\
\midrule
 Correlated Gaussian Mixture Models (Ours)  & O& -&O&O \\
 Correlated Contextual Stochastic Block Models (Ours)  & O& O&O&O \\  \bottomrule

    \end{tabular}      }
    \end{table*}
    
      \begin{table*}[t]
    \centering
    \caption{Comparison on the information-theoretic limits for exact matching and exact community recovery across various models: In the conditions for exact community recovery, we consider the cases with two communities and the parameter regimes where $p=\frac{a \log n}{n}$, $q=\frac{b \log n}{n}$, $c\log n=  \frac{R^2}{R + d/n}$ and $c' \log n=\frac{\left(\frac{2}{1+\rho}R\right)^2}{\frac{2}{1+\rho}R+ d/n}$ for constants $a,b,c,c'>0$.  Our exact matching results also require $R \geq 2\log n +\omega(1)$ or $d=\omega(\log n)$,  but these conditions have been omitted for simplicity. }
    \label{tbl:result}
   \footnotesize{    \begin{tabular}{c|c|c|c}
    \toprule
    
& Models & Exact Matching & Exact Community Recovery    \\ \midrule
\cite{WXY22} &Correlated Erd\H{o}s-R\'enyi graphs    & $nps^2 \geq (1+\epsilon) \log n$ & - \\

\cite{DCK19}& Correlated Gaussian databases    &$\frac{d}{4} \log \frac{1}{1-\rho^2} \geq \log n +\omega (1) $ & - \\

\cite{YC24} & Correlated Gaussian-Attributed Erd\H{o}s-R\'enyi model  &  $nps^2+\frac{d}{4} \log \frac{1}{1-\rho^2} \geq (1+\epsilon) \log n$ & - \\ 

\cite{Abbe17} &     Stochastic Block Model      & -& $ s\frac{(\sqrt{a}-\sqrt{b})^2}{2}>1 $  \\ 

\cite{Ndaoud22}& Gaussian Mixture Model      & - & $ c>2 $  \\ 

\cite{abbe2022}& Contextual Stochastic Block Model   & -& $\frac{s(\sqrt{a}-\sqrt{b})^2+c}{2}>1 $  \\ 

\cite{YC23,RS21}& Correlated Stochastic Block Models & $ns^2\frac{p+q}{2} \geq (1+\epsilon) \log n$ & $\left(1-(1-s)^2\right)\frac{(\sqrt{a}-\sqrt{b})^2}{2}>1$  \\
\midrule
\multirow{2}{*}{ Our results} & Correlated Gaussian Mixture Models & $\frac{d}{4}\log \frac{1}{1-\rho^2} \geq (1+\epsilon)\log n$  & $c'>2$   \\
& Correlated Contextual Stochastic Block Models& $ns^2\frac{p+q}{2}+\frac{d}{4}\log \frac{1}{1-\rho^2} \geq (1+\epsilon)\log n$  & $\frac{\left(1-(1-s)^2\right)\left(\sqrt{a}-\sqrt{b}\right)^2+c'}{2} >1$\\  \bottomrule

    \end{tabular}      }
    \end{table*}

Table~\ref{tbl:model} summarizes various graph models, including our newly introduced ones, categorized by their use of community structure, edges, node attributes, and correlation across graphs. Table~\ref{tbl:result} provides an overview of the information-theoretic limits for exact matching in correlated databases and/or graphs, as well as the limits for community recovery in graphs with latent community structure. These results highlight the performance gains in exact community recovery achievable in our proposed models by incorporating correlated edges and/or node attributes.

\subsubsection{Exact Matching}

\paragraph{Matching Correlated Random Graphs}
One of the most extensively studied settings for graph matching is the correlated Erd\H{o}s--R\'enyi (ER) model, first proposed in~\cite{PG11}. In this model, the parent graph $G$ is drawn from $\mathcal{G}(n,p)$ (an ER graph), and $G_1$ and $G'_2$ are obtained by independently sampling every edge of $G$ with probability $s$ twice. Cullina and Kiyavash~\cite{CK16,CK17} provided the first information-theoretic limits for exact matching, showing that exact matching is possible if $nps^2 \ge \log n + \omega(1)$ under the condition $p \le O(1/(\log n)^3)$. More recently, Wu et al.~\cite{WXY22} showed that exact matching remains feasible whenever $p=o(1)$ and $nps^2 \ge (1+\epsilon)\log n$ for any fixed $\epsilon>0$. However, these proofs rely on checking all permutations, yielding time complexity on the order of $\Theta(n!)$. Consequently, a significant effort has focused on more efficient algorithms. Quasi-polynomial time ($n^{O(\log n)}$) approaches were proposed in~\cite{MX20,BCL18+}, while polynomial-time algorithms in~\cite{DMWX21,FMWX23,MRT21} achieve exact matching under $s = 1 - o(1)$. Recently, the first polynomial-time algorithms for constant correlation $s \ge \alpha$ (for a suitable constant $\alpha$) appeared in~\cite{MRT23,MWXY23}, using subgraph counting or large-neighborhood statistics.

Graph matching under correlated Stochastic Block Models (SBMs), where the parent graph is an SBM, has also been investigated~\cite{OGE16,CSKM16}. Assuming known community labels in each graph, Onaran et al.~\cite{OGE16} showed that exact matching is possible when $s(1 - \sqrt{1-s^2}) \frac{p+q}{2} \ge 3 \log n$ for two communities. Cullina et al.~\cite{CSKM16} extended this to $r$ communities, where $p = \frac{a\log n}{n}$ and $q = \frac{b\log n}{n}$, demonstrating that exact matching holds if $s^2 \left(\frac{a + (r-1)b}{r} \right)> 2$. Notably, in the special case $p = q$, correlated SBMs reduce to correlated ER graphs; even under known labels, the bounds in~\cite{OGE16,CSKM16} differ from the information-theoretic limit in the correlated ER setting. R\'acz and Sridhar~\cite{RS21} refined these results for $r=2$, proving that exact matching is possible if $s^2 \left(\frac{a + b}{2} \right)> 1$. Yang and Chung~\cite{YC23} generalized these findings to SBMs with $r$ communities, showing that exact matching holds if $n s^2 \left(\frac{p + (r-1)q}{r}\right) \ge (1 + \epsilon)\log n$, under mild assumptions. As before, achieving this bound requires a time complexity of $\Theta(n!)$. Yang et al.~\cite{YSC23} designed a polynomial-time algorithm under constant correlation when community labels are known, and Chai and R\'acz~\cite{CR24} recently devised a polynomial-time method that obviates label information.

\paragraph{Database Alignment}
Database alignment~\cite{DCK19,DCK20,CMK18,DCK23} addresses the problem of finding a one-to-one correspondence between nodes in two ``databases," where each node is associated with correlated attributes. Similar to graph matching, various models have been proposed, among which the correlated Gaussian database model is popular. In this model, each pair of corresponding nodes $(\bsx_i,\bsy_i')$ is drawn i.i.d.\ from $\mathcal{N}(\boldsymbol{\mu}, \Sigma_d)$, where $\boldsymbol{\mu} \in \mathbb{R}^{2d}$ and 
$
\Sigma_d
=\begin{bmatrix}
\bsI_d & \text{diag}(\rho)\\
\text{diag}(\rho) & \bsI_d
\end{bmatrix}.
$
Dai et al.~\cite{DCK19} showed that exact alignment is possible if $\tfrac{d}{4} \log \tfrac{1}{1-\rho^2}\ge \log n + \omega(1)$. Their method uses the maximum a posteriori (MAP) estimator, with a time complexity of $O(n^2 d + n^3)$.

\paragraph{Attributed graph matching}
In many social networks, users (nodes) have both connections (edges) and personal attributes. The \textit{attributed graph alignment} problem aims to match nodes across two correlated graphs while exploiting both edge structure and node features. In the correlated Gaussian-attributed ER model~\cite{YC24}, the edges come from correlated ER graphs, and node attributes come from correlated Gaussian databases. It was shown that exact matching is possible if 
\beq
n p s^2 + \frac{d}{4} \log\frac{1}{1 - \rho^2} \geq (1 + \epsilon)\log n,
\eeq
indicating that the effective SNR is an additive combination of edge- and attribute-based signals.

Zhang et al.~\cite{ZWW21} introduced an alternative attributed ER pair model with $n$ user nodes and $m$ attribute nodes, assuming that the $m$ attribute nodes are pre-aligned. Edges between user nodes appear with probability $p$, while edges between users and attribute nodes appear with probability $q$. Similar to the correlated ER framework, edges are independently subsampled with probabilities $s_p$ and $s_q$ for user-user and user-attribute edges, respectively, and a permutation $\pi_*$ is applied to yield $G_1$ and $G_2$. Exact matching is possible if $n p s_p^2 + m q s_q^2 \ge \log n + \omega(1)$. Polynomial-time algorithms for recovering $\pi_*$ in this setting have been explored in~\cite{WZWW24}.

\subsubsection{Community Recovery in Correlated Random Graphs}

R\'acz and Sridhar~\cite{RS21} first investigated the information-theoretic limits on exact community recovery in the presence of two or more correlated networks. Focusing on correlated SBMs with $p=\frac{a\log n}{n}$ and $q=\frac{b\log n}{n}$ (for $a,b>0$) and two communities, they established conditions under which exact matching is possible. Once the exact matching is achieved, they construct a union graph $G_1 \vee_{\pi_*} G_2 \sim \text{SBM}\bigl(n,p(1-(1-s)^2),q(1-(1-s)^2)\bigr)$, which is denser than the individual graphs $G_1,G_2\sim \text{SBM}\bigl(n,ps,qs\bigr)$. By doing so, the threshold for exact community recovery becomes less stringent: while a single graph $G_1$ or $G_2$ requires $|\sqrt{a}-\sqrt{b}|>\sqrt{\tfrac{2}{s}}$, the union graph only needs $|\sqrt{a}-\sqrt{b}|>\sqrt{\tfrac{2}{{1-(1-s)^2}}}$, illustrating the existence of a regime in which exact recovery is infeasible with a single graph but feasible with two correlated ones.

Although \cite{RS21} narrowed the gap between achievability and impossibility in exact community recovery for two-community correlated SBMs, Gaudio et al.~\cite{GRS22} completely characterized the information-theoretic limit for exact recovery by leveraging partial matching, even in cases where perfect matching is not possible. Subsequent work has extended these ideas to correlated SBMs with more communities or more than two correlated graphs. Yang and Chung~\cite{YC23} generalized the results of~\cite{RS21} to SBMs with $r$ communities that may scale with $n$, while R\'acz and Zhang~\cite{RZ24} built on \cite{GRS22} to determine the exact information-theoretic threshold for community recovery in scenarios involving more than two correlated SBM graphs.

\subsection{Our Contributions}\label{sec:our work}
This paper introduces and analyzes two new models that account for correlations in both node attributes and graph structure, to better reflect latent community structure in real-world networks. 
Specifically, we focus on \emph{correlated Gaussian Mixture Models (GMMs)} and \emph{correlated Contextual Stochastic Block Models (CSBMs)}, as formally defined in Section~\ref{sec:model}, with the ultimate goal of determining the conditions under which exact community recovery becomes possible by leveraging a second graph (or database) as side information.

The CGMMs and CCSBMs introduced in this work are, to our knowledge, the first statistical models designed to jointly study the alignment (via matching) and integration of two correlated networks with latent community structure, incorporating both correlated edges and node attributes. While the correlated Stochastic Block Models capture edge correlations alone, CGMMs model attribute correlations, and CCSBMs combine both. These models serve as natural generalizations of existing models, such as correlated Erd\H{o}s--R\'{e}nyi graphs and correlated Gaussian-attributed databases, toward a unified framework for understanding how different sources of correlation can be leveraged for improved downstream tasks like community recovery.

A key preliminary step to utilize the correlation between two graphs is to establish \emph{exact matching}--that is, to recover the one-to-one correspondence between the nodes of two correlated graphs--so that the second graph (or database) can serve as informative side information for community detection. We characterize the regimes under which exact matching is either achievable or information-theoretically impossible in each of the proposed models.

In the correlated GMMs setting, we consider a simple estimator that minimizes the total squared distance between node attributes,
\begin{equation}\label{eq:estimator}
    \hat{\pi} := \argmin_{\pi \in S_n} \sum_{i=1}^n \bigl\lVert \bsx_i - \bsy_{\pi(i)} \bigr\rVert^2,
\end{equation}
and we derive sharp thresholds for this estimator to achieve exact alignment. 
In this model, the presence of hidden community labels affects the distribution of pairwise distances between node attributes, depending on whether node pairs belong to the same community or not. This violates the symmetry assumption exploited in \cite{KN22}, where distances were identically distributed across all node pairs. 
To handle this, we introduce two high-probability \emph{good events} (defined in Eqs.~\eqref{eqn:def_A1} and~\eqref{eqn:def_A2}) that control the probability of error caused by incorrect matches, particularly those involving nodes from different communities. Under these events, we carefully analyze the minimum-distance estimator and establish that exact matching is achievable under the main condition $\frac{d}{4} \log \left( \frac{1}{1 - \rho^2} \right) \ge (1 + \epsilon) \log n$, which matches the fundamental limit shown in our converse (Theorem \ref{thm:gmm matching imp}).

For the correlated CSBMs, we develop a two-step matching algorithm. In the first step, we apply \emph{$k$-core matching} using only edge information to align the majority of nodes. In the second step, we use the distance-based attribute estimator in~\eqref{eq:estimator} to match the remaining unmatched nodes. This two-step strategy, originally proposed for correlated Gaussian-attributed Erd\H{o}s--R\'enyi graphs~\cite{YC24}, is shown here to remain effective even in the presence of latent community structure.
In contrast to the model in~\cite{YC24}, where edge structures and attribute data are assumed independent, the CCSBM model introduces a more intricate dependency: both the attribute distance and the edge probability between two nodes depend on whether the nodes belong to the same community. This dependency complicates the analysis of joint correlation across edge and attribute data.

To resolve this challenge, our two-step algorithm first applies $k$-core matching~\cite{CKNP20, GRS22, RS23, YC24}, which identifies the largest matching such that all matched nodes have degree at least $k$ in the intersection graph under a given permutation. This step successfully aligns all but $n^{1 - \frac{n s^2 (p+q)}{2 \log n}}$ nodes using only edge correlation. The remaining unmatched nodes are then matched via a distance-based estimator applied to the attribute data.
The key technical insight is that, even though each step uses only a single type of correlation (edges or attributes), their sequential combination allows us to bypass the dependence between edge structure and attributes while still achieving a \emph{tight} recovery condition:
$
n s^2 \frac{p + q}{2} + \frac{d}{4} \log \left( \frac{1}{1 - \rho^2} \right) \ge (1 + \epsilon) \log n,
$
which matches the information-theoretic limit established by our converse (Theorem~\ref{thm:csbm matching imp}).

Having established the conditions for exact matching, we then investigate \emph{exact community recovery} in these correlated models. When matching is successful, one can merge the two correlated graphs by taking their union, thereby creating a denser graph, and average their correlated Gaussian attributes to reduce variance. Consequently, the achievable range for exact community detection expands relative to the scenario of having only a single graph or database, as illustrated in Figures~\ref{fig:community rec gmm} and~\ref{fig:community recovery csbm}. In particular, for correlated GMMs $(X,Y)\sim \text{CGMMs}\bigl(n,\boldsymbol{\mu},d,\rho\bigr)$ with $\|\boldsymbol{\mu}\|^2 = R$, the effective signal-to-noise ratio for exact recovery increases from 
\[
\frac{R^2}{\,R + \tfrac{d}{n}\,} 
\quad\text{to}\quad 
\frac{\bigl(\tfrac{2}{1+\rho}\,R\bigr)^2}{\,\tfrac{2}{1+\rho}\,R + \tfrac{d}{n}\,}
\]
for the correlation parameter $\rho\in[0,1]$, while for the contextual SBMs $(G_1,G_2)\sim \text{CCSBMs}\bigl(n,p,q,s;R,d,\rho\bigr)$, the corresponding SNR improves from 
\[
\frac{ s\bigl(\sqrt{a}-\sqrt{b}\bigr)^2 + c}{\,2\,}
\quad\text{to}\quad 
\frac{\bigl(1-(1-s)^2\bigr)\bigl(\sqrt{a}-\sqrt{b}\bigr)^2 + c'}{\,2\,},
\]
where $c\log n= \frac{R^2}{\,R + d/n\,}$ and $c' \log n = \frac{\bigl(\tfrac{2}{1+\rho}\,R\bigr)^2}{\,\tfrac{2}{1+\rho}\,R + d/n\,},$ compared to having only one contextual SBM ($G_1$ or $G_2$).

\subsection{Notation}
For a positive integer $n$, write $[n] := \{1, 2, \ldots, n\}$. For a graph $G$ on vertex set $[n]$, let $\deg_G(i)$ be the degree of node $i$, and let $G\{M\}$ be the subgraph induced by $M \subseteq [n]$. Define $d_{\min}(G)$ as the minimum degree of $G$. Let $\mathcal{E} := \{\{i,j\} : i,j \in [n], i \neq j\}$ be the set of all unordered vertex pairs. For a community label vector $\bssigma$, define 
\[
\mathcal{E}^{+}(\bssigma) := \{\{i, j\} \in \mathcal{E} : \sigma_i \sigma_j = +1\}
\quad\text{and}\quad
\mathcal{E}^{-}(\bssigma) := \{\{i, j\} \in \mathcal{E} : \sigma_i \sigma_j = -1\}.
\]
Then, $\mathcal{E}^{+}(\bssigma)$ and $\mathcal{E}^{-}(\bssigma)$ partition $\mathcal{E}$ into intra- and inter-community node pairs. Let $A, B', B$ be the adjacency matrices of $G_1, G'_2, G_2$, respectively, and let $X, Y', Y$ be the corresponding databases of node attributes. Denote by $\vee$ and $\wedge$ the entrywise max and min, respectively. For a permutation $\pi$, define
\[
(A \vee_{\pi} B)_{i,j} =\max\{A_{i,j},\,B_{\pi(i),\pi(j)}\}
\quad\text{and}\quad
(A \wedge_{\pi} B)_{i,j} =\min\{A_{i,j},\,B_{\pi(i),\pi(j)}\}.
\]
For $v = (v_1,\ldots,v_k)^\top \in \mathbb{R}^k$ and $i \in [k]$, let $v_{-i}$ be the vector obtained by removing $v_i$. For an event $E$, let $\mathds{1}(E)$ be its indicator. Write $\Phi(\cdot)$ for the tail distribution of a standard Gaussian. For two functions $f,g : [n] \to [m]$, define their overlap as $\mathbf{ov}(f,g) := \tfrac{1}{n}\sum_{i=1}^n \mathds{1}\bigl(f(i)=g(i)\bigr)$. Lastly, asymptotic notation $O(\cdot), o(\cdot), \Omega(\cdot), \omega(\cdot), \Theta(\cdot)$ is used with $n \to \infty$.
 For $M\subset [n]$ and an injective function $\mu $ : $M \to [n]$, we define $\mu(M)$ as the image of $M$ under $\mu$, and $\mu\{M\}:=\{(i,\mu(i)) : i \in [M]\}$.

\section{Correlated Gaussian Mixture Models}\label{sec:main results}

In this section, we investigate the \emph{correlated Gaussian Mixture Models (GMMs)} introduced in Section~\ref{sec:model1}, with a primary goal of determining conditions for \emph{exact community recovery} when two correlated databases are provided. Our approach consists of two steps: (i) establishing \emph{exact matching} between the two databases, and (ii) merging the matched databases to identify regimes in which exact community recovery becomes significantly more tractable compared to using only a single database.

\subsection{Exact Matching on Correlated Gaussian Mixture Models}

We begin by examining the requirements for exact matching in correlated GMMs. Theorem~\ref{thm:gmm matching ach} below provides sufficient conditions under which an estimator \eqref{eq:estimator} achieves perfect alignment with high probability.

\begin{theorem}[Achievability for Exact Matching]\label{thm:gmm matching ach}
Let $(X,Y)\sim \textnormal{CGMMs}(n,\boldsymbol{\mu},d,\rho)$ as defined in Section~\ref{sec:model1}. Suppose that either 
\begin{equation}\label{eq:gmm matching cond1}
    \frac{d}{4}\log \frac{1}{1-\rho^2} \geq\log n +\omega(1)
    \quad\text{and}\quad
    \|\boldsymbol{\mu}\|^2 \geq 2\log n +\omega(1),
\end{equation}
or
\begin{equation}\label{eq:gmm matching cond2}
    \frac{d}{4}\log \frac{1}{1-\rho^2} \ge (1+\epsilon)\log n
    \quad\text{and}\quad
    d =\omega(\log n)
\end{equation}
holds for an arbitrarily small constant $\epsilon>0$. Then, there exists an estimator $\hat{\pi}(X,Y)$ such that $\hat{\pi}(X,Y) = \pi_*$ with high probability.
\end{theorem}

In correlated GMMs, the MAP estimator can be written as
\begin{equation}\label{eq:MAP cgmms}
\begin{aligned}
    \hat{\pi} &= \argmax_{\pi \in S_n} \P( \pi_*=\pi | X,Y)\\
    & \stackrel{(a)}{=} \argmax_{\pi \in S_n} \P(X,Y|\pi_*=\pi)\\
    & =\argmax_{\pi \in S_n} \sum^n_{i=1}\left( -\|\bsx_i -\boldsymbol{\mu} \sigma_i\|^2 + 2\rho \langle \bsx_i-\boldsymbol{\mu} \sigma_i,\bsy_{\pi(i)}-\boldsymbol{\mu} \sigma_{\pi(i)} \rangle -\| \bsy_{\pi(i)}-\boldsymbol{\mu} \sigma_{\pi(i)}\|^2\right)\\
    & =\argmax_{\pi \in S_n}    \sum^n_{i=1} \left(-\rho\|\bsx_i-\bsy_{\pi(i)} \|^2 -(1-\rho) (\|\bsx_i\|^2+\| \bsy_{\pi(i)}\|^2 ) \right) + f(\bssigma,\bssigma_\pi)\\
    & = \argmin_{\pi \in S_n} \sum^n_{i=1}  \|\bsx_i-\bsy_{\pi(i)} \|^2 - \frac{1}{\rho} f(\bssigma,\bssigma_\pi),
\end{aligned}
\end{equation}
where 
\[
f\bigl(\bssigma,\bssigma_\pi\bigr)
=\sum_{i=1}^n\Bigl(
    2\langle \bsx_i,\boldsymbol{\mu}\sigma_i\rangle
    + 2\langle \bsy_{\pi(i)},\boldsymbol{\mu}\sigma_{\pi(i)}\rangle
    - 2\rho\langle \bsx_i,\boldsymbol{\mu}\sigma_{\pi(i)}\rangle
    - 2\rho\langle \bsy_{\pi(i)},\boldsymbol{\mu}\sigma_i\rangle
    + 2\rho\|\boldsymbol{\mu}\|^2\sigma_i\sigma_{\pi(i)}
\Bigr).
\]
Step~$(a)$ uses the fact that $\pi_*$ is uniformly distributed over $S_n$. Note also that 
$
\sum_{i=1}^n \|\bsx_i\|^2 + \|\bsy_{\pi(i)}\|^2 
$
is invariant under permutation $\pi$, ensuring this part does not affect the alignment decision. If the community labels $\bssigma$ and $\bssigma_\pi$ are known, then $f\bigl(\bssigma,\bssigma_\pi\bigr)$ is also fixed, so the MAP estimator reduces exactly to the simpler distance-based estimator in~\eqref{eq:estimator}. Even without known labels, the proof of Theorem~\ref{thm:gmm matching ach} shows that \eqref{eq:estimator} attains tight bounds for exact matching under the conditions $\|\boldsymbol{\mu}\|^2 \ge 2\log n + \omega(1)$ or $d = \omega(\log n)$; see Section~\ref{sec:proof of gmm matching ach} for details.

\begin{remark}[Interpretation of Conditions for Exact Matching]\label{rem: Int_cond_Thm1}
Assume $d = o(n\log n)$. If $\|\boldsymbol{\mu}\|^2 \ge (2+\epsilon)\log n$ for some $\epsilon>0$, then the exact community recovery is already possible in each individual database $X$ or $Y$ \cite{Ndaoud22}. In this case, we only need to align nodes within the same community. Consequently, the distance-based estimator \eqref{eq:estimator}, which coincides with the MAP rule under known labels, matches nodes optimally. This explains the sufficiency of condition~\eqref{eq:gmm matching cond1} for exact alignment when communities can first be recovered.

On the other hand, if $\|\boldsymbol{\mu}\|^2 = O(\log n)$ is not large enough to recover community labels, exact matching can still be achieved by taking the condition \eqref{eq:gmm matching cond2}. Below, we provide a high-level explanation of this condition. Suppose $d =\omega(\log n)$. Then, the distance $ \| \bsx_i-\bsy_i' \|^2$, which originally follows a scaled chi-squared distribution $\chi^2(d)$ with scale $2(1-\rho)$, can be approximated by a normal distribution,
   \begin{equation}\label{eq:correct dis}
        \| \bsx_i-\bsy_i' \|^2 \stackrel{(d)}{\approx}   2(1-\rho) \mathcal{N}(d,2d).
   \end{equation}
    Similarly, for different nodes $i\neq j$ within the same community, the distance $\| \bsx_i-\bsy_j' \|^2$, which also follows a chi-squared distribution $\chi^2(d)$ scaled by $2$,  can be approximated by a normal distribution, 
    \begin{equation}\label{eq:same dis}
        \| \bsx_i-\bsy_j' \|^2 \stackrel{(d)}{\approx}  2 \mathcal{N}(d,2d) \; \text{ for } i\neq j \text{ and }\sigma_i=\sigma_j.
    \end{equation}
    For nodes in different communities, the distance follows a non-central chi-squared distribution $\chi^2(d,2\|\bsmu\|^2)$, scaled by 2, which can be approximated as
    \begin{equation}\label{eq:diff dis}
        \| \bsx_i-\bsy_k' \|^2 \stackrel{(d)}{\approx}  2\mathcal{N}(2\| \boldsymbol{\mu}\|^2+d,8\| \boldsymbol{\mu}\|^2+2d)\; \text{ for } i\neq k \text{ and }\sigma_i\neq\sigma_k.
    \end{equation}
    Now, consider the regime where $d \gg \|\boldsymbol{\mu}\|^2$, which holds if $d = \omega(\log n)$ and $\|\boldsymbol{\mu}\|^2 = O(\log n)$. In this case, \eqref{eq:diff dis} and \eqref{eq:same dis} are approximately equal, and all distances involving mismatched nodes, whether in the same or different communities, are roughly distributed as $2 \cdot \mathcal{N}(d, 2d)$. However, the correctly matched pair $(\bsx_i, \bsy_i')$ still has a smaller mean due to correlation.
    
Using this approximation, we analyze the probability that node $i$ is not matched to its true counterpart:
\begin{equation}
\begin{aligned}
    &\P\left( \|\bsx_i - \bsy_i'\|^2 \ge \|\bsx_i - \bsy_j'\|^2 \text{ for some } j \ne i \right) \\
    &\stackrel{(a)}{\le} n \cdot \P\left( \mathcal{N}((1 - \rho) d,\, 2(1 - \rho)^2 d) \ge \mathcal{N}(d, 2d) \right) \\
    &\le n \cdot \P\left( \mathcal{N}(0,\, 2d(1 + (1 - \rho)^2)) \ge \rho d \right) \\
    &\stackrel{(b)}{\le} \exp\left( \log n - \frac{1}{4} d \rho^2 \cdot \frac{1}{1 + (1 - \rho)^2} \right),
\end{aligned}
\end{equation}
where inequality (a) follows from a union bound and approximating \eqref{eq:diff dis} by \eqref{eq:same dis}, and inequality (b) follows from the Gaussian tail bound (Lemma~\ref{lem:normal tail}).
Therefore, if
\[
\frac{1}{4} d \rho^2 \ge (1 + (1 - \rho)^2) \log n + \omega(1),
\]
then exact matching becomes possible even with a greedy algorithm that matches each point to its nearest neighbor.
In the special case where $\rho = o(1)$, this reduces to the condition $\frac{1}{4} d \rho^2 \ge 2 \log n + \omega(1)$. In comparison, the sufficient condition in \eqref{eq:gmm matching cond2} is approximately
\[
\frac{d}{4} \log \left( \frac{1}{1 - \rho^2} \right) \approx \frac{1}{4} d \rho^2 \ge (1 + \epsilon) \log n.
\]
The tighter bound in \eqref{eq:gmm matching cond2} is obtained through a refined analysis of the distance-based estimator \eqref{eq:estimator}, which minimizes the total squared distance across all pairs, rather than greedily matching each node to its nearest neighbor.
\end{remark}

We now present the \emph{impossibility} result for exact matching in correlated Gaussian Mixture Models (CGMMs).

\begin{theorem}[Impossibility for Exact Matching]\label{thm:gmm matching imp}
Let $(X,Y)\sim \textnormal{CGMMs}(n,\boldsymbol{\mu},d,\rho)$ be as defined in Section~\ref{sec:model1}. Suppose that either of the following conditions holds:
\begin{equation}\label{eqn:thm2_imp_1}
    \frac{d}{4}\log \frac{1}{1-\rho^2} \leq (1-\epsilon)\log n
    \quad\text{and}\quad
    1\ll d = O(\log n),
\end{equation}
for any small $\epsilon>0$, or
\begin{equation}\label{eqn:thm2_imp_2}
    \frac{d}{4}\log \frac{1}{1-\rho^2}
    \leq \log n - \log d + C
    \quad\text{and}\quad
    \frac{1}{\rho^2} - 1 \leq \frac{d}{40}
\end{equation}
for some positive constant $C>0$. Then, under either set of conditions, no estimator can achieve exact matching with high probability.
\end{theorem}
To prove the converse in Theorem~\ref{thm:gmm matching imp}, we show that the stated conditions imply failure of exact matching \emph{even if} the community labels $\bssigma^1$ and $\bssigma^2$ are known, and a partial matching for all but $|V'|$ nodes is revealed (where $V':=\{\,i\in[n]: \sigma_i=+1\}$). This reduces the task to a database alignment problem with $|V'|$ nodes in correlated Gaussian databases, for which previous impossibility arguments in~\cite{DCK19,Wang22} directly apply. A complete proof is given in Section~\ref{sec:proof of gmm matching imp}.
\begin{remark}[Gaps in Achievability and Converse Results]
Comparing Theorems~\ref{thm:gmm matching ach} and~\ref{thm:gmm matching imp} shows that the limiting condition for exact matching is approximately 
\beq
\frac{d}{4}\log \frac{1}{1-\rho^2} \geq (1+\epsilon)\log n.
\eeq
However, to achieve exact matching via the distance-based estimator, we additionally require either $\|\boldsymbol{\mu}\|^2 \ge 2 \log n + \omega(1)$ or $d = \omega(\log n)$. These conditions ensure that the signal strength or attribute dimensionality is sufficiently large to enable reliable matching in high-dimensional and noisy settings.
This gap between achievability and converse can be attributed to the fact that the estimator~\eqref{eq:estimator} omits the term $f(\bssigma, \bssigma_\pi)$ in the full MAP formulation~\eqref{eq:MAP cgmms}. While this simplification permits matching without knowledge of the latent community labels $\bssigma$ and $\bssigma_\pi$, it necessitates stronger conditions than the optimal MAP-based strategy.

We emphasize that these conditions are not difficult to satisfy in practice, even in large-scale real-world networks. For instance, in platforms like Facebook, the number of effective users is on the order of $10^9$, implying $\log 10^9\approx 20$. Thus, the required attribute signal strength or feature dimension only needs to be on the order of a few dozen. Therefore, while there is a gap between achievability and converse due to the growth condition on $\|\boldsymbol{\mu}\|^2$ or $d$, this requirement is modest and well aligned with the scale and complexity of modern user data.

\end{remark}

\subsection{Exact Community Recovery in Correlated Gaussian Mixture Models}
\label{sec:gmm:exact_recovery}

We now identify the conditions under which \emph{exact community recovery} becomes feasible in \emph{correlated Gaussian Mixture Models} when two correlated node-attribute databases are available.

\begin{theorem}[Achievability for Exact Community Recovery]\label{thm:gmm recovery ach}
  Let $(X,Y)\sim \textnormal{CGMMs}(n,\boldsymbol{\mu},d,\rho)$ be as defined in Section~\ref{sec:model1}. Suppose either \eqref{eq:gmm matching cond1} or \eqref{eq:gmm matching cond2} holds. If
  \begin{equation}\label{eq:gmm recovery ach}
      \|\boldsymbol{\mu}\|^2 
      \geq 
      (1+\epsilon)\frac{1+\rho}{2}\biggl(1+\sqrt{1+\frac{2d}{n \log n}}\biggr)\log n
  \end{equation}
  for an arbitrarily small constant $\epsilon>0$, then there exists an estimator $\hat{\bssigma}(X,Y)$ such that $\hat{\bssigma}(X,Y) = \bssigma$ with high probability.
\end{theorem}

If exact matching is possible, or if $\pi_*$ is given, then the two correlated node attributes can be merged by taking their average. Specifically, if node $i \in [n]$ has attributes $\bsx_i$ and $\bsy_{\pi_*(i)}$, we form  
\begin{equation}\label{eqn:avg_feature}
    \frac{\bsx_i + \bsy_{\pi_*(i)}}{2}
    =\boldsymbol{\mu}\sigma_i + \frac{(1+\rho)\bsz_i + \sqrt{1-\rho^2}\bsw_i}{2},
\end{equation}
where $\bsx_i$ and $\bsy_{\pi_*(i)}=\bsy_i'$ are defined in \eqref{eq:node feature g1} and \eqref{eq:node feature g2}, respectively. Since $\bsz_i$ and $\bsw_i$ are independent $d$-dimensional Gaussian vectors, 
$
\frac{(1+\rho)\bsz_i + \sqrt{1-\rho^2}\bsw_i}{2}
$
is a $d$-dimensional Gaussian with mean $0$ and covariance $\frac{1+\rho}{2}\,\bsI_d$. Therefore, the averaged attribute has the same mean as the individual node attributes but a smaller variance. By applying the results for exact community recovery in a single GMM~\cite{Ndaoud22,abbe2022} to this averaged feature, we establish Theorem~\ref{thm:gmm recovery ach}. A detailed proof is provided in Section~\ref{sec:proof of gmm recovery ach}.

\begin{remark}[Comparison with a Single GMM Result]
  When only $X$ is available, \cite{Ndaoud22,abbe2022} showed that exact community recovery in a Gaussian Mixture Model requires
   \begin{equation}\label{eq:single gmm}
       \lVert \boldsymbol{\mu} \rVert^2 \geq (1+\epsilon)\left(1+\sqrt{1+\frac{2d}{n \log n}}\right) \log n.
   \end{equation}
  Comparing \eqref{eq:gmm recovery ach} with \eqref{eq:single gmm} reveals that $\|\boldsymbol{\mu}\|^2$ can be \emph{smaller} by a factor of $\tfrac{1+\rho}{2}$ when the correlated database $Y$ is available and we can form the combined feature~\eqref{eqn:avg_feature} via exact matching. From \eqref{eq:gmm matching cond2}, such a gain is particularly relevant in the regime $d=\omega(\log n)$. If $d=O(\log n)$, then to satisfy \eqref{eq:gmm matching cond1} we need $\|\boldsymbol{\mu}\|^2 \geq 2\log n +\omega(1)$ for exact matching, which can exceed the threshold \eqref{eq:gmm recovery ach} for community recovery. Hence, in low-dimensional regimes, the advantage of incorporating the second database $Y$ is constrained by the stricter requirement on $\|\boldsymbol{\mu}\|^2$ for alignment.
\end{remark}

We now provide the converse result, stating when exact community recovery is \emph{not} possible in correlated GMMs.

\begin{theorem}[Impossibility for Exact Community Recovery]\label{thm:gmm recovery imp}
  Let $(X,Y) \sim \textnormal{CGMMs}(n,\boldsymbol{\mu},d,\rho)$ as defined in Section~\ref{sec:model1}. Suppose
\begin{equation}
        \lVert \boldsymbol{\mu} \rVert^2 \leq (1-\epsilon)\frac{1+\rho}{2}\left(1+\sqrt{1+\frac{2d}{n \log n}}\right) \log n
    \end{equation}
  for an arbitrarily small constant $\epsilon>0$. Then, there is no estimator that achieves exact community recovery with high probability.
\end{theorem}

\noindent
The proof of Theorem~\ref{thm:gmm recovery imp} follows the same reasoning as in Theorem~\ref{thm:csbm recovery imp} for the special case $p,q=0$ (without edges), thus it can be viewed as a corollary of that result.

\begin{remark}[Information-Theoretic Gaps in Exact Community Recovery]
  Theorem~\ref{thm:gmm recovery ach} assumes that either \eqref{eq:gmm matching cond1} or \eqref{eq:gmm matching cond2} is satisfied to ensure exact matching. This raises a natural question: \emph{if} matching is \emph{not} feasible (i.e., $\tfrac{d}{4}\log\tfrac{1}{1-\rho^2} < \log n$), must exact community recovery also remain impossible, even if \eqref{eq:gmm recovery ach} holds? Addressing this question when exact matching is precluded is an interesting open problem, discussed further in Section~\ref{sec:discu and open}.
\end{remark}

\begin{figure}[!htb]    
  \centering
  \includegraphics[scale=0.3]{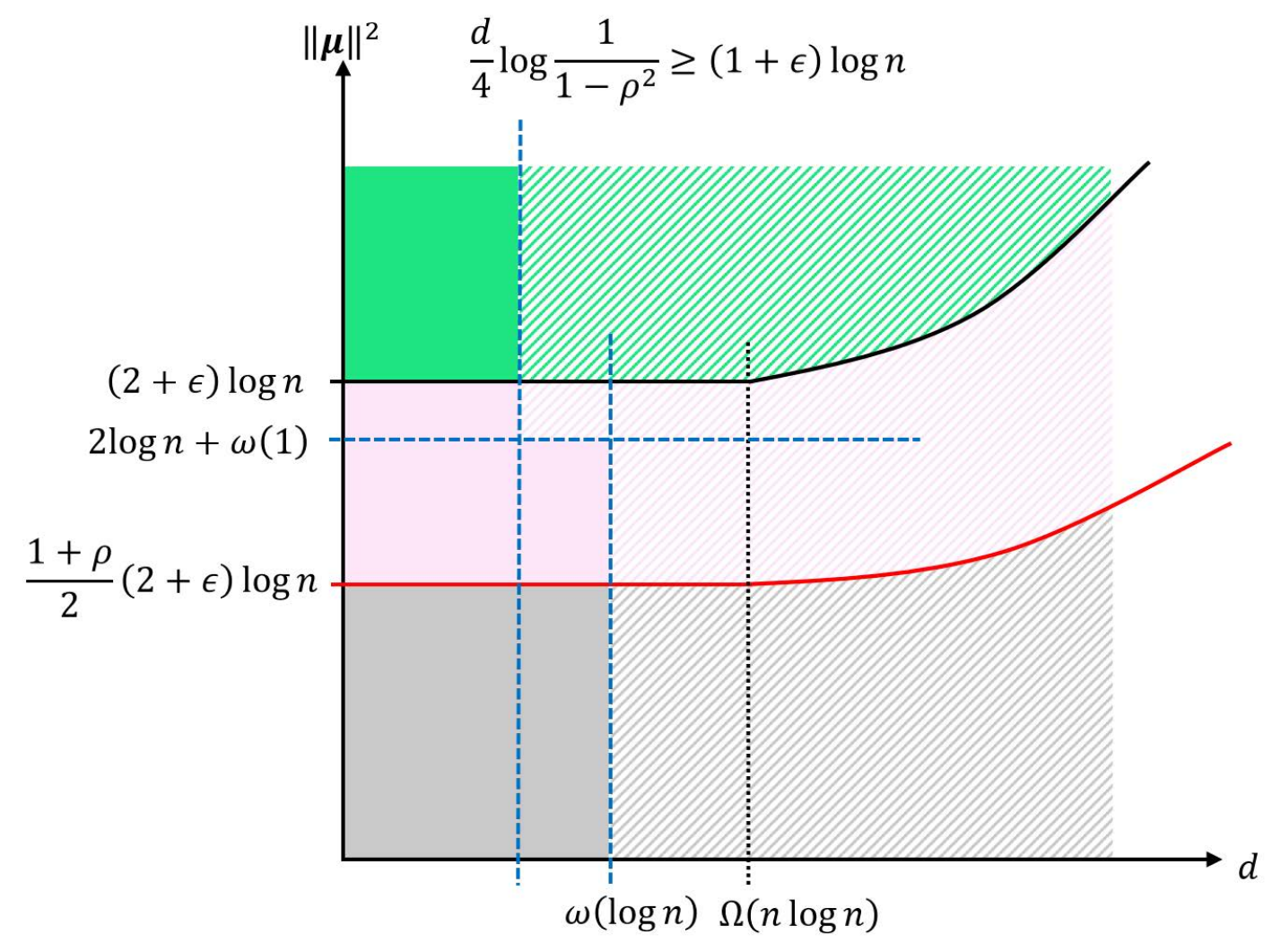} 
  \caption{%
    Conditions for exact community recovery (\emph{solid lines}) and exact matching (\emph{dashed lines}) in Gaussian Mixture Models (GMMs). 
    The \textbf{green region} indicates where exact recovery is possible with a single GMM $X$, whereas the \textbf{gray region} shows where recovery is not possible even with correlated GMMs $(X,Y)$. 
    The \textbf{pink + green region} (above the red curve) highlights where two correlated attributes $X$ and $Y$ enable exact recovery. 
    The \textbf{blue dashed line} corresponds to conditions \eqref{eq:gmm matching cond1} and \eqref{eq:gmm matching cond2} for exact matching, with \textbf{striped areas} indicating parameter regimes where matching is feasible. 
    Notably, the \textbf{striped pink region} represents cases where $X$ alone is insufficient for community recovery, but $(X,Y)$ together can achieve it, demonstrating the benefit of using correlated attributes.
  } \label{fig:community rec gmm}
\end{figure}

Figure \ref{fig:community rec gmm} illustrates the threshold conditions for both exact matching and exact community recovery in correlated GMMs. It visually highlights how the addition of correlated node attributes expands the parameter regime where community detection becomes feasible.

\section{Correlated Contextual Stochastic Block Models}
\label{sec:ccsbms}

In this section, we consider the \emph{correlated Contextual Stochastic Block Models (CSBMs)} introduced in Section~\ref{sec:model2}. Similar to Section~\ref{sec:main results}, we first establish the conditions for exact matching and then derive the conditions for exact community recovery, assuming we can perfectly match the nodes between the two correlated graphs.

\subsection{Exact Matching}\label{sec:CCSBM_exact_matching}

The following theorem provides sufficient conditions for exact matching in correlated CSBMs.

\begin{theorem}[Achievability for Exact Matching]\label{thm:csbm matching ach}
Let $(G_1,G_2) \sim \textnormal{CCSBMs}(n,p,q,s;R,d,\rho)$ be as defined in Section~\ref{sec:model2}, and suppose
\begin{equation}
\label{eq:csbm matching cond0}
    p\leq O\left(\frac{1}{e^{(\log \log n)^3}}\right),
\end{equation}
as well as
 \begin{equation}\label{eq:csbm matching cond2}
         R \geq 2\log n+\omega(1) \text{ or }  d =\omega (\log n)
     \end{equation}
     holds. If
 \begin{equation}\label{eq:csbm matching cond1}
        ns^2\frac{p+q}{2}+\frac{d}{4} \log \frac{1}{1-\rho^2} \geq (1+\epsilon) \log n
    \end{equation}
for an arbitrarily small constant $\epsilon>0$, then there exists an estimator $\hat{\pi}(G_1,G_2)$ such that $\hat{\pi}(G_1,G_2) = \pi_*$ with high probability.
\end{theorem}

In \cite{RS21,RS23,YC23}, it was shown that for correlated SBMs, exact matching is possible if the average degree of the intersection graph $ns^2\left(\frac{p + q}{2}\right)$ exceeds $\log n$. Hence, we focus on the regime $ns^2 \tfrac{p + q}{2} = O(\log n)$, where we employ a two-step approach for exact matching. First, we use \emph{$k$-core matching} to align a large fraction of the nodes via edge information alone; the $k$-core approach has been extensively studied in \cite{CKNP20,GRS22,RS23} and is adapted here to handle more general regimes of $p$ and $q$. By choosing an appropriate $k = \frac{\log n}{(\log \log n)^2} \vee \frac{nps^2}{(\log(nps^2))^2}$, we obtain a partial matching that is both large (in terms of the number of matched nodes) and accurate (no mismatches). More specifically, the $k$-core matching step recovers the node correspondence for all but $n^{1 - \frac{ns^2(p + q)}{2 \log n}}$ nodes. 
The remaining unmatched node pairs follow the correlated Gaussian Mixture Model framework, so we apply Theorem~\ref{thm:gmm matching ach} to match those residual pairs via the node-attribute estimator.
Since the size of this unmatched node set is $n^{1 - \frac{ns^2(p + q)}{2 \log n}}$, it suffices to achieve an SNR of
$
\frac{d}{4} \log \frac{1}{1 - \rho^2} \geq \log\left(n^{1 - \frac{ns^2(p + q)}{2 \log n}}\right),
$
which yields the condition in \eqref{eq:csbm matching cond1}.

\begin{remark}[Interpretation of Conditions for Exact Matching]
Condition \eqref{eq:csbm matching cond0} ensures that the $k$-core algorithm yields the correct matching for nodes within the $k$-core of $G_1 \wedge_{\pi_*} G_2$. Condition \eqref{eq:csbm matching cond2} is required for matching the remaining nodes using node attributes. The combined inequality \eqref{eq:csbm matching cond1} highlights the benefit of leveraging both edge and node-attribute correlations. As a baseline, if $d=0$ (no node attributes), the model reduces to correlated SBMs, and \eqref{eq:csbm matching cond1} becomes $ns^2 \tfrac{p+q}{2} \ge (1+\epsilon)\,\log n$, matching the exact-matching threshold in \cite{RS21,RS23,YC23}. Conversely, if $p=q=0$ (no edges), the model is just correlated GMMs, and \eqref{eq:csbm matching cond1} becomes $\frac{d}{4}\log\!\bigl(\tfrac{1}{1-\rho^2}\bigr) \ge (1+\epsilon)\log n$, consistent with Theorem~\ref{thm:gmm matching ach}.
\end{remark}

\begin{theorem}[Impossibility for Exact Matching]\label{thm:csbm matching imp}
Let $(G_1,G_2) \sim \textnormal{CCSBMs}(n,p,q,s; R,d,\rho)$ be as in Section~\ref{sec:model2}, and assume $ps^2 = o(1)$. Suppose either
    \begin{equation}\label{eq:csbm matching imp1}
      ns^2\frac{p+q}{2}+\frac{d}{4} \log \frac{1}{1-\rho^2} \leq (1-\epsilon) \log n   \text{ and }  1\ll d =O(\log n)
    \end{equation}
    for arbitrarily small $\epsilon>0$, or
        \begin{equation}\label{eq:csbm matching imp2}
        n s^2\frac{p+q}{2}+\frac{d}{4} \log \frac{1}{1-\rho^2} \leq \log n -\log d -\omega(1) \text{ and } \frac{1}{\rho^2} -1 \leq \frac{d}{40}.
    \end{equation}Then, no estimator can achieve exact matching with high probability.
\end{theorem}

To prove Theorem~\ref{thm:csbm matching imp}, we show that even the MAP estimator fails under a strengthened information setting: specifically, when we are given not only  $G_1$, $G_2$, but also the mean vector $\boldsymbol{\mu}$, and partial knowledge of the ground truth permutation $\pi_*$ on all nodes except those in the set
\begin{equation}
    \mathcal{H}_*:=\left\{ i\in [n] : \forall j \in [n], A_{ij}B_{\pi_*(i)\pi_*(j)} = 0\right\},
\end{equation}
which consists of isolated nodes in the intersection graph $A \wedge_{\pi_*} B$. Additionally, we assume access to the community labels $\bssigma^1$ on $\mathcal{H}_*$ (i.e., on $ \mathcal{H}^+_*$ and $ \mathcal{H}^-_*$), and the complete community labels $\bssigma^2$.
Even under this augmented information regime, we show that the MAP estimator fails if either condition~\eqref{eq:csbm matching imp1} or~\eqref{eq:csbm matching imp2} holds. A detailed argument is provided in Section~\ref{sec:proof of csbm matching imp}.

\begin{figure}[!htb]
  \centering
  \includegraphics[scale=0.3]{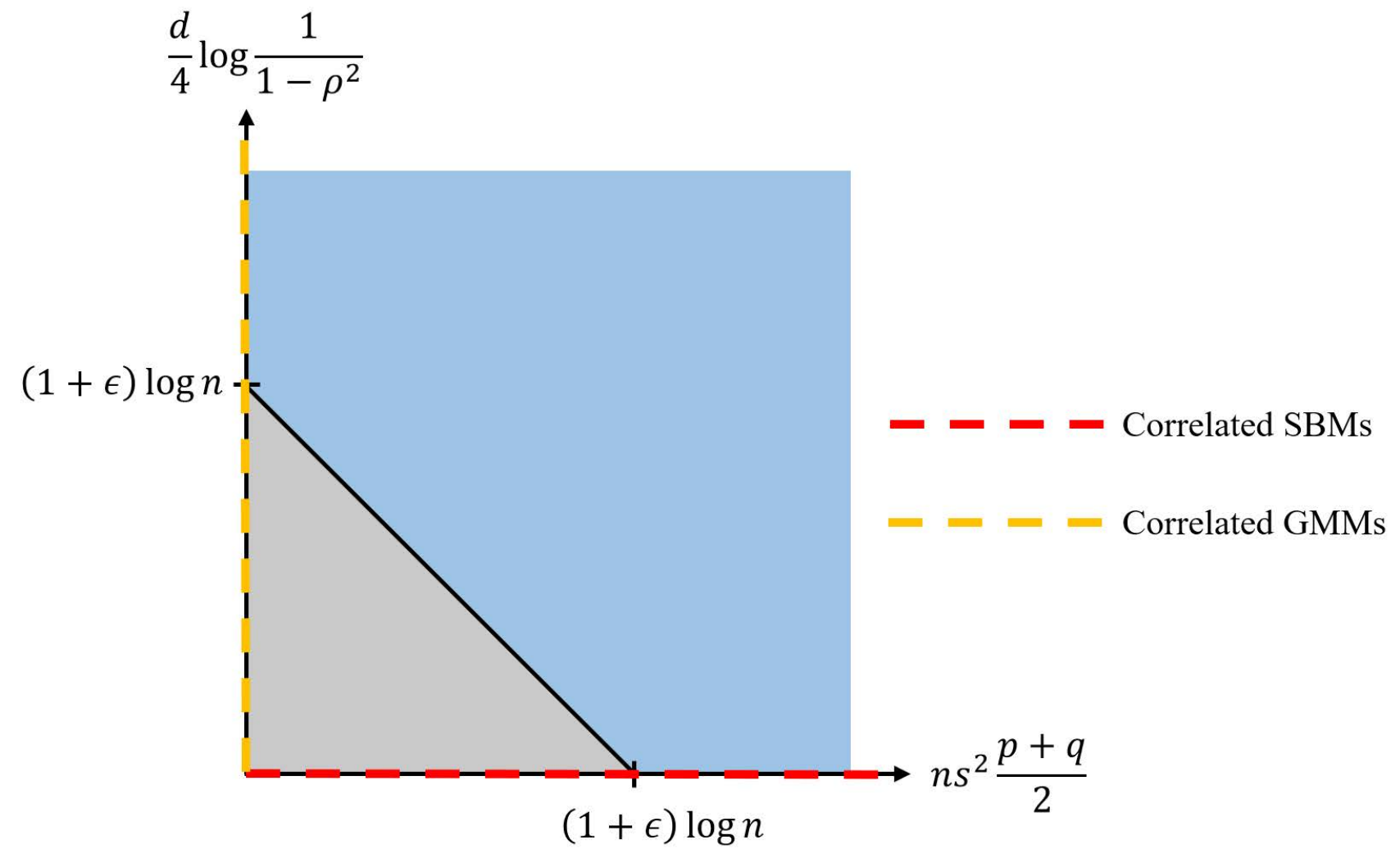}
  \caption{
    The region in \textbf{blue} indicates where exact matching is possible in correlated CSBMs, and the \textbf{gray} region shows where matching is infeasible. The red dashed line ($d=0$) captures the special case of correlated SBMs, while the yellow dashed line ($p=q=0$) represents correlated GMMs. Each dashed line intersects the boundary of the blue region at its respective matching threshold, thus recovering the known results for correlated SBMs and correlated GMMs. 
  }
  \label{fig:exact matching}
\end{figure}

Figure \ref{fig:exact matching} illustrates the parameter regions where exact node matching is (blue area) and is not (gray area) achievable in correlated Contextual Stochastic Block Models. It also highlights how removing either edges ($p,q=0$) or node attributes ($d=0$) reduces the region of feasible matching to the thresholds for correlated GMMs and correlated SBMs, respectively.

\subsection{Exact Community Recovery}
\label{sec:ccsbms_recovery}

In \cite{Abbe17}, it was shown that if \(p, q = \omega\bigl(\tfrac{\log n}{n}\bigr)\) ($p>q$), then \emph{exact community recovery} in a single CSBM graph \(G_1\) is achievable using only edge information. Likewise, it was shown in \cite{Ndaoud22,abbe2022} that if the effective SNR
$
\frac{R^2}{R + d/n} = \omega(\log n)
$ where
$\|\boldsymbol{\mu}\|^2 = R,
$
then exact community recovery is possible using only node attributes in \(G_1\). Accordingly, for correlated CSBMs, we focus on the setting
\begin{equation}\label{eq:assumption}
    p=\frac{a\log n}{n},\quad q=\frac{b\log n}{n},\quad \frac{R^2}{R + d/n}= c \log n,\quad \frac{\left(\frac{2}{1+\rho}R\right)^2}{\frac{2}{1+\rho}R+ d/n}=c' \log n
\end{equation}
for positive constants \(a,b,c,c'\). The last term in \eqref{eq:assumption} represents the SNR for the averaged correlated Gaussian attributes in \eqref{eqn:avg_feature}. Under these assumptions, we have the following feasibility and infeasibility results for exact community recovery.

\begin{theorem}[Achievability for Exact Community Recovery]\label{thm:csbm recovery ach}
  Let \((G_1,G_2) \sim \textnormal{CCSBMs}\bigl(n,p,q,s; R,d,\rho\bigr)\) be as in Section~\ref{sec:model2}. Suppose conditions \eqref{eq:csbm matching cond0}, \eqref{eq:csbm matching cond2}, and \eqref{eq:csbm matching cond1} hold, and also assume \eqref{eq:assumption}. If
  \begin{equation}\label{eq:csbm recov ach cond1}
      \frac{\bigl(1 - (1-s)^2\bigr)\,\bigl(\sqrt{a}-\sqrt{b}\bigr)^2 + c'}{\,2\,} \;>\; 1,
  \end{equation}
  then there exists an estimator \(\hat{\bssigma}(G_1,G_2)\) such that \(\hat{\bssigma}(G_1,G_2)=\bssigma\) with high probability.
\end{theorem}

When exact matching is achievable, we can construct a new Contextual Stochastic Block Model by forming a denser union graph $G_1 \vee_{\pi_*} G_2$ whose edges follow
\(\textnormal{SBM}\bigl(n,p(1-(1-s)^2),q(1-(1-s)^2)\bigr)\), and by assigning to each node the averaged correlated attribute from \eqref{eqn:avg_feature}. Applying the techniques of \cite{abbe2022} to this merged structure establishes Theorem~\ref{thm:csbm recovery ach}.

\begin{remark}[Comparison with the Single-Graph Setting]
  Abbe et al.~\cite{abbe2022} showed that exact community recovery in a \emph{single} CSBM \(G_1\) is possible if 
  \[
  \frac{s\bigl(\sqrt{a}-\sqrt{b}\bigr)^2 + c}{2} > 1.
  \]
  When two correlated graphs \(G_1\) and \(G_2\) are available, an exact matching (if feasible) allows one to form a denser \emph{union graph} from the correlated edges and to average the correlated attributes. This increases the effective SNR term to
  \[
  \frac{\bigl(1 - (1-s)^2\bigr)\bigl(\sqrt{a}-\sqrt{b}\bigr)^2 + c'}{2},
  \]
  thereby relaxing the threshold to the condition in \eqref{eq:csbm recov ach cond1}. A detailed proof of Theorem \ref{thm:csbm recovery ach} can be found in Section~\ref{sec:proof of csbm recovery ach}.
\end{remark}

\begin{theorem}[Impossibility for Exact Community Recovery]\label{thm:csbm recovery imp}
  Let \((G_1,G_2) \sim \textnormal{CCSBMs}\bigl(n,p,q,s; R,d,\rho\bigr)\) be as in Section~\ref{sec:model2}, and assume \eqref{eq:assumption}. If
  \begin{equation}\label{eq:csbm rec imp cond1}
      \frac{\bigl(1 - (1-s)^2\bigr)\bigl(\sqrt{a}-\sqrt{b}\bigr)^2 + c'}{2} < 1,
  \end{equation}
  then no estimator can achieve exact community recovery with high probability.
\end{theorem}

\noindent
The proof of Theorem~\ref{thm:csbm recovery imp} resembles the simulation argument used in Theorem~3.4 of \cite{RS21}, where one constructs a graph \(\mathrm{H}\) to mirror \((G_1,G_2)\). By showing that exact community recovery in \(\mathrm{H}\) is impossible under \eqref{eq:csbm rec imp cond1}, it follows that \((G_1,G_2)\) also fails to admit exact community recovery. A detailed proof is given in Section~\ref{sec:proof of csbm recovery imp}.

Figure \ref{fig:community recovery csbm} visualizes the parameter ranges under which community recovery is possible or impossible in correlated CSBMs. It highlights how adding correlated node attributes and edges can expand the regime where exact community detection succeeds.

\begin{figure}[ht]  
  \centering
  \includegraphics[scale=0.3]{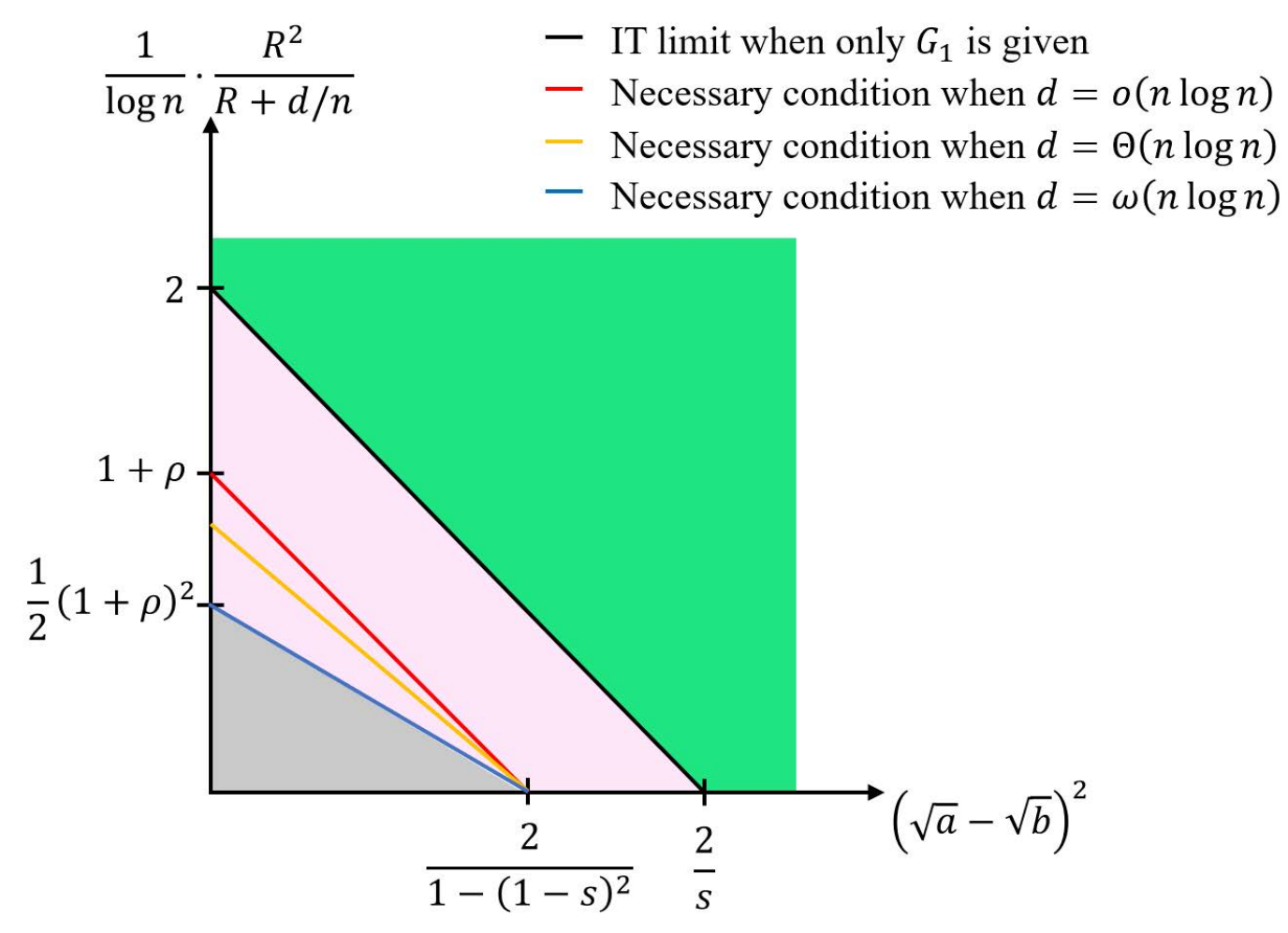}
  \caption{
    \textbf{Green region}: exact community recovery is possible using a single graph \(G_1\). 
    \textbf{Gray region}: exact recovery remains impossible even when two correlated graphs \((G_1,G_2)\) are given. 
    The \textbf{red, yellow,} and \textbf{blue lines} mark necessary conditions for exact recovery, depending on the dimensional regime of \(d\). 
    In particular, the \textbf{pink area} indicates where single-graph recovery fails, but having two correlated graphs \((G_1,G_2)\) can enable exact recovery if exact matching is feasible.
  }
  \label{fig:community recovery csbm}
\end{figure}

\begin{remark}[SNR of Node Attributes]
  Suppose
  \[
  \frac{\left(\frac{2}{1+\rho} R\right)^2}
       {\tfrac{2}{1+\rho}R + \tfrac{d}{n}}
  = 2\,\log n
  \quad\Longrightarrow\quad
  \frac{2}{1+\rho}\,R
  \approx
  \sqrt{\frac{2d}{n}\,\log n}
  \;\;\text{or}\;\;
  2\log n,
  \]
  depending on whether \(d\) dominates \(n\log n\) or not. 
  In the high-dimensional regime \(d = \omega(n\log n)\), we obtain 
  \(\frac{2}{1+\rho}\,R \approx \sqrt{\tfrac{2d}{n}\log n}\), yielding 
  \(\frac{R^2}{\,R + \tfrac{d}{n}\,}\approx \tfrac{1}{2}(1+\rho)^2\log n\). 
  In contrast, if \(d=o(n \log n)\), we have 
  \(\frac{2}{1+\rho}\,R\approx 2\log n\), so 
  \(\frac{R^2}{\,R + \tfrac{d}{n}\,}\approx (1+\rho)\log n\). 
  Finally, in the intermediate case \(d=\Theta(n\log n)\), the achievable SNR lies between these two extremes. These scenarios explain how the \(y\)-intercept in Figure~\ref{fig:community recovery csbm} varies with respect to \(d\).
\end{remark}

\begin{remark}[Scaling Regime for Community Detection]
Compared to the correlated GMM setting, the correlated CSBM model incorporates additional side information in the form of correlated edge structures between the two graphs. A comparison between Theorem~\ref{thm:gmm recovery ach} (achievability for community detection in CGMMs) and Theorem~\ref{thm:csbm recovery ach} (achievability for community detection in CCSBMs) demonstrates that this additional structural signal further relaxes the requirement for successful community detection:
\[
c' > 2
\quad \text{to} \quad
\left(1 - (1 - s)^2\right)\left(\sqrt{a} - \sqrt{b} \right)^2 + c' > 2.
\]
From a practical standpoint, both conditions--either \( c' > 2 \) or 
$
\left(1 - (1 - s)^2\right)\left(\sqrt{a} - \sqrt{b} \right)^2 + c' > 2
$--translate into logarithmic scaling of the attribute signal power \( \|\boldsymbol{\mu}\|^2 = R = \Omega(\log n) \) and/or the average node degree \( n(p + q)/2 = \Omega(\log n) \).

In real-world networks such as Facebook, where the number of effective users is on the order of \( 10^9 \), we have \( \log 10^9 \approx 20 \). Thus, the required signal strength or average degree per node need only be on the order of a few dozen. For example, in binary attribute settings (e.g., gender, education, occupation), users can be represented using tens of features. Similarly, the average number of friends per user in modern social networks typically exceeds this threshold. Hence, the theoretical requirements established in our analysis are not only reasonable but are also easily satisfied in practice.
\end{remark}

\section{Outline of the Proof}
\label{sec:outline}

Without loss of generality, let $\pi_* : [n]\to [n]$ be the identity permutation.  In this section, we outline the proofs of exact matching for the two proposed models.

\subsection{Proof Sketch of Theorem \ref{thm:gmm matching ach}}\label{sec:outline gmm matching ach}
We begin by explaining how Theorem~\ref{thm:gmm matching ach} can be established via the estimator in \eqref{eq:estimator}. The argument closely follows the ideas of \cite{KN22}, but with technical novelties highlighted in Remark \ref{rem:tech_nov_Thm1}.

Define
 \begin{equation}
     Z_{ij}:=  \lVert \bsx_i - \bsy_{j} \rVert^2.
 \end{equation}
Then the estimator \eqref{eq:estimator} can be rewritten as
 \begin{equation}\label{eq:estimator-1}
     \hat{\pi} = \argmin_{\pi \in S_n} \sum\limits_{i=1}^n  Z_{i \pi(i)}.
 \end{equation}
 Let
\begin{equation}
    \mathcal{M} := \{ i\in [n] : \hat{\pi}(i) \neq \pi_*(i)\}.
\end{equation}
be the set of mismatched nodes. Our goal is to show that $\P\bigl(|\mathcal{M}|=0\bigr)=1 - o(1)$.

Consider the event
 \begin{equation}\label{eq:t augment}
      \mathcal{F}_t=\left\{ \sum_{k=1}^t Z_{i_k i_k} \geq \sum_{k=1}^t Z_{i_k i_{k+1}}  \right\},
 \end{equation}
where $i_1,\ldots,i_t$ are $t$ distinct nodes in $[n]$, and $i_{t+1}=i_1$. We will show that the probability of $\mathcal{F}_t$ is sufficiently small if either \eqref{eq:gmm matching cond1} or \eqref{eq:gmm matching cond2} holds. Denote this upper bound by $f(t)$. Since there are ${n \choose t}(t-1)!$ ways to choose $t$ distinct, cyclically ordered nodes, we obtain
 \begin{equation}
     \E(|\mathcal{M}|) \leq \sum\limits_{t=2}^n t f(t){n \choose t} (t-1)!
 \end{equation}
where $t=1$ cannot yield a mismatch by itself. If $f(t)$ is made sufficiently small, it follows that $\E\bigl(|\mathcal{M}|\bigr)=o(1)$, and by Markov's inequality, $\P\bigl(|\mathcal{M}|=0\bigr)=1-o(1)$. Achieving a suitably small $f(t)$ necessitates either \eqref{eq:gmm matching cond1} or \eqref{eq:gmm matching cond2}. A rigorous proof can be found in Section~\ref{sec:proof of gmm matching ach}.

\begin{remark}[Technical Novelty in the Proof of Theorem \ref{thm:gmm matching ach}]\label{rem:tech_nov_Thm1}
In the proof of Theorem~\ref{thm:gmm matching ach}, we analyze a minimum-distance estimator \eqref{eq:estimator-1} that minimizes the total squared distance across all pairs of node attributes. While this estimator coincides with the MAP rule when the community labels in both databases are known (as shown in Eq.~\eqref{eq:MAP cgmms}), in our setting the labels are unknown. Although similar estimators were previously studied in  \cite{KN22} under the assumption that node attributes are uniformly distributed, resulting in symmetry in pairwise distances, our setup introduces new challenges: the distance between two nodes depends on whether they belong to the same community or not.
In particular, as discussed in Remark \ref{rem: Int_cond_Thm1},
\begin{itemize}
    \item The distance between a correctly matched pair follows a chi-squared distribution $\chi^2(d)$ scaled by $2(1 - \rho)$, with mean $2(1 - \rho)d$ and variance $8(1 - \rho)^2 d$.
    \item An incorrect pair from the \emph{same community} follows a scaled $\chi^2(d)$ distribution with scale 2, with mean $2d$ and variance $8d$.
    \item An incorrect pair from \emph{different communities} follows a \emph{non-central} chi-squared distribution $\chi^2(d, 2\|\boldsymbol{\mu}\|^2)$ scaled by $2$, with mean $4\|\boldsymbol{\mu}\|^2 + 2d$ and variance $32\|\boldsymbol{\mu}\|^2 + 8d$.
\end{itemize}

While the correct pair has the smallest expected distance, the incorrect pairs, especially those from different communities, have higher variance, which may dominate and introduce errors in the matching. Furthermore, the dependence on unknown community labels breaks the distance symmetry assumed in previous analyses  \cite{KN22} and complicates the estimation of the probability of \emph{cyclic errors} \eqref{eq:t augment}. We address this by introducing two high-probability ``good" events, $\mathcal{A}_1$ and $\mathcal{A}_2$ (defined in Eq.~\eqref{eqn:def_A1} and~\eqref{eqn:def_A2}), under which we can tightly bound $\mathbb{P}(\mathcal{F}_t)$.
Specifically, we show that $\mathbb{P}(\mathcal{F}_t \cap \mathcal{A}_1)$ and $\mathbb{P}(\mathcal{F}_t \cap \mathcal{A}_2)$
is sufficiently small (in Lemma \ref{lem:upperbound F_t}), enabling exact matching when
\[
\frac{d}{4} \log \frac{1}{1 - \rho^2} \geq (1 + \epsilon) \log n,
\]
which also matches the necessary condition in Theorem~\ref{thm:gmm matching imp}.

The events $\mathcal{A}_1$ and $\mathcal{A}_2$, which occur with high probability when either $\|\boldsymbol{\mu}\|^2 \ge 2 \log n + \omega(1)$ or $d = \omega(\log n)$, are carefully designed to ensure that incorrect pairs from different communities are unlikely to yield higher error probability than incorrect pairs from the same community. This allows us to achieve exact matching \emph{without performing community detection}, which is our ultimate goal. In other words, we demonstrate that exact matching can be achieved as an intermediate step, even in the absence of label recovery, and can be subsequently leveraged to enhance community detection.
\end{remark}

\subsection{Proof Sketch of Theorem \ref{thm:csbm matching ach}}
Exact matching in correlated CSBMs is proved via a \emph{two-step} algorithm:

\begin{enumerate}
    \item $k$-core matching using edges:
    We choose 
    \[
k=\frac{\log n}{(\log \log n)^2}\vee \frac{nps^2}{(\log nps^2)^2}.
    \]
    Applying $k$-core matching with this $k$ matches $n - o(n)$ nodes correctly, leaving a set $F$ of unmatched nodes satisfying
\begin{equation}
    |F|\leq  n^{1-\frac{ns^2(p+q)}{2\log n}+o(1)} .
\end{equation}
    \item Matching the remainder via node attributes:  
    The remaining unmatched nodes in $F$ form a correlated Gaussian Mixture Model. Under assumptions \eqref{eq:csbm matching cond2} and \eqref{eq:csbm matching cond1},  Theorem~\ref{thm:gmm matching ach} applies, ensuring that these nodes can also be matched correctly using their node attributes.
\end{enumerate}
A full proof of Theorem \ref{thm:csbm matching ach}, detailing each step, can be found in Section~\ref{sec:proof of csbm matching ach}.

\begin{remark}[Technical Novelty in the Proof of Theorem \ref{thm:csbm matching ach}]
The main technical challenge in proving the achievability of exact matching in CCSBMs lies in \emph{how to combine the two distinct sources of correlation}, the correlated edges and the correlated node attributes, to achieve the tight information-theoretic bound \eqref{eq:csbm matching cond1}:
\beq\nonumber
ns^2\frac{p+q}{2}+\frac{d}{4} \log \frac{1}{1-\rho^2} \geq (1+\epsilon) \log n,
\eeq
which is tight only when the signals from both edge and attribute correlations are optimally integrated.

To this end, we propose a two-step matching procedure, described above. First, we use $k$-core matching to recover a partial matching based solely on edge information. Then, we apply minimum-distance matching to the remaining unmatched nodes using their attribute vectors.
The $k$-core matching step recovers the correspondence for all but $n^{1 - \frac{ns^2(p + q)}{2 \log n}+o(1)}$ nodes. A key reason we adopt $k$-core matching in the first step, despite its exponential complexity in general, is that it not only identifies a correct partial matching, but also reveals which specific nodes are matched (those in the $k$-core) and which are not. This property distinguishes it from other partial matching methods that offer similar recovery guarantees but do not explicitly indicate which nodes have been matched. Knowing the identity of matched versus unmatched nodes is crucial for the second step.

In the second step, we apply minimum-distance matching to the unmatched node set. Since the size of this set is $n^{1 - \frac{ns^2(p + q)}{2 \log n}}$, it suffices to achieve an signal-to-noise ratio (SNR) of
\[
\frac{d}{4} \log \frac{1}{1 - \rho^2} \geq \log\left(n^{1 - \frac{ns^2(p + q)}{2 \log n}}\right),
\]
which yields the information-theoretically tight condition \eqref{eq:csbm matching cond1}.

Moreover, while our analysis of $k$-core matching builds on prior theoretical results~\cite{CKNP20, RS23, GRS22}, most previous works provided only coarse estimates of the matched set size, typically of the form $\Omega(n)$ or $n - o(n)$, or derived exact sizes only under the assumption that $p, q = \Theta(\log n / n)$. Our contribution generalizes this analysis across a broader range of $p, q$ regimes.
Since smaller values of $k$ yield larger matched sets but lower accuracy, and larger values of $k$ yield smaller matched sets but higher accuracy, we select
$
k = \frac{\log n}{(\log \log n)^2} \vee \frac{nps^2}{(\log(nps^2))^2}
$
to balance this trade-off. Under these conditions, we demonstrate that it is still possible to recover a sufficiently large and reliable matched node set for the two-step procedure to succeed.
\end{remark}

\section{Discussion and Open Problems}\label{sec:discu and open}

In this work, we investigated the problem of community recovery in the presence of two correlated networks. We introduced, for the first time, the \emph{Correlated Gaussian Mixture Model (CGMM)}, which captures correlation through node attributes, and the \emph{Correlated Contextual Stochastic Block Model (CCSBM)}, which incorporates both correlated node attributes and edge structures. To the best of our knowledge, CGMMs and CCSBMs are the first statistical models that jointly address the alignment (via graph matching) and integration of two correlated networks with latent community structure, while simultaneously accounting for correlations in both edge connections and node-level features.

We showed that in certain parameter regimes where exact community recovery is information-theoretically impossible using a single network, the availability of a second, correlated network can provide powerful side information, rendering exact recovery feasible. To establish these results, we first derived sharp conditions under which \emph{exact matching} between the two networks is achievable. Notably, we demonstrated that exact matching becomes significantly more tractable when both node attributes and edge structures are correlated, compared to scenarios where only one modality is present. We conclude by outlining several open problems and potential extensions motivated by our findings.

\subsubsection{Closing the information-theoretic gap for exact matching}
Consider first the case of Correlated Gaussian Mixture Models. In Theorem~\ref{thm:gmm matching ach}, either \(d=\omega(\log n)\) or \(\|\boldsymbol{\mu}\|^2 \ge 2\log n + \omega(1)\) must hold for exact matching. In contrast, in the \emph{community-free} correlated Gaussian database model (i.e., \(\boldsymbol{\mu}=0\)), Dai et al.~\cite{DCK19} showed that exact matching is possible with
$\frac{d}{4}\log \frac{1}{1-\rho^2}\geq \log n + \omega(1)$ 
without further constraints on \(d\). Meanwhile, Theorem~\ref{thm:gmm matching imp} leaves some gaps, especially for \(d=\Omega(n)\). We hypothesize that higher-dimensional attributes should make exact matching progressively easier, suggesting that
$\frac{d}{4}\log \frac{1}{1-\rho^2}\geq (1+\epsilon) \log n$
may indeed be the core information-theoretic threshold, even without additional conditions on \(\boldsymbol{\mu}\) and \(d\). Resolving this would also address the same gap in the CCSBM setting, since the latter relies on CGMMs for unmatched nodes.

\subsubsection{Closing the information-theoretic gap for exact community recovery}
Theorems~\ref{thm:gmm recovery ach} and~\ref{thm:csbm recovery ach} establish achievable regions for exact community recovery in CGMMs and CCSBMs under the assumption that exact matching is possible. Comparing these to the converse results (Theorems~\ref{thm:gmm recovery imp} and~\ref{thm:csbm recovery imp}) reveals that the stated conditions are not necessarily tight. Thus, it is natural to ask whether \emph{exact matching} is truly needed for exact community recovery.  
Indeed, in correlated SBMs, R\'acz and Sridhar~\cite{RS21} used exact matching to derive conditions for community recovery, whereas Gaudio et al.~\cite{GRS22} refined the threshold by using \emph{partial matching} (e.g., $k$-core). Similarly, one might conjecture that for correlated GMMs there exists a constant \(\kappa\) such that, even if $\frac{d}{4} \log \frac{1}{1-\rho^2} \geq (1-\kappa) \log n$
there exists an estimator \(\tilde{\bssigma}\) that can exactly recover \(\bssigma\) even without full matching. Exploring this possibility is an intriguing open problem.

\subsubsection{Generalizing to more communities}
Our analysis has focused on the simplest case of two communities. A natural next step is to consider models with \(r \ge 2\) communities. It would be interesting to investigate whether the conditions for exact matching and exact community recovery generalize in a straightforward manner or require fundamentally new techniques.

\subsubsection{Multiple correlated graphs}
We have considered only two correlated graphs. Another direction is to explore the setting where more than two correlated graphs are given. For example, Ameen and Hajek~\cite{AH24} established exact matching thresholds for $r \ge 2$ correlated Erd\H{o}s-R\'enyi graphs when the average degree is on the order of \(\log n\). An open question is to characterize exact matching in the presence of $r$ correlated graphs (possibly with node attributes) and then design optimal strategies for combining their edge and attribute information to further improve community recovery. We conjecture that taking the union of edges and averaging node attributes across multiple graphs may yield the best results.

\subsubsection{Efficient algorithms}
From a computational viewpoint, there remain major challenges. Although \cite{abbe2022} showed that exact community recovery can be done with spectral methods in Gaussian Mixture Models or Contextual Stochastic Block Models in \(O(n^3)\) time, exact matching under CGMMs currently relies on a Hungarian algorithm to minimize pairwise distances, incurring \(O(n^2 d + n^3)\) complexity. For CCSBMs, the situation is more severe: the $k$-core method requires examining all matchings, leading to $\Theta(n!)$ complexity. Developing \emph{polynomial-time} algorithms that achieve $k$-core-level performance for exact matching--and thus enable more efficient community recovery in CCSBMs would be a significant breakthrough for large-scale network analysis.

\subsubsection{A Closer Model to the Real World}
The Stochastic Block Model (SBM) assumes that edge probabilities depend solely on community memberships and that edges are generated independently. However, real-world networks often exhibit spatial or geometric structure. A representative model that captures these characteristics is the Geometric Stochastic Block Model (GSBM)~\cite{gaudio2024exact,gaudio2025sharp}, which incorporates node proximity to more accurately reflect real-world structures.
Additionally, in real networks, connection probabilities are not uniform but vary according to node degrees. The Degree-Corrected Stochastic Block Model (DC-SBM)~\cite{karrer2011stochastic} accounts for this heterogeneity by allowing degree variability across nodes.
Based on these models, an interesting problem is to investigate whether it is possible to construct correlated networks by sampling two subgraphs and to examine whether the $k$-core matching approach can be extended in this setting to identify the conditions required for exact matching.

\section{Proof of Theorem \ref{thm:gmm matching ach}: \\Achievability of Exact Matching in Correlated Gaussian Mixture Models}\label{sec:proof of gmm matching ach}
 We analyze the estimator \eqref{eq:estimator-1}, which finds a permutation that minimizes the sum of attribute distances, and establish the conditions under which no mismatched node pairs arise. Our proof technique builds on the approach of \cite{KN22}, where the estimator \eqref{eq:estimator-1} was analyzed in the context of geometric partial matching without community structures, assuming an identical distribution for all node attribute vectors.  In contrast, we analyze the correlated Gaussian Mixture Models where node attribute distributions vary with the unknown community labels. We demonstrate that the estimator \eqref{eq:estimator-1} achieves the exact matching when conditions \eqref{eq:gmm matching cond1} or \eqref{eq:gmm matching cond2} are met.

First, for the analysis, we present a lemma that provides an upper bound for $\P\left(\mathcal{F}_t\right)$, where $\mathcal{F}_t$ is defined in \eqref{eq:t augment}. For $t\geq 2$ and $\alpha>0$, let us define 
\begin{equation}\label{eq:def S}
    S\left(\alpha, t\right):=\sum_{j=1}^{t-1}\log \left(1+\frac{1}{2 \alpha}\left(1-\cos \left(\frac{2 \pi}{t}j\right)\right)\right),
\end{equation}
\begin{equation}\label{eq:def I}
    I(\alpha) :=\int_0^1\log \left(1+\frac{1}{2 \alpha}\left(1-\cos \left(2 \pi x\right)\right)\right) d x .
\end{equation}
Additionally, let us define two events

\begin{equation}\label{eqn:def_A1}
\mathcal{A}_1:=\left\{    -\|\bsmu\|^2\leq \langle \bsmu, \bsz_i \rangle \leq \|\bsmu\|^2 \text{ for all } i\in[n] \right\}
\end{equation}
\begin{equation}
\begin{gathered}\label{eqn:def_A2}
        \mathcal{A}_2:=\left\{ \frac{\rho-\lambda}{2}\left\|\bsz_i-\bsz_j+\boldsymbol{\mu}+\frac{1-\lambda}{\rho-\lambda} \boldsymbol{\mu}\right\|^2 \geq \frac{\rho-\lambda}{2}\left(1+\frac{1-\lambda}{\rho-\lambda}\right)^2\|\boldsymbol{\mu}\|^2-2(1-\lambda)\|\boldsymbol{\mu}\|^2 \text{ and } \right. \\  \left.  \frac{\rho-\lambda}{2}\left\|\bsz_i-\bsz_j+\boldsymbol{\mu}-\frac{1-\lambda}{\rho-\lambda} \boldsymbol{\mu} \right\|^2 \geq \frac{\rho-\lambda}{2}\left(1-\frac{1-\lambda}{\rho-\lambda}\right)^2\|\boldsymbol{\mu}\|^2 \text{ for all distinct } i,j\in[n] \right\}.
\end{gathered}
\end{equation}

\begin{lemma}\label{lem:upperbound F_t}
    For any $t$ distinct integers $i_1,...,i_t \in [n]$, it holds that
    \begin{equation}\label{eq:upperbound F_t1}
        \P\left(\mathcal{F}_t\right) \leq \exp\left(-\frac{d}{2}S\left(\frac{1-\rho^2}{\rho^2},t\right) \right)+\P(\mathcal{A}^c_1); 
    \end{equation}
    and for some $\lambda \in (0,\rho)$, it holds that
     \begin{equation}\label{eq:upperbound F_t2}
        \P\left(\mathcal{F}_t\right) \leq  \exp\left(-\frac{d}{2}S\left(\frac{1-\rho^2}{\lambda^2},t\right) \right)+\P(\mathcal{A}^c_2) .
    \end{equation}
\end{lemma}
We analyze $\P(\mathcal{F}_t)\leq \P(\mathcal{F}_t\cap\mathcal{A}_1 )+\P(\mathcal{A}^c_1)$ and $\P(\mathcal{F}_t)\leq \P(\mathcal{F}_t\cap\mathcal{A}_2 )+\P(\mathcal{A}^c_2)$ to derive \eqref{eq:upperbound F_t1} and \eqref{eq:upperbound F_t2}, respectively. 
The full proof of Lemma \ref{lem:upperbound F_t} can be found in Section \ref{sec:proof of lemma upperbound f_t}. 
Furthermore, it is known that  $S(\alpha,t)$ has the following lower bound:

\begin{lemma}[Corollary 2.1 in \cite{KN22}]\label{lem:cor in kn}
For all $t_0\geq 2$ and $t>t_0$, we have
\begin{equation}
    S(\alpha,t)>S(\alpha,t_0)+(t-t_0) I(\alpha),
\end{equation}
where  $t_0 I(\alpha)>S(\alpha,t_0) $.
\end{lemma}

Using the above two lemmas, we can prove Theorem \ref{thm:gmm matching ach} as follows:

\begin{IEEEproof}[Proof of Theorem \ref{thm:gmm matching ach}]
    Recall that the estimator we use is $\hat{\pi} = \argmin_{\pi \in S_n} \sum\limits_{i=1}^n  Z_{i \pi(i)}$. We define the set of mismatched nodes by $ \mathcal{M} := \{ i\in [n] : \hat{\pi}(i) \neq \pi_*(i)\}$. To show that $\hat{\pi}$ can achieve exact matching, we will show that $\P(|\mathcal{M}|=0)=1-o(1)$. First, let us show that exact matching is possible when \eqref{eq:gmm matching cond1} holds. 
 When $\lVert \boldsymbol{\mu} \rVert^2 \geq 2\log n +\omega(1)$, we can have that
    \begin{equation}
    \begin{aligned}\label{eqn:bound_EM}
        \P(\mathcal{A}_1)  \geq 1-n \P\left(\langle \boldsymbol{\mu},\bsz_1\rangle \leq - \lVert \boldsymbol{\mu}\rVert^2 \text{ or } \langle \boldsymbol{\mu},\bsz_1\rangle \geq  \lVert \boldsymbol{\mu}\rVert^2 \right) \geq 1-2n e^{-\frac{1}{2}\lVert \boldsymbol{\mu} \rVert^2} = 1-o(1)
    \end{aligned}
    \end{equation}
    by the tail bound of normal distributions (Lemma \ref{lem:normal tail}).     
    Therefore,  we can obtain
    \begin{equation}\label{eq:EM1}
    \begin{aligned}
        \E(|\mathcal{M}|) &\leq \sum\limits_{t=2}^n \P(\mathcal{F}_t) t{n \choose t} (t-1)!\\
        &\stackrel{(a)}{\leq}  \sum\limits_{t=2}^n \ \exp\left(-\frac{d}{2}S\left(\frac{1-\rho^2}{\rho^2},t\right) +o(1)\right) n(n-1)\cdots(n-t+1)\\
        &\stackrel{(b)}{\leq} \sum\limits_{t=2}^n \ \exp\left(t \log n-\frac{d}{2}S\left(\frac{1-\rho^2}{\rho^2},t\right) +o(1)\right)\\
        &\stackrel{(c)}{\leq}  \sum\limits_{t=2}^n \ \exp\left(t \log n-\frac{d}{2}S\left(\frac{1-\rho^2}{\rho^2},2\right) -\frac{d}{2}(t-2)I\left(\frac{1-\rho^2}{\rho^2}\right )+o(1) \right)\\
        &\stackrel{(d)}{\leq} \sum\limits_{t=2}^n \ \exp\left(t \log n-\frac{dt}{4} S\left(\frac{1-\rho^2}{\rho^2},2\right)+o(1) \right).
    \end{aligned}
    \end{equation}
The inequality $(a)$ holds by the first result \eqref{eq:upperbound F_t1} in Lemma \ref{lem:upperbound F_t}. The inequality $(b)$ holds since $\frac{n(n-1)\cdots(n-t+1)}{t} \leq n^t$. 
The inequality $(c)$ holds by Lemma \ref{lem:cor in kn} with $t_0=2$. 
To show inequality $(d)$, we use Lemma~\ref{lem:S<tI}.  
In showing
\beq
\begin{split}\nonumber
& t\log n-\frac{d}{2} S\left(\frac{1-\rho^{2}}{\rho^{2}}, 2\right)-\frac{d}{2}(t-2) I\left(\frac{1-\rho^{2}}{\rho^{2}}\right)\leq t\log n-\frac{d t}{4} S\left(\frac{1-\rho^{2}}{\rho^{2}}, 2\right),
\end{split}
\eeq
first note that the left-hand side can be expressed as 
\beq\nonumber
\begin{split}
&t\log n-\frac{d}{2} S\left(\frac{1-\rho^{2}}{\rho^{2}}, 2\right)-\frac{d}{2}(t-2) I\left(\frac{1-\rho^{2}}{\rho^{2}}\right)\\
&=\log n\left(2 -\frac{d}{2\log n} S\left(\frac{1-\rho^{2}}{\rho^{2}}, 2\right)+(t-2)\left(1-\frac{d}{2\log n }I\left(\frac{1-\rho^{2}}{\rho^{2}}\right)\right)\right).
\end{split}
\eeq
Using Lemma~\ref{lem:S<tI}, we have 
$$
\frac{1}{2} S\left(\frac{1-\rho^{2}}{\rho^{2}}, 2\right)< I\left(\frac{1-\rho^{2}}{\rho^{2}}\right),
$$
which is equivalent to
$$
2-\frac{d}{2\log n} S\left(\frac{1-\rho^{2}}{\rho^{2}}, 2\right)> 2-\frac{d}{2\log n} 2I\left(\frac{1-\rho^{2}}{\rho^{2}}\right)=2\left(1-\frac{d}{2\log n} I\left(\frac{1-\rho^{2}}{\rho^{2}}\right)\right).
$$
Using $\left(1-\frac{d}{2\log n} I\left(\frac{1-\rho^{2}}{\rho^{2}}\right)\right)<\frac{1}{2}\left(2-\frac{d}{2\log n} S\left(\frac{1-\rho^{2}}{\rho^{2}}, 2\right)\right) $, we get
\beq\nonumber
\begin{split}
&\log n\left(2 -\frac{d}{2\log n} S\left(\frac{1-\rho^{2}}{\rho^{2}}, 2\right)+(t-2)\left(1-\frac{d}{2\log n }I\left(\frac{1-\rho^{2}}{\rho^{2}}\right)\right)\right)\\
&\leq \log n\left(2 -\frac{d}{2\log n} S\left(\frac{1-\rho^{2}}{\rho^{2}}, 2\right)+\frac{(t-2)}{2}\left(2-\frac{d}{2\log n} S\left(\frac{1-\rho^{2}}{\rho^{2}}, 2\right)\right)\right)\\
&=\log n\left(\frac{t}{2} \left(2-\frac{d}{2\log n} S\left(\frac{1-\rho^{2}}{\rho^{2}}, 2\right)\right)\right)=t\log n-\frac{d t}{4} S\left(\frac{1-\rho^{2}}{\rho^{2}}, 2\right),
\end{split}
\eeq
which implies $(d)$.

From the definition of $S(\alpha,t)$ in \eqref{eq:def S}, we have $ S\left(\frac{1-\rho^2}{\rho^2},2\right)= \log \frac{1}{1-\rho^2}$. Thus, if $\frac{d}{4} \log \frac{1}{1-\rho^2} \geq \log n +\omega(1)$, then  $ \log n-\frac{d}{4} S\left(\frac{1-\rho^2}{\rho^2},2\right) \leq -\omega(1)$. 
Therefore, if \eqref{eq:gmm matching cond1} holds, from \eqref{eqn:bound_EM} and $ \log n-\frac{d}{4} S\left(\frac{1-\rho^2}{\rho^2},2\right) \leq -\omega(1)$, we get
  \begin{equation}
        \E(|\mathcal{M}|) \leq \sum\limits_{t=2}^n  \exp\left(-\omega(1)\cdot t\right) = o(1).
    \end{equation}
 Finally, by using Markov's inequality, we can show that
\begin{equation}\label{eq:1}
    \P(|\mathcal{M}|\geq 1)\leq \E(|\mathcal{M}|) =o(1).
\end{equation}
Thus, it holds that $\P(|\mathcal{M}|=0)=1-o(1)$.

We will next show that the exact matching is possible with high probability when \eqref{eq:gmm matching cond2} holds.  As the case where $\| \bsmu \|^2 \geq 2 \log n +\omega(1)$ has already been considered in the proof above, we now assume that $\| \bsmu\|^2 = O(\log n)$. First, we will show that $\P(\mathcal{A}_2) \geq 1- o(1)$ when $d = \omega(\log n )$.

For $\bsz_1,\bsz_2 \sim \mathcal{N}(0,\bsI_d)$, let $\bsz=\frac{\bsz_1-\bsz_2}{\sqrt{2}} \sim \mathcal{N}(0,\bsI_d)$. 
We can have that 
\begin{equation}
    \begin{aligned}
        &\P\left( \frac{\rho-\lambda}{2}\left\|\bsz_1-\bsz_2+\boldsymbol{\mu}+\frac{1-\lambda}{\rho-\lambda} \boldsymbol{\mu}\right\|^2 \leq \frac{\rho-\lambda}{2}\left(1+\frac{1-\lambda}{\rho-\lambda}\right)^2\|\boldsymbol{\mu}\|^2-2(1-\lambda)\|\boldsymbol{\mu}\|^2  \right)\\
        &= \P \left(\left\| \bsz+\frac{\rho-2\lambda+1}{\sqrt{2}(\rho-\lambda)}\boldsymbol{\mu}  \right\|^2 \leq  \frac{1}{2}\left(\frac{\rho-2\lambda+1}{\rho-\lambda}\right)^2\|\boldsymbol{\mu}\|^2-\frac{2(1-\lambda)}{\rho-\lambda}\|\boldsymbol{\mu}\|^2 \right).
    \end{aligned}
\end{equation}
Then, $\left\| \bsz+\frac{\rho-2\lambda+1}{\sqrt{2}(\rho-\lambda)} \bsmu \right\|^2$ follows a noncentral chi-squared distribution ($\chi^2$ distribution) with $d$ degrees of freedom and noncentrality parameter $\frac{(\rho-2\lambda+1)^2}{2(\rho-\lambda)^2}\|\boldsymbol{\mu}\|^2$. By the tail bound for noncentral chi-squared distribution (Lemma \ref{lem:noncentral}), if $\frac{\left(d+\frac{2-2 \lambda}{\rho-\lambda}\lVert \boldsymbol{\mu} \rVert^2\right)^2}{d+\left(\frac{\rho-2 \lambda+1}{\rho-\lambda}\right)^2\lVert \boldsymbol{\mu} \rVert^2} \geq 8 \log n +\omega(1)$, then
we have 
\begin{equation}\label{eq:noncent1}
    \P\left( \left\| \bsz+\frac{\rho-2\lambda+1}{\sqrt{2}(\rho-\lambda)}\boldsymbol{\mu}  \right\|^2 \leq  \frac{1}{2}\left(\frac{\rho-2\lambda+1}{\rho-\lambda}\right)^2\|\boldsymbol{\mu}\|^2-\frac{2(1-\lambda)}{\rho-\lambda}\|\boldsymbol{\mu}\|^2 \right)\leq \exp\left(-2\log n -\omega(1) \right).
\end{equation}
Similarly, we can have that
\begin{equation}
    \begin{aligned}
        &\P\left(  \frac{\rho-\lambda}{2}\left\|\bsz_1-\bsz_2+\boldsymbol{\mu}-\frac{1-\lambda}{\rho-\lambda} \boldsymbol{\mu} \right\|^2 \leq \frac{\rho-\lambda}{2}\left(1-\frac{1-\lambda}{\rho-\lambda}\right)^2\|\boldsymbol{\mu}\|^2 \right)\\
        &= \P \left(\left\| \bsz+\frac{\rho-1}{\sqrt{2}(\rho-\lambda)}\boldsymbol{\mu}  \right\|^2 \leq  \frac{1}{2}\left(\frac{\rho-1}{\rho-\lambda}\right)^2\|\boldsymbol{\mu}\|^2 \right).
    \end{aligned}
\end{equation}
Then, $\left\| \bsz+\frac{\rho-1}{\sqrt{2}(\rho-\lambda)}\boldsymbol{\mu}  \right\|^2$ follows a noncentral chi-squared distribution ($\chi^2$ distribution)  with $d$ degrees of freedom and noncentrality parameter $\frac{(\rho-1)^2}{2(\rho-\lambda)^2}\|\boldsymbol{\mu}\|^2$. Again, by the tail bound for noncentral chi-squared distribution (Lemma \ref{lem:noncentral}), if $\frac{d^2}{d+\left(\frac{\rho-1}{\rho-\lambda}\right)^2\lVert \boldsymbol{\mu} \rVert^2} \geq 8 \log n +\omega(1)$, then
we have 
\begin{equation}\label{eq:noncent2}
    \P\left( \left\| \bsz+\frac{\rho-1}{\sqrt{2}(\rho-\lambda)}\boldsymbol{\mu}  \right\|^2 \leq  \frac{1}{2}\left(\frac{\rho-1}{\rho-\lambda}\right)^2\|\boldsymbol{\mu}\|^2\right)\leq \exp\left(-2\log n -\omega(1) \right).
\end{equation}
Under the assumption $d=\omega(\log n)$ in \eqref{eq:gmm matching cond2} and $\| \bsmu\|^2 = O(\log n)$, for any constant $r\in(0,1)$ and the associated $\lambda=r \rho$, we have $\frac{\left(d+\frac{2-2 \lambda}{\rho-\lambda}\lVert \boldsymbol{\mu} \rVert^2\right)^2}{d+\left(\frac{\rho-2 \lambda+1}{\rho-\lambda}\right)^2\lVert \boldsymbol{\mu} \rVert^2} \geq 8 \log n +\omega(1)$ and  $\frac{d^2}{d+\left(\frac{\rho-1}{\rho-\lambda}\right)^2\lVert \boldsymbol{\mu} \rVert^2} \geq 8 \log n +\omega(1)$. Therefore, combining \eqref{eq:noncent1} and \eqref{eq:noncent2} and taking a union bound over all distinct $i,j\in[n]$, we obtain that \begin{equation}
    \P(\mathcal{A}_2)\geq 1-o(1).
\end{equation}

Thus, we can obtain that
\begin{equation}\label{eq:EM2}
\begin{aligned}
\E(|\mathcal{M}|) &\leq \sum\limits_{t=2}^n \P(\mathcal{F}_t) t{n \choose t} (t-1)!\\
&\stackrel{(e)}{\leq} \sum\limits_{t=2}^n \ \exp\left(-\frac{d}{2}S\left(\frac{1-\rho^2}{\lambda^2},t\right) +o(1)\right) n(n-1)\cdots(n-t+1)\\
&\leq \sum\limits_{t=2}^n \ \exp\left(t \log n-\frac{d}{2}S\left(\frac{1-\rho^2}{\lambda^2},t\right)+o(1) \right)\\
&\stackrel{(f)}{\leq}  \sum\limits_{t=2}^n \ \exp\left(t \log n-\frac{d}{2}S\left(\frac{1-\rho^2}{\lambda^2},2\right) -\frac{d}{2}(t-2)I\left(\frac{1-\rho^2}{\lambda^2}\right )+o(1) \right)\\
&\stackrel{(g)}{\leq} \sum\limits_{t=2}^n \ \exp\left(t \log n-\frac{dt}{4} S\left(\frac{1-\rho^2}{\lambda^2},2\right) +o(1)\right).
\end{aligned}
\end{equation}
The inequality $(e)$ holds by \eqref{eq:upperbound F_t2} in Lemma \ref{lem:upperbound F_t}, the inequality $(f)$ holds by Lemma \ref{lem:cor in kn} with $t_0=2$, and the inequality $(g)$ holds  by $\frac{1}{2} S\left(\frac{1-\rho^2}{\lambda^2},2\right)<I\left(\frac{1-\rho^2}{\lambda^2}\right )$, which can be shown using Lemma \ref{lem:S<tI}. Moreover, when $\frac{d}{4} \log \frac{1}{1-\rho^2} \geq (1+\epsilon)\log n $, there exists some constant $r\in (0,1)$, depending on $\epsilon$, such that $\frac{d}{4} \log \left(1+\frac{r^2\rho^2}{1-\rho^2}\right) \geq  (1+\epsilon/2)\log n $.  For such an $r$, let $\lambda=r\rho$. From the definition of $S(\alpha,t)$ in \eqref{eq:def S}, we have $ S\left(\frac{1-\rho^2}{\lambda^2},2\right)= \log \left(1+\frac{r^2\rho^2}{1-\rho^2}\right)$.
Thus, we have  $ \log n-\frac{d}{4} S\left(\frac{1-\rho^2}{\lambda^2},2\right) \leq -\frac{\epsilon}{2}\log n$. So if 
\eqref{eq:gmm matching cond2} holds, we can get
\begin{equation}
    \E(|\mathcal{M}|) \leq \sum\limits_{t=2}^n  \exp\left(-t\frac{\epsilon}{3} \log n \right) = o(1).
\end{equation}
 Finally, by using Markov's inequality, we can have that
\begin{equation}\label{eq:2}
    \P(|\mathcal{M}|\geq 1)\leq \E(|\mathcal{M}|) =o(1).
\end{equation}
Hence, it holds that $\P(|\mathcal{M}|=0)=1-o(1)$. By combining the results of \eqref{eq:1} and \eqref{eq:2}, we complete the proof of Theorem \ref{thm:gmm matching ach}.

\end{IEEEproof}

\section{Proof of Theorem \ref{thm:gmm matching imp}: \\Impossibility of Exact Matching in Correlated Gaussian Mixture Models}\label{sec:proof of gmm matching imp}

    Assume that the community labels  $\bssigma^1$ and $\bssigma^2$ are given. Without loss of generality, assume that at least half of the nodes, i.e., $n/2$ more, are assigned $+$ label. Let this set of nodes be denoted as $V'$. Additionally, assume that $\pi_*\left\{V \backslash V' \right\}$ and the mean vector $\boldsymbol{\mu}$ are also given. 
Then, we can consider that we are solving the database alignment problem for $|V'|$ nodes in the given correlated Gaussian databases.
Dai et al. \cite{DCK19} identified the conditions under which exact matching is impossible, and the results are as follows.
\begin{theorem}[Theorem 2 in \cite{DCK19}]\label{thm:dai imp}
Consider the correlated Gaussian databases $X,Y \in \mathbb{R}^{n\times d}$. Suppose that $1\ll d =O(\log n)$ and 
\begin{equation}
    \frac{d}{4}\log \frac{1}{1-\rho^2} \leq (1-\epsilon)\log n
\end{equation}
for an arbitrary small constant $\epsilon>0$. Then, there is no estimator that can exactly recover $\pi_*$ with probability $1-o(1).$
\end{theorem}
Assume that $1\ll d =O(\log n)$. Then, by applying Theorem \ref{thm:dai imp}, we can conclude that exact matching is impossible if $\frac{d}{4}\log \frac{1}{1-\rho^2} \leq (1-\epsilon) \log |V'| \leq (1-\epsilon) \log n$.

Wang et al. \cite{Wang22} also idenitfied the conditions under which exact matching is not possible, and the results are as follows.
\begin{theorem}[Theorem 19 in \cite{Wang22}]\label{thm:wang imp}
Consider the correlated Gaussian databases $X,Y \in \mathbb{R}^{n\times d}$. Suppose that $\frac{1}{\rho^2} -1 \leq \frac{d}{40}$ and 
\begin{equation}
    \frac{d}{4}\log \frac{1}{1-\rho^2} \leq \log n - \log d +C,
\end{equation}
for a constant $C>0$. Then, there is no estimator that can exactly recover $\pi_*$ with probability $1-o(1).$
\end{theorem}
Assume that $\frac{1}{\rho^2} -1 \leq \frac{d}{40}$. Then, by applying Theorem \ref{thm:wang imp}, we can conclude that exact matching is impossible if $\frac{d}{4}\log \frac{1}{1-\rho^2} \leq \log |V'| - \log d +C \leq  \log n - \log d +C $.

\section{Proof of Theorem \ref{thm:gmm recovery ach}:\\ Achievability of Exact Community Recovery in Correlated Gaussian Mixture Models}\label{sec:proof of gmm recovery ach}
Given a permutation $\pi : [n]\to [n]$, let $X+_\pi Y$ represent the database where each node $i$ is assigned the vector $\frac{\bsx_i + \bsy_{\pi(i)}}{2}$. Recall that for two functions $f,g$, $\mathbf{ov}(f,g):=\sum\limits_{i=1}^n \frac{\mathds{1}\left(f(i)=g(i)\right)}{n}$.
Ndaoud \cite{Ndaoud22} found the conditions under which there exists an estimator $\hat{\bssigma}$, which is based on a variant of Lloyd's iteration initialized by a spectral method, that can exactly recover $\bssigma$ from a Gaussian mixture model $X$. 
\begin{theorem}[Theorem 8 in \cite{Ndaoud22}]\label{thm:ndaoud ach}
    For $k>0$ and $\bsz_i \sim \mathcal{N}(0,\bsI_d)$, let $X:=\left\{\bsx_i\right\}_{i=1}^n \sim \operatorname{GMM}(\boldsymbol{\mu}, \bssigma)$, where $\boldsymbol{\mu}\in \mathbb{R}^d$ and $\bsx_i:= \boldsymbol{\mu} \sigma_i +k \bsz_i$. If 
    \begin{equation}
           \lVert \boldsymbol{\mu} \rVert^2 \geq k^2(1+\epsilon)\left(1+\sqrt{1+\frac{2d}{n \log n}}\right) \log n
    \end{equation}
    for an arbitrary small constant $\epsilon>0$, then there exists an estimator $\hat{\bssigma}$ achieving $\P(\mathbf{ov}(\hat{\bssigma},\bssigma)=1)=1-o(1)$.
\end{theorem}
When applying this estimator $\hat{\bssigma}$ to $X+_{\hat{\pi}} Y$ for the estimator $\hat{\pi}$ defined in \eqref{eq:estimator-1}, we can have
\begin{equation}\label{eq:recov eq1}
\begin{aligned}
    \P\left(\mathbf{ov}(\hat{\bssigma}(X+_{\hat{\pi}} Y),\bssigma) \neq 1\right) &\leq \P\left(\{\mathbf{ov}(\hat{\bssigma}(X+_{\hat{\pi}} Y),\bssigma) \neq 1\} \cap \{\hat{\pi}= \pi_*\}\right) +\P(\hat{\pi}\neq \pi_*)\\
     &\leq  \P\left(\mathbf{ov}(\hat{\bssigma}(X+_{\pi_*} Y),\bssigma) \neq 1\right) +\P(\hat{\pi}\neq \pi_*).
\end{aligned}
\end{equation}
If \eqref{eq:gmm matching cond1} or \eqref{eq:gmm matching cond2} holds, then by Theorem \ref{thm:gmm matching ach}, we can obtain that
\begin{equation}\label{eq:recov eq2}
    \P(\hat{\pi}\neq \pi_*)=o(1).
\end{equation}
Moreover, we have $\frac{\bsx_i+\bsy_{\pi_*(i)}}{2}= \boldsymbol{\mu}\bssigma_i + \frac{(1+\rho)\bsz_i+\sqrt{1-\rho^2}\bsw_i}{2}\sim \boldsymbol{\mu}\bssigma_i + \sqrt{\frac{1+\rho}{2}}\bso_i$, where $\bsz_i,\bsw_i,\bso_i\sim \mathcal{N}(0,\bsI_d)$. Therefore, by Theorem \ref{thm:ndaoud ach}, if $ \lVert \boldsymbol{\mu} \rVert^2 \geq (1+\epsilon)\frac{1+\rho}{2}\left(1+\sqrt{1+\frac{2d}{n \log n}}\right) \log n$, then it holds that 
\begin{equation}\label{eq:recov eq3}
    \P\left(\mathbf{ov}(\hat{\bssigma}(X+_{\pi_*} Y),\bssigma) \neq 1\right)=o(1).
\end{equation}
By combining the results from \eqref{eq:recov eq1}, \eqref{eq:recov eq2} and \eqref{eq:recov eq3}, the proof is complete.

\section{Proof of Theorem \ref{thm:csbm matching ach}: \\ Achievability of Exact Matching in Correlated Contextual Stochastic Block Models}\label{sec:proof of csbm matching ach}
To prove Theorem \ref{thm:csbm matching ach}, we consider a two-step procedure for exact matching. The first step utilizes the $k$-core matching based solely on edge information to recover the matching over $n-n^{1-\frac{ns^2(p+q)}{2\log n}+o(1)}$ nodes. Then, the second step utilizes the node features to match the rest $n^{1-\frac{ns^2(p+q)}{2\log n}+o(1)}$ nodes, where we apply the estimator \eqref{eq:estimator-1} used in proving Theorem \ref{thm:gmm matching ach}.

\subsection{The $k$-core matching and the proof of Theorem \ref{thm:csbm matching ach} }\label{sec:k-core}
The $k$-core matching has been extensively studied in recovering the latent vertex correspondence between the edge-correlated Erd\H{o}s-R\'enyi graphs or more general inhomogeneous random graphs including the correlated Stochastic Block Models  \cite{CKNP20,GRS22,RS23,YC24}. In this subsection, we will demonstrate that the analytical techniques developed in the previous papers can be effectively applied to the general correlated SBMs we consider in this paper.

We restate the definitions of matching, $k$-core matching and $k$-core estimator, introduced in \cite{CKNP20}, for the completeness of the paper.

\begin{definition}[Matching]\label{def:matching}
    Consider two graphs $G_1$ and $G_2$. $(M,\varphi)$ is a matching between $G_1$ and $G_2$ if $M\subset [n]$ and $\varphi $ : $M \to [n]$ is injective. For a matching $(M,\varphi)$, we define $\varphi(M)$ as the image of $M$ under $\varphi$, and $\varphi\{M\}:=\{(i,\varphi(i)) : i \in [M]\}$.
\end{definition}

Consider two graphs $G_1$ and $G_2$ with a matching $(M,\varphi)$. We define the intersection graph $G_1 \wedge_{\varphi} G_2$ as follows:
\begin{itemize}
    \item For $u,v\in M$, $(u,v)$ is an edge in $G_1 \wedge_{\varphi} G_2$ if and only if $(u,v)$ is an edge in $G_1$ and $(\varphi(u),\varphi(v))$ is an edge in $G_2$.
\end{itemize}

\begin{definition}[$k$-core matching and $k$-core estimator]\label{def:k-core estimator}
    Consider two graphs $G_1$ and $G_2$. A matching $(M,\varphi)$ is a $k$-core matching if $d_{\min}(G_1\wedge_{\varphi} G_2)\geq k$. Furthermore, the $k$-core estimator $(\widehat{M}_k,\widehat{\varphi}_k)$ is the $k$-core matching that includes the largest nodes among all the $k$-core matchings.
\end{definition}

By using the $k$-core estimator, we can achieve a partial matching with no mismatched node pairs, as stated in the following theorem. 
\begin{theorem}[Partial matching achievable by the $k$-core estimator]\label{thm:k-core matching}
    Consider the correlated Stochastic Block Models with two communities $(G_1,G_2)\sim \textnormal{CSBMs}(n,p,q,s)$. Suppose that
      \begin{equation}\label{eq:k-core p regmie}
    p\leq O\left(\frac{1}{e^{(\log \log n)^3}}\right) \text{ and}
\end{equation}
      \begin{equation}\label{eq:k-core k regmie}
   k=\frac{nps^2}{(\log nps^2)^2} \vee \frac{\log n}{(\log \log n)^2}.
\end{equation}
Then, the $k$-core estimator $(\widehat{M}_k,\widehat{\varphi}_k)$ satisfies that 
\begin{equation}\label{eq:k-core size}
    |\widehat{M}_k|\geq n-n^{1-\frac{ns^2(p+q)}{2\log n}+o(1)} \text{ and}
\end{equation}
    \begin{equation}\label{eq:k-core groundtruth}
    \widehat{\varphi}_k\{\widehat{M}_k\}=\pi_*\{\widehat{M}_k\}
    \end{equation}
  with probability $1-o(1)$.
\end{theorem}
Furthermore, the matched node set $\widehat{M}_k$ is equivalent to the $k$-core set of the graph $G_1 \wedge_{\pi_*} G_2$. A detailed explanation and the related lemmas for the $k$-core matching can be found in Section \ref{sec:lemma k core}.
We also state the sufficient conditions for exact matching in the correlated Stochastic Block Models achievable by the $k$-core estimator.

\begin{theorem}[Exact matching achievable by the $k$-core estimator]\label{thm:exact matching}
     Consider the correlated Stochastic Block Models with two communities $(G_1,G_2)\sim \textnormal{CSBMs}(n,p,q,s)$. Suppose that $\eqref{eq:k-core p regmie}$ and \eqref{eq:k-core k regmie} hold. Also, assume that
    \begin{equation}\label{eq:exact matching}
        ns^2 \frac{p+q}{2}\geq (1+\epsilon) \log n
    \end{equation}
    for an arbitrary small constant $\epsilon>0$. Then, $\widehat{\varphi}_{k}=\pi_*$ with probability $1-o(1)$.
\end{theorem}

The proofs of Theorem~\ref{thm:k-core matching} and \ref{thm:exact matching} will be presented in Section \ref{sec:proof k core}.

\begin{IEEEproof}[Proof of Theorem \ref{thm:csbm matching ach}]
    Let $ k=\frac{nps^2}{(\log nps^2)^2} \vee \frac{\log n}{(\log \log n)^2}$. Let us first consider the case where $ns^2\left(\frac{p+q}{2}\right)\geq (1+\epsilon)\log n$ for an arbitrary small constant $\epsilon >0$. By Theorem \ref{thm:exact matching}, we can confirm that the exact matching is achievable through the $k$-core matching using only edge information, without relying on node information.
    
    Now, let us consider the case where $nps^2=O(\log n)$. First, by Theorem \ref{thm:k-core matching}, we can obtain a matching $(\widehat{M}_k,\widehat{\varphi}_k)$ that satisfies \eqref{eq:k-core size} and $\eqref{eq:k-core groundtruth}$ through the $k$-core matching. Let $F$ denote the set of nodes that remain unmatched after performing the $k$-core matching. That is, $F:=[n]\backslash \widehat{M}_k$. From \eqref{eq:k-core size}, we have 
    \begin{equation}\label{eq:size F}
        |F|\leq n^{1-\frac{ns^2(p+q)}{2\log n}+o(1)}.
    \end{equation}
    From \eqref{eq:size F} and \eqref{eq:csbm matching cond1}, we can also have
    \begin{equation}
        \frac{d}{4}\log \frac{1}{1-\rho^2}\geq (1+\epsilon) \log n - ns^2\frac{p+q}{2} \geq \left(1+\epsilon/2 \right) \log |F|
    \end{equation}
    for a sufficiently large $n$. Therefore, if   $\lVert \boldsymbol{\mu} \rVert^2 \geq 2\log n+\omega(1)$ or $ d =\omega (\log n)$, then the exact matching is possible between the unmatched nodes belonging to $F$ solely using the node features by Theorem \ref{thm:gmm matching ach}. Thus, the proof is complete.
\end{IEEEproof}

\subsection{Lemmas for the analysis of $k$-core matching}\label{sec:lemma k core}

For a matching $(M,\varphi)$, define 
\begin{equation}
    f(M,\varphi)=\Sigma_{i\in M : \varphi(i)\neq \pi_*(i)} \deg_{G_1 \wedge_{\varphi} G_2}(i).
\end{equation}
The weak $k$-core matching,  $\pi_*$-maximal matching and $k$-core set were defined in \cite{CKNP20,GRS22,RS23,YC24}. For the completeness, we present the corresponding definitions below.
\begin{definition}[Weak $k$-core matching]
    We say that a matching $(M,\varphi)$ is a weak $k$-core matching if 
    \begin{equation}
        f(M,\varphi)\geq k |\{i\in M : \varphi(i)\neq \pi_*(i)\}|.
    \end{equation}
    
    \end{definition}

\begin{definition}[$\pi_*$-maximal matching]
    We say that a matching $(M,\varphi)$ is $\pi_*$-maximal if for every $i \in [n]$, either $i\in M$ or $\pi_*(i) \in \varphi(M)$, where $\varphi(M)$ is the image of $M$ under $\varphi$. Additionally, let us define $\mathcal{M}(t)$ as the set of $\pi_*$-maximal matchings that have $t$ errors. It means that
    \begin{equation}
    \begin{aligned}
         \mathcal{M}(t):=\{(M,\varphi) : (M,\varphi) \text{ is } \pi_* \text{-maximal and } |\{i\in M : \varphi(i)\neq \pi_*(i)\}| =t \}.
    \end{aligned}       
    \end{equation}
\end{definition}

\begin{definition}[$k$-core set]
    For a graph $G$, a vertex set $M$ is referred to as the $k$-core of $G$ if it is the largest set such that $d_{\min}(G\{M\})\geq k$.
\end{definition}

Let $M_k$ denote the $k$-core of $G_1 \wedge_{\pi_*} G_2$. The following lemma provides a lower bound on the probability to obtain the $k$-core of $G_1 \wedge_{\pi_*} G_2$ and the correct matching over the $k$-core as the result of $k$-core matching. 
Since Gaudio et al. \cite{GRS22} presented this result for general pairs of random graphs in Corollary 20, it can also be applied to the correlated SBMs.

\begin{lemma}[Corollary 4.7 in \cite{GRS22}]\label{lem:k-core lem1}
    Consider the $(G_1,G_2)\sim \textnormal{CSBMs}(n,p,q,s)$. For any positive integer $k$, define the quantity
\begin{equation}\label{def:xi}
    \xi:=\max _{1 \leq t \leq n} \max _{(M, \varphi) \in \mathcal{M}(t)} \mathbb{P}(f(M,\varphi) \geq k t)^{1 / t} .
\end{equation}
Then, the $k$-core estimator $(\widehat{M}_k, \widehat{\varphi}_k)$ satisfies that
\begin{equation}
    \begin{aligned}
\mathbb{P}\left(\widehat{M}_k=M_k \text { and } \widehat{\varphi}_k\{\widehat{M}_k\}\right. & \left.=\pi_*\{\widehat{M}_k\}\right) \geq 2-\exp \left(n^2 \xi\right).
\end{aligned}
\end{equation}
\end{lemma}
From the above lemma, we can see that if $\xi=o(n^{-2})$, the matching obtained through the $k$-core estimator is the correct matching over the $k$-core of $G_1 \wedge_{\pi_*} G_2$ with high probability.
The next lemma represents an upper bound on $\xi$. This result has been proven in \cite{GRS22,YC24,RS23}, and in particular \cite{RS23} provides a proof for general random graph pairs in Lemma A.4. We obtain the similar results in the correlated Stochastic Block Models.
\begin{lemma}[Lemma A.4 in \cite{RS23}]\label{lem:k-core lem2}
 Consider the $(G_1,G_2)\sim \textnormal{CSBMs}(n,p,q,s)$.
    For any matching $(M, \varphi) \in \mathcal{M}(t)$ and any $\theta>0$, we have that
\begin{equation}
\begin{aligned}
    \mathbb{P}(f(M,\varphi) & \geq k t)  \leq 3 \exp \left(-t\left(\theta k-e^{2 \theta} ps^2-n e^{6 \theta} p^2s^2\right)\right).
\end{aligned}
\end{equation}
\end{lemma}

Now, let us analyze the size of the matched node set.
Let us introduce a good event that will be useful in our analysis.
\begin{definition}[Balanced communities]\label{def:balanced}

\begin{equation}
       \mathcal{B}:=\left\{ \frac{n}{2}-n^{2/3} \leq |V^+|,|V^-|\leq \frac{n}{2}+n^{2/3}\right\}.
\end{equation} 
\end{definition}
When $ + $ or $ - $ label is assigned to each node with equal probabilities,  the probability of the event $\mathcal{B}$ is $1-o(1)$. 
\begin{lemma}\label{lem:Balanced comm}
    It holds that $\P(\mathcal{B})=1-o(1).$
\end{lemma}

\begin{lemma}\label{lem:size degree k}
    Let $G\sim \textnormal{SBM}(n,p,q)$. Define the set 
    \begin{equation}\label{def:Lk}
         L_k := \{i\in [n] : \deg_G(i) \leq k\}.
    \end{equation}    
    Then, on the event $\mathcal{B}$, it holds that
    \begin{equation}
        \E[|L_k|]\leq n  \exp\left(-n\frac{p+q}{2} +o\left(n\frac{p+q}{2}\right)+ k \log np+1\right).
    \end{equation}
\end{lemma}

Using the Luczak expansion \cite{Luczak}, the size of the $k$-core of the graph $G$ can be obtained as follows. In \cite{RS23}, this result was proven for a general case.

\begin{lemma}[Lemma IV.6 in \cite{RS23}]\label{lem:3L}
     Let $G\sim \textnormal{SBM}(n,p,q)$ with $np,nq=O(\log n)$. Let $J_k$ denote the $k$-core of $G$, and let $\bar{J}_k:=[n]\backslash J_k$.  Recall $L_k$ defined in \eqref{def:Lk}. If $|L_{k}|\leq n^c$ for $c\in (0,1)$, then $|\bar{J}_k | \leq 3|L_{k}|$.
\end{lemma}

\subsection{Proofs of Theorems~\ref{thm:k-core matching} and \ref{thm:exact matching} and Lemmas \ref{lem:Balanced comm} and \ref{lem:size degree k}}\label{sec:proof k core}

Based on the lemmas above, we can now prove Theorem \ref{thm:k-core matching} as follows.

\begin{IEEEproof}[Proof of Theorem \ref{thm:k-core matching}]
    If $ns^2\frac{p+q}{2}\geq (1+\epsilon)\log n$ for an arbitrary small constant $\epsilon>0$, then we can obtain $\widehat{\varphi}_k=\pi_*$ by Theorem \ref{thm:exact matching}. 
Now, let us consider the case where $nps^2,nqs^2=O(\log n)$.
    By Lemma \ref{lem:k-core lem2} and the definition of $\xi$ in \eqref{def:xi}, we have that
    \begin{equation}\label{eq:xi upper}
        \begin{aligned}
            \xi &\leq 3^{1/t} \exp \left(-\theta k+e^{2 \theta} ps^2+n e^{6 \theta} p^2s^2\right) \leq 3 \exp \left(-\theta k+e^{2 \theta} ps^2+n e^{6 \theta} p^2s^2\right)
        \end{aligned}
    \end{equation}
     for any $\theta>0$. Through Lemma \ref{lem:k-core lem1} and \eqref{eq:xi upper}, we can show that if there exists an $\theta>0$ such that $\theta k-e^{2 \theta} ps^2-n e^{6 \theta} p^2s^2 \geq 2 \log n +\omega(1)$, then $\widehat{M}_k=M_k \text { and } \widehat{\varphi}_k\{\widehat{M}_k\}=\pi_*\{\widehat{M}_k\}$ hold with probability $1-o(1)$. Let $\theta=(\log \log n)^{2.5}$ and recall that $p\leq O\left(e^{-(\log \log n)^3}\right)$ and $k=\frac{nps^2}{(\log nps^2)^2} \vee \frac{\log n}{(\log \log n)^2}$. Then, we can have that
     \begin{equation}
         \theta k \geq \log n \cdot (\log \log n)^{0.5}
     \end{equation}
and 
\begin{equation}
    e^{2\theta} ps^2 \leq  e^{2\theta} p =o(1).
\end{equation}
     Moreover, we can obtain
     \begin{equation}
         \theta k \geq \theta \frac{nps^2}{(\log nps^2)^2} \geq  (\log \log n)^{2.5}\frac{nps^2}{(\log n)^2} \geq 2ne^{6\theta}p^2s^2.
     \end{equation}
     The last inequality holds by $p\leq O\left(e^{-(\log \log n)^3}\right)$. Therefore, we have 
     \begin{equation}
         \theta k-e^{2 \theta} ps^2-n e^{6 \theta} p^2s^2 \geq \frac{1}{3}\theta k = \omega(\log n). 
     \end{equation}

    Now, we will prove that $|M_k|\geq n-n^{1-\frac{ns^2(p+q)}{2\log n}+o(1)}$.  If $ns^2(p+q)=o(\log n)$, the right-hand side converges to 0, making the result trivial. So, let us consider the case where $ns^2(p+q)=\Theta(\log n)$.    
    Recall that $M_k$ is the $k$-core of $G_1\wedge_{\pi_*}G_2$. Let $\overline{L}_k:=\{i\in [n] : \deg_{G_1\wedge_{\pi_*}G_2}(i) \leq k\}$ and $\overline{M}_k:=[n]\backslash M_k$.   
    Since $G_1\wedge_{\pi_*}G_2 \sim \textnormal{SBM}(n,ps^2,qs^2)$, combining Lemma \ref{lem:Balanced comm} and Lemma \ref{lem:size degree k} allows us to conclude that 
    \begin{equation}\label{eq:L1}
        \E[|\overline{L}_k|]\leq   n\exp\left(-ns^2\frac{p+q}{2} +o\left(ns^2\frac{p+q}{2}\right)+ k \log nps^2+1\right)
    \end{equation}
    with probability $1-o(1)$. By Markov's inequality, we can also obtain that 
 \begin{equation}\label{eq:L2}
        |\overline{L}_k| \leq (\log n) \E [|\overline{L}_k|]
    \end{equation}
with probability at least $1-\frac{1}{\log n}$. By applying Lemma \ref{lem:3L} with \eqref{eq:k-core k regmie}, \eqref{eq:L1} and \eqref{eq:L2}, we can obtain
\begin{equation}
    \overline{M}_k \leq 3|\overline{L}_k| \leq n^{1-\frac{ns^2(p+q)}{2\log n}+o(1)}. 
\end{equation}
Hence, the proof is complete.
\end{IEEEproof}

\begin{IEEEproof}[Proof of Theorem \ref{thm:exact matching}]
    
For the same reasons as in the proof of Theorem \ref{thm:k-core matching}, we obtain that $\widehat{M}_k=M_k \text { and } \widehat{\varphi}_k\{\widehat{M}_k\}=\pi_*\{\widehat{M}_k\}$ with probability $1-o(1)$. Thus, if we can show that $|M_k|=n$ with probability $1-o(1)$, the proof will be complete.

On the event $\mathcal{B}$, we can obtain 
   \begin{equation}\label{eq:Mk=n}
        \begin{aligned}
            \P(|M_k|\neq n)&=\P(d_{\min}(G_1\wedge_{\pi_*}G_2)<k)\\
            & \stackrel{(a)}{\leq} n\P(\deg_{G_1\wedge_{\pi_*}G_2}(i)<k)\\
            & \stackrel{(b)}{\leq} n  \exp\left(-ns^2\frac{p+q}{2} +o\left(ns^2\frac{p+q}{2}\right)+ k \log nps^2+1\right)\\
            & \stackrel{(c)}{=} o(1)            .
        \end{aligned}
    \end{equation}
    The inequality $(a)$ holds by taking union bound over $i\in [n]$, the inequality $(b)$ holds by Lemma \ref{lem:size degree k}, and the equality $(c)$ holds by \eqref{eq:k-core k regmie} and \eqref{eq:exact matching}. Therefore, combining \eqref{eq:Mk=n} and Lemma \ref{lem:Balanced comm}, the proof is complete.
\end{IEEEproof}

\begin{IEEEproof}[Proof of Lemma \ref{lem:Balanced comm}]
    We know that $|V^+| \sim \operatorname{Bin}(n,1/2)$. By applying Hoeffding's inequality (stated in Lemma \ref{lem:hoeff}), we can obtain that
    \begin{equation}
        \P\left( \left| |V^+|-n/2  \right| \geq n^{2/3} \right) \leq 2\exp(-2n^{1/3}).
    \end{equation}
     Since $|V^-|$ also follows a $\operatorname{Bin}(n,1/2)$, the proof is complete.
\end{IEEEproof}

\begin{IEEEproof}[Proof of Lemma \ref{lem:size degree k}]
 Let $n_1 := |V^+|$ and $n_2:=|V^-|$. On the event $\mathcal{B}$, we can obtain that 
 \begin{equation}\label{eq:n1n2}
     \frac{n}{2}-n^{2/3} \leq n_1, n_2\leq \frac{n}{2}+n^{2/3}.
 \end{equation} 
For $i\in V^{+}$, it holds that $\deg_G(i)=\sum^{n_1-1}_{j=1}X_j+\sum^{n_2}_{j=1}Y_j$, where $X_j\sim \operatorname{Bern}(p)$ and $Y_j\sim \operatorname{Bern}(q)$. Therefore, we can obtain that 
    \begin{equation}
        \begin{aligned}
              \P(\deg_G(i) \leq k)&=\P\left(\sum^{n_1-1}_{j=1}X_j+\sum^{n_2}_{j=1}Y_j\leq k \right)\\
            & \stackrel{(a)}{\leq} \inf_{t>0}(1-p+pe^{-t})^{n_1-1}(1-q+qe^{-t})^{n_2} e^{kt}\\
            &\stackrel{(b)}{\leq} \inf_{t>0} \exp\left(-(n_1-1)p(1-e^{-t})-n_2q(1-e^{-t})+kt \right)\\
            &\stackrel{(c)}{\leq}  \exp\left(-n\frac{p+q}{2} +o\left(n\frac{p+q}{2}\right)+ k \log np+1\right).
        \end{aligned}
    \end{equation}
    The inequality $(a)$ holds by Chernoff bound, the inequality $(b)$ holds since $1-x\geq e^{-x}$, and the inequality $(c)$ holds by \eqref{eq:n1n2} and choosing $t= \log np$.
    Since the same result can be obtained for $i\in V^{-}$ as well, the proof is complete.
\end{IEEEproof}

\section{Proof of Theorem \ref{thm:csbm matching imp}: \\ Impossibility of Exact Matching in Correlated Contextual Stochastic Block Models}\label{sec:proof of csbm matching imp}

To prove Theorem \ref{thm:csbm matching imp}, we will analyze the MAP estimator and find the conditions where the MAP estimator fails. The posterior distribution of the permutation $\pi \in S_n$ in the correlated SBMs was analyzed in \cite{YC23,RS21,RS21supp,GRS22}. We will introduce this result and the related lemmas and explain how these previous results can be utilized to prove Theorem \ref{thm:csbm matching imp}.  

Given community labels $\bssigma_2$, for $a,b\in \{0,1\}$, define
\begin{align*}
        \psi^+_{ab}(\pi)&:=\sum_{(\pi(i),\pi(j)) \in \mathcal{E}^+(\bssigma_2)} \mathds{1}\{(A_{i,j},B_{\pi(i),\pi(j)})=(a,b) \},\\
          \psi^-_{ab}(\pi)&:=\sum_{(\pi(i),\pi(j)) \in \mathcal{E}^-(\bssigma_2)} \mathds{1}\{(A_{i,j},B_{\pi(i),\pi(j)})=(a,b) \},\\
          \chi^+(\pi)&:=\sum_{(\pi(i),\pi(j)) \in \mathcal{E}^+(\bssigma_2)} B_{\pi(i),\pi(j)},\\
           \chi^-(\pi)&:=\sum_{(\pi(i),\pi(j)) \in \mathcal{E}^+(\bssigma_2)} B_{\pi(i),\pi(j)},
    \end{align*}
where $\mathcal{E}^{+}(\bssigma_2)$ represents the set of node pairs belonging to the same community in the graph $G_2$, while $\mathcal{E}^{-}(\bssigma_2)$ represents the set of node pairs belonging to different communities in the graph $G_2$. 

We will also use the notation for the edge probabilities as follows.
\begin{equation*}
    p_{ab}:=\P\left((A_{i,j},B_{\pi_*(i),\pi_*(j)})=(a,b) | \bssigma \right)=\begin{cases} 
		ps^2 & \text{if } (a,b)=(1,1) \text{ and } \sigma_i=\sigma_j; \\ 
        ps(1-s) & \text{if } (a,b)=(1,0),(0,1) \text{ and } \sigma_i=\sigma_j; \\
        1-2ps+ps^2 &\text{if } (a,b)=(0,0) \text{ and } \sigma_i=\sigma_j.
     \end{cases}
\end{equation*}  
\begin{equation*}
      q_{ab}:=\P\left((A_{i,j},B_{\pi_*(i),\pi_*(j)})=(a,b) | \bssigma \right)=\begin{cases} 
		qs^2 & \text{if } (a,b)=(1,1) \text{ and } \sigma_i \neq \sigma_j; \\ 
        qs(1-s) & \text{if } (a,b)=(1,0),(0,1) \text{ and } \sigma_i \neq \sigma_j;\\
        1-2qs+qs^2 &\text{if } (a,b)=(0,0)\text{ and } \sigma_i \neq \sigma_j.
     \end{cases}
\end{equation*}  

That is, $p_{ab}$   represents the edge probability between nodes within the same community, while $q_{ab}$ represents the edge probability between nodes in different communities. 
Given $A,B$ and $\bssigma_2$, the posterior distribution for $\pi$ can be written as follows:

\begin{lemma}[Lemma 49 in \cite{GRS22}]\label{lem:posterior}
Let $\pi \in S_n$. Then, we have
\begin{equation}
    \mathbb{P}\left(\pi_*=\pi \mid A, B, \bssigma_2 \right)=c\left(\frac{p_{00} p_{11}}{p_{01} p_{10}}\right)^{\psi^{+}_{11}(\pi)}\left(\frac{q_{00} q_{11}}{q_{01} q_{10}}\right)^{\psi^{-}_{11}(\pi)}\left(\frac{p_{01}}{p_{00}}\right)^{\chi^{+}(\pi)}\left(\frac{q_{01}}{q_{00}}\right)^{\chi^{-}(\pi)},
\end{equation}
where $c$ is a constant depending on $A,B$ and $\bssigma_2$.
\end{lemma}
For a permutation $\pi \in S_n$, let us define
\begin{equation}
    \mathcal{H}(\pi):=\left\{i \in[n]: \forall j \in[n], A_{i, j} B_{\pi(i), \pi(j)}=0\right\},
\end{equation}
where $A$ and $B$ are the adjacency matrices of the graphs $G_1$ and $G_2$. 
Additionally, let $\mathcal{H}(\pi)^+ := \mathcal{H}(\pi) \cap V^+$ and $\mathcal{H}(\pi)^- := \mathcal{H}({\pi}) \cap V^-$. To simplify the notation, let the node sets corresponding to $\pi_*$  be denoted by $\mathcal{H}_*,\mathcal{H}_*^+$ and $\mathcal{H}_*^-$.

Now, we will derive a lower bound for $|\mathcal{H}_*^+|$ and $|\mathcal{H}_*^-|$. Recall that
\begin{equation}
       \mathcal{B}:=\left\{ \frac{n}{2}-n^{2/3} \leq |V^+|,|V^-|\leq \frac{n}{2}+n^{2/3}\right\}
\end{equation} 
as defined in Definition \ref{def:balanced}.

\begin{lemma}\label{lem:size H}
    Suppose that 
    the event $\mathcal{B}$ holds. If $ns^2\frac{p+q}{2} =O(\log n)$, then it holds that 
    \begin{equation}
        |\mathcal{H}_*^+|,|\mathcal{H}_*^-| \geq \frac{1}{8} n^{1-\frac{ns^2(p+q)}{2 \log n}}
    \end{equation}
    with probability $1-o(1)$.
\end{lemma}

We define the set of permutations $\mathcal{T}_*$ as follows. The permutations belonging to  $\mathcal{T}_*$ will have a posterior probability greater than or equal to that of $\pi_*$

\begin{definition}\label{def:T}
  A permutation  $\pi\in \mathcal{T}_*$ if and only if the following conditions hold: 
\begin{itemize}
    \item $\pi(i)=\pi_*(i)$  \; \quad \ if $i\in  [n] \backslash \mathcal{H}_*$
    \item $\pi(i)=\pi_*(\rho^+(i))$  if $i \in \mathcal{H}_*^+$,
     \item $\pi(i)=\pi_*(\rho^-(i))$  if $i \in \mathcal{H}_*^-$,
\end{itemize}
    where $\rho^+$ and $\rho^-$ are any permutations over $\mathcal{H}^+_*$ and $\mathcal{H}^-_*$, respectively.
\end{definition}

The permutations in set $\mathcal{T}_*$ satisfy the following properties: 
\begin{lemma}[Proposition C.2 and C.3 in \cite{RS21supp}]\label{lem:psi chi}
    For any permutation $\pi \in \mathcal{T}_*$, we have that $\psi^+_{11}(\pi)\geq \psi^+_{11}(\pi_*)$, $\psi^-_{11}(\pi)\geq \psi^-_{11}(\pi_*)$, $\chi^+(\pi)=\chi^+(\pi_*)$ and $\chi^-(\pi)=\chi^-(\pi_*)$.
\end{lemma}

\begin{IEEEproof}   
We consider the posterior distribution of $\pi\in {S}_n$ given $G_1,G_2$, along with the additional side information $\pi_*\left\{[n] \backslash \mathcal{H}_*\right\}, \mathcal{H}_*^+,\mathcal{H}_*^-,$ and $ \bssigma_2$. Then, the MAP estimator can be expressed as follows :
    \begin{equation}
        \widehat{\pi}_{\text {MAP }}:=\argmax_{\pi\in S_n} \P\left(\pi_*=\pi \mid G_1, G_2, \pi_*\left\{[n] \backslash \mathcal{H}_*\right\}, \mathcal{H}_*^+,\mathcal{H}_*^-, \bssigma_2\right).
    \end{equation}

Recall that 
$A$ and $B$ are the adjacency matrices of $G_1$ and $G_2$, respectively, and  $X$ and $Y$ are the corresponding database matrices.
We can obtain that
\begin{equation}\label{eq:condition poster}
    \begin{aligned}
       &\P\left(\pi_*=\pi \mid G_1, G_2, \pi_*\left\{[n] \backslash \mathcal{H}_*\right\}, \mathcal{H}_*^+,\mathcal{H}_*^-,  \bssigma_2\right) \\
         & =\frac{\mathbb{P}\left(\pi_*=\pi \mid G_1, G_2, \bssigma_2\right) \mathbb{P}\left( \pi_*\left\{[n] \backslash \mathcal{H}_*\right\}, \mathcal{H}_*^+,\mathcal{H}_*^-, \mid \pi_*=\pi, G_1, G_2, \bssigma_2\right)}{\mathbb{P}\left(\pi_*\left\{[n] \backslash \mathcal{H}_*\right\}, \mathcal{H}_*^+,\mathcal{H}_*^- \mid G_1, G_2 , \bssigma_2\right)} \\
& \stackrel{(a)}{=} \frac{\mathbb{P}\left(\pi_*=\pi \mid A, B, X, Y,\bssigma_2\right)}{\mathbb{P}\left(\pi_*\left\{[n] \backslash \mathcal{H}_*\right\}, \mathcal{H}_*^+,\mathcal{H}_*^- \mid G_1, G_2,   \bssigma_2\right)} \mathds{1}\left(\pi \in \mathcal{T}_*\right) \\
& \stackrel{(b)}{=} C_1 \frac{\mathbb{P}\left(\pi_*=\pi \mid A, B,  \bssigma_2\right) \mathbb{P}\left(\pi_*=\pi \mid X, Y,  \bssigma_2\right)}{\mathbb{P}\left(\pi_*=\pi\mid \bssigma_2\right) {\mathbb{P}\left(\pi_*\left\{[n] \backslash \mathcal{H}_*\right\}, \mathcal{H}_*^+,\mathcal{H}_*^- \mid G_1, G_2,   \bssigma_2\right)}} \mathds{1}\left(\pi \in \mathcal{T}_*\right) \\
& \stackrel{(c)}{=}C_2 \mathbb{P}\left(\pi_*=\pi \mid A, B, \bssigma_2\right) \mathbb{P}\left(\pi_*=\pi \mid X, Y, \bssigma_2\right) \mathds{1}\left(\pi \in \mathcal{T}_*\right) \\
&  \stackrel{(d)}{=}C_3 \mathbb{P}\left(\pi_*=\pi \mid X, Y, \bssigma_2\right)\left(\frac{p_{00} p_{11}}{p_{01} p_{10}}\right)^{\psi^{+}_{11}(\pi)}\left(\frac{q_{00} q_{11}}{q_{01} q_{10}}\right)^{\psi^{-}_{11}(\pi)}\left(\frac{p_{01}}{p_{00}}\right)^{\chi^{+}(\pi)}\left(\frac{q_{01}}{q_{00}}\right)^{\chi^{-}(\pi)} \mathds{1}\left(\pi \in \mathcal{T}_*\right),
    \end{aligned}
\end{equation}
where $C_1,C_2,$ and $C_3$ are constants that do not depend on $\pi$. The equality $(a)$ holds from Definition \ref{def:T}, the equality $(b)$ holds since
\beq
\begin{split}
&\mathbb{P}\left(\pi_*=\pi \mid A, B, X, Y,  \bssigma_2\right)\\
&=\frac{\mathbb{P}\left(  \pi_*=\pi \mid  \bssigma_2\right) \mathbb{P}\left( A, B, X, Y\mid \pi_*=\pi,  \bssigma_2\right)}{\mathbb{P}\left( A, B, X, Y\mid  \bssigma_2\right)}\\
&=\frac{\mathbb{P}\left(  \pi_*=\pi \mid  \bssigma_2\right) \mathbb{P}\left( A, B\mid \pi_*=\pi,  \bssigma_2\right)\mathbb{P}\left(X, Y\mid \pi_*=\pi,  \bssigma_2\right)}{\mathbb{P}\left( A, B, X, Y\mid  \bssigma_2\right)}\\
&=\frac{\mathbb{P}\left(  \pi_*=\pi \mid  \bssigma_2\right) \mathbb{P}\left( A, B\mid  \bssigma_2\right)\mathbb{P}\left(  \pi_*=\pi\mid A, B,  \bssigma_2\right)
\mathbb{P}\left(X, Y\mid  \bssigma_2\right) \mathbb{P}\left( \pi_*=\pi\mid X, Y,  \bssigma_2\right) }{\mathbb{P}\left( A, B, X, Y\mid  \bssigma_2\right)\mathbb{P}\left(  \pi_*=\pi \mid  \bssigma_2\right)\mathbb{P}\left(  \pi_*=\pi \mid  \bssigma_2\right)}\\
&=C_1\frac{\mathbb{P}\left( \pi_*=\pi\mid A, B,  \bssigma_2\right) \mathbb{P}\left( \pi_*=\pi\mid X, Y,  \bssigma_2\right)}{\mathbb{P}\left(  \pi_*=\pi \mid  \bssigma_2\right)},
\end{split}
\eeq
the equality $(c)$ holds since $\mathbb{P}\left(\pi_*=\pi\mid \bssigma_2\right)=1/n!$ and ${\mathbb{P}\left(\pi_*\left\{[n] \backslash \mathcal{H}_*\right\}, \mathcal{H}_*^+,\mathcal{H}_*^- \mid G_1, G_2,   \bssigma_2\right)}$ does not depend on $\pi$, and the equality $(d)$ holds by Lemma \ref{lem:posterior}. 

Let $K(\pi):=\left(\frac{p_{00} p_{11}}{p_{01} p_{10}}\right)^{\psi^{+}_{11}(\pi)}\left(\frac{q_{00} q_{11}}{q_{01} q_{10}}\right)^{\psi^{-}_{11}(\pi)}\left(\frac{p_{01}}{p_{00}}\right)^{\chi^{+}(\pi)}\left(\frac{q_{01}}{q_{00}}\right)^{\chi^{-}(\pi)}$.
By Lemma \ref{lem:psi chi}, we can obtain that for $\pi \in \mathcal{T}_*$,
\begin{equation}\label{eq:Kpi}
     K(\pi)\geq K(\pi_*).
\end{equation}

Assume that the event $\mathcal{B}$ holds. First, suppose that \eqref{eq:csbm matching imp1} holds. By Lemma \ref{lem:size H}, we can obtain
\begin{equation}\label{eq:H size}
      |\mathcal{H}_*^+|,|\mathcal{H}_*^-| \geq  \frac{1}{8} n^{1-\frac{ns^2(p+q)}{2 \log n}}=\omega(1)
\end{equation}
with probability $1-o(1)$.
Therefore, $|\mathcal{T}_*|=\omega(1)$ also holds. Additionally, we can have
\begin{equation}
\begin{aligned}
     \frac{d}{4} \log \frac{1}{1-\rho^2} &\leq (1-\epsilon) \log n -ns^2\frac{p+q}{2} \leq  \left(1-\frac{\epsilon}{2}\right) \log |\mathcal{H}^+_*|.
\end{aligned}
\end{equation}
Thus, by combining Theorem \ref{thm:gmm matching imp} with \eqref{eq:condition poster} and \eqref{eq:Kpi}, we can conclude that the MAP estimator fails.

Conversely, suppose that \eqref{eq:csbm matching imp2} holds. It holds that $|\mathcal{T}_*|=\omega(1)$ due to \eqref{eq:H size}. Additionally, we can have
\begin{equation}
\begin{aligned}
     \frac{d}{4} \log \frac{1}{1-\rho^2} &\leq \log n -ns^2\frac{p+q}{2}-\log d -\omega(1) \leq \log |\mathcal{H}^+_*| - \log d .
\end{aligned}
\end{equation}
Thus, by combining Theorem \ref{thm:gmm matching imp} with \eqref{eq:condition poster} and \eqref{eq:Kpi}, we can conclude that the MAP estimator fails.
Therefore we can obtain that
\begin{equation}
    \P(\text{MAP estimator fails})\geq \P(\text{MAP estimator fails}|\mathcal{B})\P(\mathcal{B}) =1-o(1)
\end{equation}
since $\P(\mathcal{B})=1-o(1)$ by Lemma \ref{lem:Balanced comm}.
\end{IEEEproof}

\subsection{Proof of Lemma \ref{lem:size H}}

\begin{IEEEproof}[Proof of Lemma \ref{lem:size H}]
    Let $n_1=|V^+|$ and $n_2=|V^-|$, and let $\mathcal{H}^+_*$ and $\mathcal{H}^-_*$ represent the isolated nodes in $G_1\wedge_{\pi_*}G_2 \sim \text{SBM}(n,ps^2,qs^2)$ belonging to $V^+$ and $V^-$, respectively.
    On the event $\mathcal{B}$ in Definition \ref{def:balanced}, we have 
    \begin{equation}
         \frac{n}{2}(1-o(1))\leq n_1,n_1\leq \frac{n}{2}(1+o(1)).
    \end{equation}   
    We can obtain that
    \begin{equation}\label{eq:expect}
    \begin{aligned}
           \E(|\mathcal{H}^+_*|)&=n_1(1-ps^2)^{n_1-1}(1-qs^2)^{n_2}.
    \end{aligned}     
    \end{equation}
Let $I_i$ be the indicator variable representing the event that the vertex $i$ is isolated. Then, we have
\begin{equation}\label{eq:variance}
\begin{aligned}
    \operatorname{Var}(|\mathcal{H}^+_*|)&=\E(|\mathcal{H}^+_*|^2)-\E(|\mathcal{H}^+_*|)^2\\
    &=\E\left(\left(\Sigma_{i\in V^+}I_i\right)^2\right)-\E(|\mathcal{H}^+_*|)^2\\
    &=\E(|\mathcal{H}^+_*|)+n_1(n_1-1)(1-ps^2)^{2n_1-3}(1-qs^2)^{2n_2}-\E(|\mathcal{H}^+_*|)^2.
\end{aligned}    
\end{equation}
By \eqref{eq:expect}, \eqref{eq:variance} and Chebyshev’s inequality, we can obtain that
\begin{equation}
    \begin{aligned}
        \P(|\mathcal{H}^+_*| \leq \frac{1}{2}\E(|\mathcal{H}^+_*|)) \leq \frac{4\operatorname{Var}(|\mathcal{H}^+_*|)}{\E(|\mathcal{H}^+_*|)^2}    \leq \frac{5(n_1ps^2 -1)}{n_1-n_1ps^2} =o(1)
    \end{aligned}
\end{equation}
    Moreover, we have
    \begin{equation}\label{eq:expect2}
    \begin{aligned}
           \E(|\mathcal{H}^+_*|)&=n_1(1-ps^2)^{n_1-1}(1-qs^2)^{n_2}\\
           & \geq n_1 \left(1+\frac{ps^2}{1-ps^2}\right)^{-n_1}\left(1+\frac{qs^2}{1-qs^2}\right)^{-n_2}\\
           &\stackrel{(a)}{\geq} n_1 \exp\left(-\frac{ps^2}{1-ps^2}n_1 -\frac{qs^2}{1-qs^2}n_2 \right)\\
           &\stackrel{(b)}{\geq} \frac{1}{4}n^{1-\frac{ns^2(p+q)}{2\log n}}.
    \end{aligned}     
    \end{equation}
    The inequality $(a)$ holds by $e^x \geq 1+x$ and the inequality $(b)$ holds by the assumption $ns^2\frac{p+q}{2}=O(\log n)$ and \eqref{eq:n1n2}. Therefore, it holds that $|\mathcal{H}^+_*|\geq  \frac{1}{8}n^{1-\frac{ns^2(p+q)}{2\log n}} $ with probability $1-o(1)$.  Similarly, $|\mathcal{H}^-_*|\geq  \frac{1}{8}n^{1-\frac{ns^2(p+q)}{2\log n}} $  can also be proven using the same steps.

\end{IEEEproof}

\section{Proof of Theorem \ref{thm:csbm recovery ach}: \\ Achievability of Exact Community Recovery in Correlated Contextual Stochastic Block Models}\label{sec:proof of csbm recovery ach}
Recall that the graph $G_1$ consists of a database $X$ and an adjacency matrix $A$, while the graph $G_2$ consists of a database $Y$ and an adjacency matrix $B$.
Given a permutation $\pi : [n]\to [n]$, let $G_1+_\pi G_2$ represent a graph where database matrix is given by $X+_{\pi}Y$ and the adjacency matrix is given by $A \vee_{\pi}B$, where each node $i$ is assigned the vector $\frac{\bsx_i + \bsy_{\pi(i)}}{2}$ and  $(A \vee_{\pi}B)_{i,j}=\max\{A_{i,j},B_{\pi(i),\pi(j)}\}$. It can be seen that  
\begin{equation}\label{eq:union fea}
     \frac{\bsx_i+\bsy_{\pi_*(i)}}{2}= \boldsymbol{\mu}\bssigma_i + \frac{(1+\rho)\bsz_i+\sqrt{1-\rho^2}\bsw_i}{2}\sim \boldsymbol{\mu}\bssigma_i + \sqrt{\frac{1+\rho}{2}}\bso_i
\end{equation}
and
\begin{equation}\label{eq:union gra}
    A \vee_{\pi_*}B \sim  \text{SBM}(n,p(1-(1-s)^2),q(1-(1-s)^2)),
\end{equation}
where $\bsz_i,\bsw_i,\bso_i\sim \mathcal{N}(0,I_d)$.

Abbe et al. \cite{abbe2022} found the conditions under which, given a Contextual Stochastic Block Model(CSBM), the estimator $\hat{\bssigma}$ based on a spectral method can exactly recover $\bssigma$.
\begin{theorem}[Theorem 4.1 in \cite{abbe2022}]\label{thm:abbe2022ach}
    Assume that $\eqref{eq:assumption}$ holds. Let $G \sim \textnormal{CSBM}(n,p,q;R,d)$ with community labels $\bssigma : [n]\to \{+,-\}$ defined in Section \ref{sec:intro}. If
    \begin{equation}
        \frac{(\sqrt{a}-\sqrt{b})^2+c}{2}>1,
    \end{equation}
    then, then there exists an estimator $\hat{\bssigma}(G)$ such that $\hat{\bssigma}(G)=\bssigma$ with high probability.
\end{theorem}
When using the same spectral estimator $\hat{\bssigma}$ used in \cite{abbe2022}, we can have
\begin{equation}\label{eq:recov eq-1}
\begin{aligned}
    \P\left(\mathbf{ov}(\hat{\bssigma}(G_1+_{\hat{\pi}} G_2),\bssigma) \neq 1\right) &\leq \P\left(\{\mathbf{ov}(\hat{\bssigma}(G_1+_{\hat{\pi}} G_2),\bssigma) \neq 1\} \cap \{G_1+_{\hat{\pi}} G_2 =G_1+_{\pi_*} G_2\}\right) +\P(G_1+_{\hat{\pi}} G_2 \neq G_1+_{\pi_*} G_2)\\
     &\leq  \P\left(\mathbf{ov}(\hat{\bssigma}(G_1+_{\pi_*} G_2),\bssigma) \neq 1\right) +\P(\hat{\pi}\neq \pi_*).
\end{aligned}
\end{equation}
Suppose that \eqref{eq:csbm matching cond0}, \eqref{eq:csbm matching cond2}  and \eqref{eq:csbm matching cond1} hold. Then, by Theorem \ref{thm:csbm matching ach}, there exists an estimator $\hat{\pi}$ such that $\hat{\pi}=\pi_*$ with probability $1-o(1)$. Thus, it holds that 
\begin{equation}\label{eq:recov eq-2}
    \P(\hat{\pi}\neq \pi_*)=o(1).
\end{equation}
Moreover, combining Theorem \ref{thm:abbe2022ach} with \eqref{eq:union fea} and \eqref{eq:union gra}, we can obtain that if \eqref{eq:csbm recov ach cond1} holds, then
\begin{equation}\label{eq:recov eq-3}
    \P\left(\mathbf{ov}(\hat{\bssigma}(G_1+_{\pi_*} G_2),\bssigma) \neq 1\right)=o(1).
\end{equation}
Finally, by combining the results from \eqref{eq:recov eq-1}, \eqref{eq:recov eq-2} and \eqref{eq:recov eq-3}, the proof is complete.

\section{Proof of Theorem \ref{thm:csbm recovery imp}: \\ Impossiblity of Exact Community Recovery in Correlated Contextual Stochastic Block Models}\label{sec:proof of csbm recovery imp}

\begin{proof}[Proof of Theorem \ref{thm:csbm recovery imp}]
Suppose that \eqref{eq:assumption} holds. 
First, we generate a graph $H$ with the following distributions for the edges and node attributes.
Let $\bssigma^H$ be the community labels of the graph $H$. The probability of an edge forming between nodes within the same community is $p(1-(1-s)^2)=ps(2-s)$, while the probability of an edge forming between nodes in different communities is $q(1-(1-s)^2)=qs(2-s)$. Moreover, each node $i$ has a attribute represented by the vector $\begin{bmatrix}
\bsf_i\\
\bsg_i
\end{bmatrix} \sim \mathcal{N}\left(\begin{bmatrix}
\sigma_i \boldsymbol{\mu}\\
\sigma_i \boldsymbol{\mu}
\end{bmatrix}      , \left[
\begin{matrix}
    \bsI_d & \text{diag}(\rho) \\
\text{diag}(\rho) & \bsI_d \\
\end{matrix}
\right]\right)$, where $\bsf_i,\bsg_i\in \mathbb{R}^d$ and  $\boldsymbol{\mu}$ is generated from the uniform distribution over $\{ \boldsymbol{\mu} \in \mathbb{R}^d : \lVert\boldsymbol{\mu}\rVert^2 =R\}$ for $R>0$.
We have identified the conditions under which the exact community recovery is impossible in the graph $H$.
\begin{lemma}\label{lem:rec imp H}
    Suppose that \eqref{eq:assumption} holds.  If $\eqref{eq:csbm rec imp cond1}$ holds, then for any estimator $\tilde{\bssigma}$, we have
    \begin{equation}
       \P(\mathbf{ov}(\tilde{\bssigma}(H),\bssigma^H)=1) = o(1).
    \end{equation}
\end{lemma}

From the graph $H$, we generate two Contextual Stochastic Block Models $H_1$ and $H_2$ using the following method.
First, the edges of $H_1$  and $H'_2$ are generated as follows.
For $(i,j) \in [n]\times [n]$, 
\begin{itemize}
    \item If $(i,j)$ is not an edge in the graph $H$, it is also not an edge in both $H_1$ and $H'_2$
    \item If $(i,j)$ is an edge in the graph $H$, then
    \begin{itemize}
    \item $(i,j)$ is an edge in $H_1$ and $H'_2$ with probability $s_{11}$;
    \item $(i,j)$  is an edge in $H_1$ but not an edge in $H'_2$ with probability $s_{10}$;
    \item  $(i,j)$  is an edge in $H'_2$ but not an edge in $H_1$ with probability $s_{01}$,
    \end{itemize}
\end{itemize}
where $\left(s_{01}, s_{10}, s_{11}\right):=\left(\frac{s(1-s)}{1-(1-s)^2}, \frac{s(1-s)}{1-(1-s)^2}, \frac{s^2}{1-(1-s)^2}\right)$. Additionally, the node attributes of $H_1$  and $H'_2$  will be represented by the database matrices $[\bsf_1,\ldots,\bsf_n]^\top\in \mathbb{R}^{n\times d}$ and $[\bsg_1,\ldots,\bsg_n]^\top \in \mathbb{R}^{n\times d}$, respectively. 
Finally, by applying an arbitrary permutation $\pi :[n]\to [n]$ to the nodes of $H'_2$, we obtain the graph $H_2$.
Then, we can confirm that $(H_1,H_2,\bssigma^H)$ and $(G_1,G_2,\bssigma_*)$ follow the same distribution.  

Suppose that there exists an estimator $\tilde{\bssigma}$ such that we can exactly recover $\bssigma_*$ by observing $G_1$ and $G_2$. 
Since $(H_1,H_2,\bssigma^H)$ and $(G_1,G_2,\bssigma_*)$ follow the same distribution, the estimator should also be able to exactly recover $\bssigma^H$ with high probability by observing $H_1$ and $H_2$.
However, since $(H_1,H_2)$ was generated  from $H$ using only random sampling and random permutation, this assumption leads to a contradiction according to Lemma \ref{lem:rec imp H}.
\end{proof}

\subsection{Proof of Lemma \ref{lem:rec imp H}}

For two sequences of random variables $\{X_n\}$ and $\{Y_n\}$, and a deterministic sequence $\{r_n\}_{n=1}^\infty \subset (0, \infty)$, we write
\[
X_n = o_{\mathbb{P}}(Y_n; r_n)
\]
if there exists a constant $C > 0$ and a deterministic sequence $\{c_n\}$ with $c_n \to 0$ such that for all $C' > 0$, there exists $N_0 > 0$ satisfying
\[
\mathbb{P}\left( |X_n| \geq c_n |Y_n| \right) \leq C e^{-C' r_n}, \quad \forall n \geq N_0.
\]

\begin{proof}[Proof of Lemma \ref{lem:rec imp H}]
When \eqref{eq:assumption} holds, Abbe et al. \cite{abbe2022} demonstrated that the condition $\frac{(\sqrt{a}-\sqrt{b})^2+c}{2}<1$ makes the exact community recovery impossible in the Contextual Stochastic Block Model(CSBM), as explained in Section \ref{sec:intro}. We apply a similar proof technique as in \cite{abbe2022} to prove Lemma \ref{lem:rec imp H}.

Let
\begin{equation}
    I(t,a,b,c)=\frac{a}{2}\left[1-(a / b)^t\right]+\frac{b}{2}\left[1-(b / a)^t\right]-2 c\left(t+t^2\right)
\end{equation}
and
\begin{equation}
    I^*(a,b,c):=  \frac{(\sqrt{a}-\sqrt{b})^2 +c}{2}.
\end{equation}
Then, $I(t,a,b,c)$ and $I^*(a,b,c) $ satisfy the following relationship.
\begin{lemma}\label{lem:II*}
    For $a,b,c>0$, we have that
    \begin{equation}
        I^*(a,b,c)=I(-1/2,a,b,c)=\sup_{t\in \mathbb{R}}I(t,a,b,c).
    \end{equation}
\end{lemma}
For an estimator $\hat{\bssigma}$, define
\begin{equation}
    r(\hat{\bssigma},\bssigma^H):= \frac{1}{n}\min_{\nu\in \{-1,1\}}\sum\limits_{i=1}^n |\hat{\sigma}_i-\nu \sigma^H_i|.
\end{equation}
We will show that  for any estimator $\hat{\bssigma}$, $\E [ r(\hat{\bssigma},\bssigma^H)]>n^{1-I^*(a',b',c')}$, where $a'=as(2-s),b'=bs(2-s)$ and $\frac{\left(\frac{2}{1+\rho}R\right)^2}{\frac{2}{1+\rho}R+ d/n}=c' \log n$. Through this, we can prove that if $I^*(a',b',c')<1$, then the exact community recovery is impossible.

    Let the adjacency matrix of the graph $H$ be denoted by $J$, and let $Z:=[(\bsf_1,\bsg_1), (\bsf_2, \bsg_2,),\ldots,(\bsf_n,\bsg_n)]^\top \in \mathbb{R}^{n\times 2d}$. Additionally, let the
database matrix be $U:=[\bsu_1,\bsu_2,\ldots,\bsu_n]^\top \in \mathbb{R}^{n\times d}$ where $\bsu_i=\frac{1}{\sqrt{2+2\rho}}(\bsf_i+\bsg_i)$.

Abbe et al. \cite{abbe2022} studied the fundamental limit via a genie-aided approach, and the results are as follows.
\begin{lemma}[Lemma F.3 in \cite{abbe2022}]\label{lem:genie-aided}
     Suppose that $\mathcal{S}$ is a Borel space and $(\bssigma, \boldsymbol{X})$ is a random element in $\{ \pm 1\}^n \times \mathcal{S}$. Let $\mathcal{F}$ be a family of Borel mappings from $\mathcal{S}$ to $\{ \pm 1\}^n$. Define
\begin{equation}
    \begin{aligned}
& \mathcal{M}(\bsu, \bsv)=\min \left\{\frac{1}{n} \sum_{i=1}^n \mathds{1}_{\left\{u_i \neq v_i\right\}}, \frac{1}{n} \sum_{i=1}^n \mathds{1}_{\left\{-u_i \neq v_i\right\}}\right\}, \quad \forall \bsu, \bsv \in\{ \pm 1\}^n \\
& f\left(\cdot \mid \tilde{\boldsymbol{X}}, \tilde{\sigma}_{-i}\right)=\mathbb{P}\left(\sigma_i=\cdot \mid \boldsymbol{X}=\tilde{\boldsymbol{X}}, \sigma_{-i}=\tilde{\sigma}_{-i}\right), \quad \forall i \in[n], \quad \tilde{\boldsymbol{X}} \in \mathcal{S}, \tilde{\sigma}_{-i} \in\{ \pm 1\}^{n-1}.
\end{aligned}
\end{equation}
We have
\begin{equation}
   \inf _{\hat{\bssigma} \in \mathcal{F}} \mathbb{E} \mathcal{M}(\hat{\bssigma}, \bssigma) \geq \frac{n-1}{3 n-1} \cdot \frac{1}{n} \sum_{i=1}^n \mathbb{P}\left[f\left(\sigma_i \mid \boldsymbol{X}, \sigma_{-i}\right)<f\left(-\sigma_i \mid \boldsymbol{X}, \sigma_{-i}\right)\right].
\end{equation}
\end{lemma}

Let $\mathcal{U}:=\mathcal{H}(UU^\top)$, where $\mathcal{H}(\cdot)$ represents the hollowing operator, which sets all diagonal entries of a square matrix to zero. Recall that $\boldsymbol{\mu}$ is generated from the uniform distribution over $\{ \boldsymbol{\mu} \in \mathbb{R}^d : \lVert\boldsymbol{\mu}\rVert^2 =R\}$. The following lemma is the result of applying the genie-aided analysis by Abbe et al. (Lemma 4.1 in \cite{abbe2022}) for the Contextual Stochastic Block Model to our graph model, which have two correlated node attributes.
The key difference from the Contextual Stochastic Block Model is as follows: In the Contextual Stochastic Block Model, each node $i$ is assigned a single attribute $\bsx_i \sim \mathcal{N} (\boldsymbol{\mu}\sigma_i,\bsI_d)$. Consequently, given $\bssigma$, the probability distribution for the database $X=[\bsx_1,\ldots,\bsx_n]^\top$ is proportional to $$\E_{\boldsymbol{\mu}}\exp\left( -\frac{1}{2} \sum^{n}_{i=1}\|\bsx_i-\sigma_i \boldsymbol{\mu} \|^2\right) \propto \E_{\boldsymbol{\mu}}\exp\left( \langle \sum^{n}_{i=1}\bsx_i\sigma_i,\boldsymbol{\mu} \rangle\right).$$ In contrast, in our graph $H$, each node $i$ is assigned two correlated attributes $\bsf_i$ and $\bsg_i$, so given $\bssigma^H$, the probability distribution for the database $Z:=[\bsf_1,\bsf_2,\ldots,\bsf_n,\bsg_1,\bsg_2,\ldots,\bsg_n]^\top$ becomes proportional to  $\E_{\boldsymbol{\mu}}\exp\left( -\frac{1}{2(1-\rho^2)} \sum^n_{i=1} \left( \|\bsf_i-\sigma_i \boldsymbol{\mu} \|^2 -2\rho \langle \bsf_i -\boldsymbol{\mu},\bsg_i-\boldsymbol{\mu}\rangle  + \|\bsg_i-\sigma_i \boldsymbol{\mu} \|^2\right) \right) \propto  \E_{\boldsymbol{\mu}}\exp\left(  \frac{1}{1+\rho}\langle\sum^{n}_{i=1}(\bsf_i+\bsg_i)\sigma_i,\boldsymbol{\mu} \rangle\right)$. Moreover, we can have that $\bsu_i=\frac{1}{\sqrt{2+2\rho}}(\bsf_i+\bsg_i)\sim \mathcal{N}(\boldsymbol{\mu}'\sigma_i,I_d)$, where $\boldsymbol{\mu}'=\sqrt{\frac{2}{1+\rho}}\boldsymbol{\mu}$. Thus, it holds that $  \frac{1}{1+\rho}\langle\sum^{n}_{i=1}(\bsf_i+\bsg_i)\sigma_i,\boldsymbol{\mu} \rangle =\langle\sum^{n}_{i=1}\bsu_i \sigma_i,\boldsymbol{\mu}' \rangle $.
Then, we can obtain the following lemma for our graph $H$, which is a generalized version of Lemma 4.1 in \cite{abbe2022}.
\begin{lemma}\label{lem:abbe1}
 Suppose that \eqref{eq:assumption} holds. For each given  $i$, we have 
 \begin{equation}
     \left|\log \left(\frac{\mathbb{P}\left(\sigma^H_i=1 \mid J, Z, \sigma^H_{-i}\right)}{\mathbb{P}\left(\sigma^H_i=-1 \mid J, Z, \sigma^H_{-i}\right)}\right)-\left[\left(\log (a' / b') J+\frac{2}{n+d / R'^2} \mathcal{U}\right) \bssigma^H \right]_i\right|=o_{\mathbb{P}}\left(\log n ; \log n\right)
 \end{equation}
 where  $a'=a(1-(1-s)^2),b'=b(1-(1-s)^2)$  and $R'=\frac{2}{1+\rho}\lVert\boldsymbol{\mu}\rVert^2=\frac{2}{1+\rho}R$.
\end{lemma}

Furthermore, given $\sigma_i$, $\sum_{j\neq i} \langle \bsu_i,\bsu_j \rangle\sigma_j$ and $\sum_{j\neq i} J_{ij}\sigma_j$ are independent, allowing us to obtain the same result as Lemma F.2 in \cite{abbe2022} for our model.

\begin{lemma} \label{lem:abbe2}
  Suppose that \eqref{eq:assumption} holds. Let $\bsu_i:=\frac{1}{\sqrt{2+2\rho}}(\bsf_i+\bsg_i)$.
  Define
  \begin{equation}
      W_{n i}=\left(\frac{2 R'^2}{n  R'^2+d} \sum_{j \neq i}\left\langle \bsu_i, \bsu_j\right\rangle \sigma^H_j+\log (a' / b') \sum_{j \neq i} J_{i j} \sigma^H_j\right) \sigma^H_i.
  \end{equation}
For any fixed $i$,
\begin{equation}
    \lim _{n \rightarrow \infty} \mathbb{P}\left(W_{n i} \leq \varepsilon \log n\right)=n^{-\sup_{t \in \mathbb{R}}\{\varepsilon t+I(t, a', b', c')\}}, \quad \forall \varepsilon<\frac{a'-b'}{2} \log (a' / b')+2 c',
\end{equation}
where $a'=a(1-(1-s)^2),b'=b(1-(1-s)^2)$ and $\frac{R'^2}{R'+ d/n}=c' \log n$.
\end{lemma}

By Lemma \ref{lem:genie-aided}, we can obtain that
\begin{equation}
    \E r(\hat{\bssigma},\bssigma^H) \geq \frac{n-1}{3 n-1} \P \left(f\left(\sigma^H_1 \mid J,Z, \sigma^H_{-1}\right)<f\left(-\sigma^H_1 \mid J,Z, \sigma^H_{-1}\right)\right).
\end{equation}
Define two events
\begin{equation}
 \begin{aligned}
& \mathcal{A}_{\varepsilon}:=\left\{\left|\log \left(\frac{f\left(\sigma^H_1 \mid J,Z, \sigma^H_{-1}\right)}{f\left(-\sigma^H_1 \mid J,Z, \sigma^H_{-1}\right)}\right)-\left(\log (a' / b')(J \sigma^H)_1+\frac{2 R'^2}{n R'^2+d}(\mathcal{U}\sigma^H)_1\right) \sigma^H_1\right|<\varepsilon \log n\right\} \\
& \mathcal{B}_{\varepsilon}:=\left\{\left(\frac{2 R'^2}{n R'^2+d}(\mathcal{U}\sigma^H)_1 + \log (a' / b')(J\sigma^H)_1\right) \sigma^H_1 \leq-\varepsilon \log n\right\}.
\end{aligned}
\end{equation}
If $\mathcal{A}_\varepsilon \cap \mathcal{B}_\varepsilon$ holds, then it is easy to verify that  $f\left(\sigma^H_1 \mid J,Z, \sigma^H_{-1}\right)<f\left(-\sigma^H_1 \mid J,Z, \sigma^H_{-1}\right)$ also  holds. 
Therefore, we can have
\begin{equation}
     \E r(\hat{\bssigma},\bssigma^H) \geq \frac{n-1}{3 n-1}\P(\mathcal{A}_\varepsilon \cap \mathcal{B}_\varepsilon)\geq  \frac{n-1}{3 n-1}\left(\P(\mathcal{B}_\varepsilon)-\P(\mathcal{A}^c_\varepsilon)\right).
\end{equation}
By Lemma \ref{lem:abbe1}, it holds that
\begin{equation}
    \P(\mathcal{A}_\varepsilon^c)=o(1/n). 
\end{equation}
Furthermore, by Lemma \ref{lem:abbe2}, it holds that 
\begin{equation}
    \lim_{n\to \infty }\P(\mathcal{B}_\varepsilon)= n^{-\sup_{t \in \mathbb{R}}\{\varepsilon t+I(t, a', b', c')\}}.
\end{equation}
We can have that $I^*(a',b',c')=I(-1/2,a',b',c')=\sup_{t\in \mathbb{R}}I(t,a',b',c')$ since Lemma \ref{lem:II*}. Therefore, taking $\varepsilon \to 0$, we can obtain that
\begin{equation}
     \liminf_{n\to \infty }\E r(\hat{\bssigma},\bssigma^H) \geq n^{-I^*(a',b',c')}.
\end{equation}
Thus, if $I^*(a',b',c')<1$, then the exact community recovery is impossible.
\end{proof}

\section{Proof of Lemma \ref{lem:upperbound F_t}}\label{sec:proof of lemma upperbound f_t}

Recall that 
\begin{equation}
    \mathcal{F}_t=\left\{ \sum_{k=1}^t Z_{i_k i_k} \geq \sum_{k=1}^t Z_{i_k i_{k+1}}  \right\}
\end{equation}
for $t\geq 2$, where $i_1,\ldots,i_t$ are $t$ distinct integers belonging to $[n]$ and $i_{t+1}=i_1$. 

 We will first show that  
    \begin{equation}
        \P\left(\mathcal{F}_t\right) \leq \exp\left(-\frac{d}{2}S\left(\frac{1-\rho^2}{\rho^2},t\right) \right) +\P(\mathcal{A}^c_1)
    \end{equation}
    for any $t$ distinct integers $i_1,...,i_t \in [n]$, i.e., \eqref{eq:upperbound F_t1} holds. 
    We will prove this for the cases $t=2$ and $t\geq 3$, separately. 

Let us first consider the case where $t=2$.  
Without loss of generality, we assume that $\pi_* : [n]\to [n]$ is the identity permutation. Thus, $Y=Y'$, i.e., $\bsy_i=\bsy_i'$ for all $i\in[n]$.
From the definition of $Z_{i,j}=\|\bsx_i-\bsy_j\|^2$, we have
\beq
\begin{split}
Z_{1,1}&=\|\bsx_1-\bsy_1\|^2=\|(\bsmu \sigma_1+\bsz_1)-(\bsmu \sigma_1+\rho\bsz_1+\sqrt{1-\rho^2} \bsw_1)\|^2=\| (1-\rho)\bsz_1-\sqrt{1-\rho^2}\bsw_1 \|^2;\\
Z_{2,2}&=\|\bsx_2-\bsy_2\|^2=\| (1-\rho)\bsz_2-\sqrt{1-\rho^2}\bsw_2 \|^2;\\
Z_{1,2}&=\|\bsx_1-\bsy_2\|^2=\|(\bsmu \sigma_1+\bsz_1)-(\bsmu \sigma_2+\rho\bsz_2+\sqrt{1-\rho^2} \bsw_2)\|^2=\|\bsmu(\sigma_1-\sigma_2)+\bsz_1-\rho \bsz_2-\sqrt{1-\rho^2} \bsw_2 \|^2;\\
Z_{2,1}&=\|\bsx_2-\bsy_1\|^2=\|\bsmu(\sigma_2-\sigma_1)+\bsz_2-\rho \bsz_1-\sqrt{1-\rho^2} \bsw_1 \|^2.
\end{split}
\eeq
 If nodes $1$ and $2$ belong to the same community, i.e., $\sigma_1=\sigma_2$,  we can obtain that 
\begin{equation}\label{eq:belong same}
    \begin{aligned}
        \P(\mathcal{F}_2)&=\P(Z_{1,1}+Z_{2,2}\geq Z_{1,2}+Z_{2,1})\\
        &=\P\left(\sqrt{1-\rho^2}\left\langle \bsz_1-\bsz_2, \bsw_2-\bsw_1\right\rangle \geq \rho\left\|\bsz_1-\bsz_2\right\|^2\right)\\
        &=\E_{\bsz_1, \bsz_2} \P_{u \sim \mathcal{N}\left(0,2\left(1-\rho^2\right)\left\|\bsz_1-\bsz_2\right\|^2\right)}\left(u \geq \rho\left\|\bsz_1-\bsz_2\right\|^2\right)\\
        &=\E_{\bsz_1, \bsz_2} \P_{u \sim \mathcal{N}\left(0,1\right)}\left(u \geq \sqrt{\frac{\rho^2 \left\|\bsz_1-\bsz_2\right\|^2}{2(1-\rho^2)}} \right)\\
        &\stackrel{(a)}{\leq} \E_{\bsz_1,\bsz_2} \exp \left(-\frac{\rho^2}{4(1-\rho^2) }\left\|\bsz_1-\bsz_2\right\|^2\right)\\
        &\stackrel{(b)}{= }\det \left(\bsI_{2d}+\frac{\rho^2}{2(1-\rho^2)}\bsL^{P_2} \otimes \bsI_d \right)^{-1/2}\\
        &=\det \left(\bsI_2+\frac{\rho^2}{2(1-\rho^2)}\bsL^{P_2}\right)^{-d/2}\\
        &=\left(\frac{\rho^2}{1-\rho^2}\right)^{-d/2}\\
        &= \exp\left(-\frac{d}{2}S\left(\frac{1-\rho^2}{\rho^2},2\right)\right),
    \end{aligned}
\end{equation}
where $\bsL^{P_2}=\begin{bmatrix}
1 & -1\\
-1 & 1 
\end{bmatrix}	
$ is the Laplacian matrix for a path graph consisting of two nodes.
The inequality $(a)$ holds by the tail bound on normal distribution (Lemma \ref{lem:normal tail}), the equality $(b)$ holds by Lemma \ref{lem:MGF QF} and $\left\|\bsz_1-\bsz_2\right\|^2=\bsz^\top (\bsL^{P_2} \otimes \bsI_d) \bsz$ where $\bsz$ is the concatenation of $\bsz_1$ and $\bsz_2$, and the last equality holds by the definition  of $S(\alpha,t)$ in \eqref{eq:def S}.

Conversely, if nodes $1$ and $2$ belong to different communities, i.e., $\sigma_1\neq \sigma_2$, we can obtain that
\begin{equation}\label{eq:augment-1}
    \begin{aligned}
        \P(\mathcal{F}_2)&=\P(Z_{1,1}+Z_{2,2}\geq Z_{1,2}+Z_{2,1})\\
        &=\P\left(\sqrt{1-\rho^2}\left\langle \boldsymbol{\mu}(\sigma_1-\sigma_2)+\bsz_1-\bsz_2, \bsw_2-\bsw_1\right\rangle    \geq     \left\|2\boldsymbol{\mu} \right\|^2+ (1+\rho) \langle \boldsymbol{\mu}(\sigma_1-\sigma_2),\bsz_1-\bsz_2 \rangle +\rho\left\|\bsz_1-\bsz_2\right\|^2\right).
    \end{aligned}
\end{equation} 
On the event $\mathcal{A}_1$, defined in \eqref{eqn:def_A1}, it holds that 
\begin{equation}\label{eq:augment-2}
    \left\|2\boldsymbol{\mu} \right\|^2+ (1+\rho) \langle \boldsymbol{\mu}(\sigma_1-\sigma_2),\bsz_1-\bsz_2 \rangle +\rho\left\|\bsz_1-\bsz_2\right\|^2 \geq \rho \lVert \boldsymbol{\mu}(\sigma_1-\sigma_2)+\bsz_1-\bsz_2 \rVert^2,
\end{equation} 
both for the cases where $(\sigma_1,\sigma_2)=(+1, -1)$ and $(\sigma_1,\sigma_2)=(-1, +1)$, since \eqref{eq:augment-2} is equivalent to
\beq
\|2\bsmu\|^2\geq-\langle \bsmu, \bsz_1-\bsz_2\rangle (\sigma_1-\sigma_2),
\eeq
and we have $-\|\bsmu\|^2\leq -\langle \bsmu, \bsz_i\rangle \leq \|\bsmu\|^2$ for $\forall i\in[n]$ on the event $\mathcal{A}_1$.

Thus, by combining \eqref{eq:augment-1} and \eqref{eq:augment-2} and assuming $(\sigma_1,\sigma_2)=(+1, -1)$ without loss of generality, we can obtain that
\begin{equation}\label{eq:t=2 first}
    \begin{aligned}
         \P(\mathcal{F}_2)&\leq \P(\mathcal{F}_2\cap\mathcal{A}_1 )+\P(\mathcal{A}^c_1) \\
         &\leq \P\left(\sqrt{1-\rho^2}\left\langle 2\boldsymbol{\mu}+\bsz_1-\bsz_2, \bsw_2-\bsw_1\right\rangle    \geq      \rho \lVert 2\boldsymbol{\mu}+\bsz_1-\bsz_2 \rVert^2\right) +\P(\mathcal{A}^c_1) \\
        &=\P\left(\sqrt{1-\rho^2}\left\langle \bszeta_1-\bszeta_2, \bsw_2-\bsw_1\right\rangle   \geq      \rho \lVert \bszeta_1-\bszeta_2 \rVert^2\right)+\P(\mathcal{A}^c_1) \\
        &=\E_{\bszeta_1, \bszeta_2} \P_{u \sim \mathcal{N}\left(0,2\left(1-\rho^2\right)\left\|\bszeta_1-\bszeta_2\right\|^2\right)}\left(u \geq \rho\left\|\bszeta_1-\bszeta_2\right\|^2\right)+\P(\mathcal{A}^c_1)\\
        &=\E_{\bszeta_1, \bszeta_2} \P_{u \sim \mathcal{N}\left(0,1\right)}\left(u \geq \sqrt{\frac{\rho^2 \left\|\bszeta_1-\bszeta_2\right\|^2}{2(1-\rho^2)}} \right)+\P(\mathcal{A}^c_1)\\
        &\stackrel{(a)}{\leq} \E_{\bszeta_1,\bszeta_2} \exp \left(-\frac{\rho^2}{4(1-\rho^2) }\left\|\bszeta_1-\bszeta_2\right\|^2\right)+\P(\mathcal{A}^c_1)\\
        &\stackrel{(b)}{= }\det \left(\bsI_{2d}+\frac{\rho^2}{2(1-\rho^2)}(\bsL^{P_2} \otimes \bsI_d) \right)^{-1/2} \exp (-\rho^2 \lVert \boldsymbol{\mu}\rVert^2 )+\P(\mathcal{A}^c_1)\\
        &\leq \exp\left(-\frac{d}{2}S\left(\frac{1-\rho^2}{\rho^2},2\right) \right)+\P(\mathcal{A}^c_1),
    \end{aligned}
\end{equation}
 where $\bszeta_1=\boldsymbol{\mu}+\bsz_1,$ $\bszeta_2=-\boldsymbol{\mu}+\bsz_2$. The inequality $(a)$ holds by the tail bound on normal distribution (Lemma \ref{lem:normal tail}), the equality $(b)$ holds by Lemma \ref{lem:MGF QF} and $\left\|\bszeta_1-\bszeta_2\right\|^2=\bszeta^\top (L^{P_2} \otimes I_d) \bszeta$ where $\bszeta$ is the concatenation of $\bszeta_1$ and $\bszeta_2$, and the last equality holds by the definition  of $S(\alpha,t)$ in \eqref{eq:def S}.

Next, we consider the case where $t \geq 3$. It holds that 
\begin{equation}
    \sum_{i=1}^t Z_{i, i} \geq \sum_{i=1}^{t-1} Z_{i, i+1}+Z_{t,1} \text{ if and only if}
\end{equation}
\begin{equation}\label{eq:main eq}
    \sqrt{1-\rho^2} \sum_{i=1}^t \left\langle  \bsx_i-\bsx_{i+1}, \bsw_{i+1}\right\rangle \geq  \sum_{i=1}^t \frac{\|\boldsymbol{\mu}\|^2}{2}\left(\sigma_i-\sigma_{i+1}\right)^2  + \rho\langle\boldsymbol{\mu}, \bsz_i \rangle\left(\sigma_i-\sigma_{i-1}\right)+\left\langle\boldsymbol{\mu}, \bsz_i\right\rangle\left(\sigma_i-\sigma_{i+1}\right)+ \frac{\rho}{2} \left\|\bsz_i-\bsz_{i+1}\right\|^2
\end{equation}
where $\bsw_{t+1}=\bsw_1, \bsx_{t+1}=\bsx_1, \bsz_{t+1}=\bsz_1$ and $\sigma_{t+1}=\sigma_1$. 
Note that
 \begin{equation*}
  \begin{aligned}
       \sum_{i=1}^t \frac{\rho}{2}\|\bsx_i-\bsx_{i+1}\|^2 &=  \sum_{i=1}^t  \frac{\rho}{2} \| \bsmu(\sigma_i-\sigma_{i+1}) +(\bsz_i-\bsz_{i+1})  \|^2\\
       &  =\sum_{i=1}^t  \frac{\rho}{2} \|\bsmu \|^2(\sigma_i-\sigma_{i+1})^2 +\frac{\rho}{2} \|\bsz_i-\bsz_{i+1}\|^2 +\rho \langle \bsmu,\bsz_i \rangle (\sigma_i-\sigma_{i+1})+\rho \langle \bsmu,\bsz_{i+1} \rangle (\sigma_{i+1}-\sigma_i).
  \end{aligned}   
  \end{equation*}
  On the event $\mathcal{A}_1$,  we have $2\|\bsmu\|^2\geq -\langle \bsmu, \bsz_i\rangle (\sigma_i-\sigma_{i+1})$  for all  $i \in[t]$. 
  Therefore, we can obtain that
  \begin{equation*}
      \begin{aligned}
              \text{Right-hand side of   \eqref{eq:main eq}  }-\sum_{i=1}^t \frac{\rho}{2}\|\bsx_i-\bsx_{i+1}\|^2 =  \sum_{i=1}^t \frac{1-\rho}{2}\|\bsmu \|^2 (\sigma_i-\sigma_{i+1})^2 + (1-\rho) \langle \bsmu,z_i \rangle (\sigma_i-\sigma_{i+1}) \geq 0.
      \end{aligned}
  \end{equation*}
Thus, we have
\begin{equation}
\begin{split}\label{eqn:bd_A1}
    \text{Right-hand side of  \eqref{eq:main eq} }& \geq \sum_{i=1}^t \frac{\rho}{2}\|\bsx_i-\bsx_{i+1}\|^2= \frac{\rho}{2}\bsx^\top (\bsL^{C_t} \otimes \bsI_d) \bsx,
\end{split}
\end{equation}
where $\bsx$ is the concatenation of $\bsx_i$ for $i=1,\dots,t$ and $\bsL^{C_t}$ is the Laplacian matrix for a cycle graph consisting of $t$ nodes.

Therefore, we can obtain
\begin{equation}\label{eq:t>3 first}
    \begin{aligned}
        \P(\mathcal{F}_t)&\leq \P(\mathcal{F}_t\cap\mathcal{A}_1 )+\P(\mathcal{A}^c_1) \\
        &\leq \P\left( \sqrt{1-\rho^2} \sum_{i=1}^t \left\langle  \bsx_i-\bsx_{i+1}, \bsw_{i+1}\right\rangle \geq   \frac{\rho}{2}\bsx^\top (\bsL^{C_t} \otimes \bsI_d) \bsx \right) +\P(\mathcal{A}^c_1)\\
        &=\E_{\bsx_1,\ldots, \bsx_t} \P_{u \sim \mathcal{N}\left(0,\left(1-\rho^2\right)\bsx^\top (\bsL^{C_t} \otimes \bsI_d) \bsx\right)}\left(u \geq  \frac{\rho}{2}\bsx^\top (\bsL^{C_t} \otimes \bsI_d) \bsx\right) +\P(\mathcal{A}^c_1)\\
        &=\E_{\bsx_1,\ldots, \bsx_t} \P_{u \sim \mathcal{N}\left(0,1\right)}\left(u \geq \sqrt{\frac{\rho^2 \bsx^\top (\bsL^{C_t} \otimes \bsI_d )\bsx}{4(1-\rho^2)}} \right) +\P(\mathcal{A}^c_1)\\
        &\stackrel{(a)}{\leq} \E_{\bsx_1,\ldots, \bsx_t} \exp \left(-\frac{\rho^2}{8(1-\rho^2) }\bsx^\top (\bsL^{C_t} \otimes \bsI_d )\bsx\right) +\P(\mathcal{A}^c_1)\\
        &\stackrel{(b)}{= } \det \left(\bsI_{dt}+\frac{\rho^2}{4(1-\rho^2)}\bsL^{C_t} \otimes \bsI_d \right)^{-1/2}\exp \left(-\frac{1}{2} \bsv^\top \left[\bsI_{dt}-\left[\bsI_{dt}+\frac{\rho^2}{4\left(1-\rho^2\right)} \bsL^{C_t} \otimes \bsI_d\right]^{-1}\right] \bsv\right) +\P(\mathcal{A}^c_1)\\
        &\stackrel{(c)}{\leq} \det \left(\bsI_t+\frac{\rho^2}{4(1-\rho^2)}\bsL^{C_t}\right)^{-d/2} +\P(\mathcal{A}^c_1)\\
        &\stackrel{(d)}{=}\left(\prod_{j=0}^{t-1}\left\{1+\frac{\rho^2}{2(1-\rho^2)}\left(1-\cos \left(\frac{2 \pi j}{t}\right)\right)\right\}\right)^{-d / 2} +\P(\mathcal{A}^c_1)\\
        &= \exp\left(-\frac{d}{2}S\left(\frac{1-\rho^2}{\rho^2},t\right)\right) +\P(\mathcal{A}^c_1),
    \end{aligned}
\end{equation}
where $\bsv$ is the concatenation of $\boldsymbol{\mu} \sigma_i$ for $i=1,\dots,t$.
The inequality $(a)$ holds by Lemma \ref{lem:normal tail}, the equality $(b)$ holds by Lemma \ref{lem:MGF QF}, the inequality $(c) $ holds by confirming through Lemma \ref{lem:eigen LCT} that the eigenvalue of $\bsI_{dt}-\left[\bsI_{dt}+\frac{\rho^2}{4\left(1-\rho^2\right)} \bsL^{C_t} \otimes \bsI_d\right]^{-1}$ is non-negative, thereby establishing that $\bsI_{dt}-\left[\bsI_{dt}+\frac{\rho^2}{4\left(1-\rho^2\right)} \bsL^{C_t} \otimes \bsI_d\right]^{-1}$ is positive semi-definite, and the last equality holds by the definition  of $S(\alpha,t)$ in \eqref{eq:def S}. 
Thus, \eqref{eq:upperbound F_t1} holds.

Now, we will show that \eqref{eq:upperbound F_t2} holds.
Similar to the previous analysis, let us start by considering the case when $t=2$. Since the scenario where the two nodes belong to the same community yields a result analogous to that of \eqref{eq:belong same}, we will now focus on the case where the two nodes belong to different communities. On the event $\mathcal{A}_2$, it holds that
\begin{equation}\label{eq:noncentral-1}
    \left\|2\boldsymbol{\mu} \right\|^2+ (1+\rho) \langle \boldsymbol{\mu}(\sigma_1-\sigma_2),\bsz_1-\bsz_2 \rangle +\rho\left\|\bsz_1-\bsz_2\right\|^2 \geq \lambda \lVert \boldsymbol{\mu}(\sigma_1-\sigma_2)+\bsz_1-\bsz_2 \rVert^2.
\end{equation}
both for the cases where $(\sigma_1,\sigma_2)=(+1, -1)$ and $(\sigma_1,\sigma_2)=(-1, +1)$.
Thus, by combining \eqref{eq:augment-1} and \eqref{eq:noncentral-1} and assuming $(\sigma_1,\sigma_2)=(+1, -1)$ without loss of generality, we can obtain  
\begin{equation}\label{eq:t=2 second}
    \begin{aligned}
        \P(\mathcal{F}_2)&\leq \P(\mathcal{F}_2\cap\mathcal{A}_2 )+\P(\mathcal{A}^c_2) \\
        &\leq \P\left(\sqrt{1-\rho^2}\left\langle 2\boldsymbol{\mu}+\bsz_1-\bsz_2, \bsw_2-\bsw_1\right\rangle    \geq      \lambda \lVert 2\boldsymbol{\mu}+\bsz_1-\bsz_2 \rVert^2\right) +\P(\mathcal{A}^c_2) \\
        &=\P\left(\sqrt{1-\rho^2}\left\langle \bszeta_1-\bszeta_2, \bsw_2-\bsw_1\right\rangle   \geq     \lambda \lVert \bszeta_1-\bszeta_2 \rVert^2\right) +\P(\mathcal{A}^c_2)\\
        &=\E_{\bszeta_1, \bszeta_2} \P_{u \sim \mathcal{N}\left(0,2\left(1-\rho^2\right)\left\|\bszeta_1-\bszeta_2\right\|^2\right)}\left(u \geq \lambda \left\|\bszeta_1-\bszeta_2\right\|^2\right)+\P(\mathcal{A}^c_2)\\
        &=\E_{\bszeta_1, \bszeta_2} \P_{u \sim \mathcal{N}\left(0,1\right)}\left(u \geq \sqrt{\frac{\lambda^2 \left\|\bszeta_1-\bszeta_2\right\|^2}{2(1-\rho^2)}} \right)+\P(\mathcal{A}^c_2)\\
        &\stackrel{(a)}{\leq} \E_{\bszeta_1,\bszeta_2} \exp \left(-\frac{\lambda^2}{4(1-\rho^2) }\left\|\bszeta_1-\bszeta_2\right\|^2\right)+\P(\mathcal{A}^c_2)\\
        &\stackrel{(b)}{= }\det \left(\bsI_{2d}+\frac{\lambda^2}{2(1-\rho^2)}(\bsL^{P_2} \otimes \bsI_d) \right)^{-1/2} \exp \left(-\frac{\lambda^2}{1+\lambda^2-\rho^2} \lVert \boldsymbol{\mu}\rVert^2 \right)+\P(\mathcal{A}^c_2)\\
        &\leq \exp\left(-\frac{d}{2}S\left(\frac{1-\rho^2}{\lambda^2},2\right)  \right)+\P(\mathcal{A}^c_2)
    \end{aligned}
\end{equation}
where $\bszeta_1=\boldsymbol{\mu}+\bsz_1,$ $\bszeta_2=-\boldsymbol{\mu}+\bsz_2$. The inequality $(a)$ holds by the tail bound on normal distribution (Lemma \ref{lem:normal tail}), the equality $(b)$ holds by Lemma \ref{lem:MGF QF} and $\left\|\bszeta_1-\bszeta_2\right\|^2=\bszeta^\top (\bsL^{P_2} \otimes \bsI_d) \bszeta$ where $\bszeta$ is the concatenation of $\bszeta_1$ and $\bszeta_2$, and the last equality holds by the definition  of $S(\alpha,t)$ in \eqref{eq:def S}.

Let us consider the case where $t \geq 3$. Recall that it holds that 
\begin{equation}
    \sum_{i=1}^t Z_{i, i} \geq \sum_{i=1}^{t-1} Z_{i, i+1}+Z_{t,1}
\end{equation}
if and only if
\begin{equation}\label{eq:main eq2}
   \sqrt{1-\rho^2} \sum_{i=1}^t \left\langle  \bsx_i-\bsx_{i+1}, \bsw_{i+1}\right\rangle \geq  \sum_{i=1}^t \frac{\|\boldsymbol{\mu}\|^2}{2}\left(\sigma_i-\sigma_{i+1}\right)^2  + \rho\langle\boldsymbol{\mu}, \bsz_i \rangle\left(\sigma_i-\sigma_{i-1}\right)+\left\langle\boldsymbol{\mu}, \bsz_i\right\rangle\left(\sigma_i-\sigma_{i+1}\right)+ \frac{\rho}{2} \left\|\bsz_i-\bsz_{i+1}\right\|^2 
\end{equation}
where $\bsw_{t+1}=\bsw_1, \bsx_{t+1}=\bsx_1, \bsz_{t+1}=\bsz_1$ and $\sigma_{t+1}=\sigma_1$. Furthermore, it holds that
\begin{equation}
    \begin{gathered}\label{eq:lambda_t3_lower}
        \text{Right-hand side of  \eqref{eq:main eq2} } \geq \sum_{i=1}^t \frac{\lambda}{2}\|\bsx_i-\bsx_{i+1}\|^2=  \frac{\lambda}{2}\bsx^\top (\bsL^{C_t} \otimes \bsI_d) \bsx \; \text{  if and only if} \\
    \frac{\rho-\lambda}{2}\sum_{i=1}^t\left\|z_i-z_{i+1}+ \boldsymbol{\mu} \sigma_i-\frac{1-\lambda}{\rho-\lambda} \boldsymbol{\mu} \sigma_{i+1}\right\|^2 \geq \frac{\rho-\lambda}{2}\| \boldsymbol{\mu}\|^2\sum_{i=1}^t\left(\sigma_i-\frac{1-\lambda}{\rho-\lambda} \sigma_{i+1}\right)^2-\frac{1-\lambda}{2}\| \boldsymbol{\mu}\|^2 \sum_{i=1}^t\left(\sigma_i-\sigma_{i+1}\right)^2.
\end{gathered}
\end{equation}
This can be shown similar to \eqref{eqn:bd_A1}.
Therefore, on the event $\mathcal{A}_2$, \eqref{eq:lambda_t3_lower} holds and we can obtain that 
\begin{equation}\label{eq:t>3 second}
    \begin{aligned}
        \P(\mathcal{F}_t)&\leq \P(\mathcal{F}_t\cap\mathcal{A}_2 )+\P(\mathcal{A}^c_2) \\
        &\leq \P\left( \sqrt{1-\rho^2} \sum_{i=1}^t \left\langle \bsx_i-\bsx_{i+1}, \bsw_{i+1}\right\rangle \geq   \frac{\lambda}{2}\bsx^\top (\bsL^{C_t} \otimes \bsI_d) \bsx \right)+\P(\mathcal{A}^c_2)\\
        &=\E_{\bsx_1,\ldots, \bsx_t} \P_{u \sim \mathcal{N}\left(0,\left(1-\rho^2\right)\bsx^\top (\bsL^{C_t} \otimes \bsI_d) \bsx\right)}\left(u \geq  \frac{\lambda}{2}\bsx^\top (\bsL^{C_t} \otimes \bsI_d )\bsx\right)+\P(\mathcal{A}^c_2)\\
        &=\E_{\bsx_1,\ldots, \bsx_t} \P_{u \sim \mathcal{N}\left(0,1\right)}\left(u \geq \sqrt{\frac{\lambda^2 \bsx^\top (\bsL^{C_t} \otimes \bsI_d )\bsx}{4(1-\rho^2)}} \right)+\P(\mathcal{A}^c_2)\\
        &\stackrel{(a)}{\leq} \E_{\bsx_1,\ldots, \bsx_t} \exp \left(-\frac{\lambda^2}{8(1-\rho^2) }\bsx^\top (\bsL^{C_t} \otimes \bsI_d )\bsx\right)+\P(\mathcal{A}^c_2)\\
        &\stackrel{(b)}{= } \det \left(\bsI_{dt}+\frac{\lambda^2}{4(1-\rho^2)}\bsL^{C_t} \otimes \bsI_d \right)^{-1/2}\exp \left(-\frac{1}{2} \bsv^\top \left[\bsI_{dt}-\left[\bsI_{dt}+\frac{\lambda^2}{4\left(1-\rho^2\right)} \bsL^{C_t} \otimes \bsI_d\right]^{-1}\right] \bsv\right)+\P(\mathcal{A}^c_2)\\
        &\stackrel{(c)}{\leq} \det \left(\bsI_t+\frac{\lambda^2}{4(1-\rho^2)}\bsL^{C_t}\right)^{-d/2}+\P(\mathcal{A}^c_2)\\
        &\stackrel{(d)}{=}\left(\prod_{j=0}^{t-1}\left\{1+\frac{\lambda^2}{2(1-\rho^2)}\left(1-\cos \left(\frac{2 \pi j}{t}\right)\right)\right\}\right)^{-d / 2}+\P(\mathcal{A}^c_2)\\
        &= \exp\left(-\frac{d}{2}S\left(\frac{1-\rho^2}{\lambda^2},t\right)\right)+\P(\mathcal{A}^c_2),
    \end{aligned}
\end{equation}
where $\bsv$ is the concatenation of $\boldsymbol{\mu} \sigma_i$ for $i=1,\dots, t$. 
The inequality $(a)$ holds by Lemma \ref{lem:normal tail}, the equality $(b)$ holds by Lemma \ref{lem:MGF QF}, the inequality $(c) $ holds by confirming through Lemma \ref{lem:eigen LCT} that the eigenvalue of $\bsI_{dt}-\left[\bsI_{dt}+\frac{\lambda^2}{4\left(1-\rho^2\right)} \bsL^{C_t} \otimes \bsI_d\right]^{-1}$ is non-negative, thereby establishing that $\bsI_{dt}-\left[\bsI_{dt}+\frac{\lambda^2}{4\left(1-\rho^2\right)} \bsL^{C_t} \otimes \bsI_d\right]^{-1}$ is positive semi-definite, the inequality holds by Lemma \ref{lem:eigen LCT}, and the last equality holds by the definition  of $S(\alpha,t)$ in \eqref{eq:def S}. 
Thus, \eqref{eq:upperbound F_t2} holds.

\section{Technical Tools}

\begin{lemma}[Proposition 2.2 in \cite{KN22}]\label{lem:I value}
For all $\alpha>0$
\begin{equation}
   I\left(\alpha\right)=2 \log \left(\frac{1+\sqrt{1+\alpha^{-1}}}{2}\right) .
\end{equation}
\end{lemma}

\begin{lemma}[Lemma 2.2 in \cite{KN22}]\label{lem:S<tI}
    For $t \geq 2$ and $\alpha>0, S\left(\alpha, t\right)<t I\left(\alpha\right)$.
\end{lemma}

\begin{lemma}[Moment generating function of quadratic form]\label{lem:MGF QF}
    Let $X \sim \mathcal{N}(\vec{\mu}, \Sigma)$ and $A$ is symmetric. We have
    \begin{equation}
        M_{X^{\top} A X}(t):=\E(e^{X^\top A X})=\frac{1}{|I-2 t A \Sigma|^{\frac{1}{2}}} e^{-\frac{1}{2} \mu^{\top}\left[I-(I-2 t A \Sigma)^{-1}\right] \Sigma^{-1} \mu}.
    \end{equation}
\end{lemma}

\begin{lemma}[Tail bound of normal distribution]\label{lem:normal tail}
    Let $X \sim \mathcal{N}(0,1)$. For $t>0$, we have
    \begin{equation}
        \P(X\geq t)\leq \exp(-t^2/2).
    \end{equation}
\end{lemma}

\begin{lemma}[Hoeffding's inequality]\label{lem:hoeff}
    Let $X_1, \ldots, X_n$ be independent random variables such that $a_i \leq X_i \leq b_i$ almost surely. Consider the sum of these random variables,
\begin{equation}
    S_n=X_1+\cdots+X_n
\end{equation}
Then it holds that , for all $t>0$
\begin{equation}
    \begin{gathered}
\P\left(S_n-\E\left[S_n\right] \geq t\right) \leq \exp \left(-\frac{2 t^2}{\sum_{i=1}^n\left(b_i-a_i\right)^2}\right) \\
\P\left(\left|S_n-\E\left[S_n\right]\right| \geq t\right) \leq 2 \exp \left(-\frac{2 t^2}{\sum_{i=1}^n\left(b_i-a_i\right)^2}\right).
\end{gathered}
\end{equation}
\end{lemma}

\begin{lemma}[Proposition 2.1 in \cite{KN22}]\label{lem:eigen LCT}
    Let $L^{C_t}$ be the Laplacian matrix for a cycle graph consisting of $t$ nodes. The eigenvalues of $L^{C_t}$ are $2\left(1-\cos \left(\frac{2 \pi k}{t}\right)\right)=4 \sin ^2\left(\frac{\pi k}{t}\right)$ for $k=0, \ldots, t-1$.
\end{lemma}

\begin{lemma}[Lemma 8.1 in \cite{noncentral} : Tail bound for noncentral chi-squared distribution]\label{lem:noncentral}
    Let $X$ be a noncentral $\chi^2$ variable with $D$ degrees of freedom and noncentrality parameter $B$, then for all $x>0$,
\begin{equation}
    \mathbb{P}[X \geq(D+B)+2 \sqrt{(D+2 B) x}+2 x] \leq \exp (-x)
\end{equation}
and
\begin{equation}
    \mathbb{P}[X \leq(D+B)-2 \sqrt{(D+2 B) x}] \leq \exp (-x).
\end{equation}
\end{lemma}

\section*{Acknowledgment} 
This work was supported by the National Research Foundation of Korea (NRF) grant funded by the Korea government (MSIT) (No. RS-2024-00408003, No. RS-2025-00516153 and RS-2024-00444862).

\bibliographystyle{IEEEtran}

\end{document}